\documentclass[11pt]{article}

\usepackage{pgf,tikz}
\usetikzlibrary{arrows}
\usepackage{array,multirow,makecell}
\setcellgapes{1pt} \makegapedcells
\usepackage{amssymb}
\usepackage{epstopdf}
\usepackage{graphicx}
\usepackage{here}
\usepackage{epsfig}
\usepackage{color}
\usepackage{amsthm}
\usepackage{amsmath}
\usepackage{subfig}
\usepackage{siunitx}
\usepackage{textcomp}
\usepackage{lmodern}

\usepackage[normalem]{ulem}
\usepackage{soul}

\usepackage{natbib}
\usepackage{hyperref}
\hypersetup{
    colorlinks=true,
    linkcolor=blue,
    filecolor=magenta,      
    urlcolor=cyan,
    citecolor=red
}

\newtheorem{thm}{\textbf{Theorem}}

\newtheorem{propo}{\textbf{Proposition}}

\newtheorem{rmq}{\textbf{Remark}}

\newcommand{\R}{\mathbb{R}}
\newcommand{\G}{\Tilde{G}}
\newcommand{\comment}[1]{}

\newcommand{\tcb}{\textcolor{blue}}

\usepackage{cleveref}

\begin{document}

\title{On the usefulness of a minimalistic model to study tree-grass biomass distributions along biogeographic gradients in the savanna biome}
\author{I.V. Yatat Djeumen$^1$\footnote{Corresponding Author: ivric.yatatdjeumen@up.ac.za}, Y. Dumont$^{2,3,1}$, A. Doizy$^{4,5}$, P. Couteron$^3$ \\
\small{$^1$University of Pretoria, Department of Mathematics and Applied Mathematics, Pretoria, South Africa} \\
\small{$^2$CIRAD, UMR AMAP, F-97410 St Pierre, Reunion island, France} \\
\small{$^3$AMAP, University of Montpellier, CIRAD, CNRS, INRAE, IRD, Montpellier, France,} \\
\small{$^4$CIRAD, UMR PVBMT, F-97410 St Pierre, Reunion island, France} \\
\small{$^5$DoAna - Statistiques R\'eunion, F-97480 Saint-Joseph, R\'eunion island, France}
 }
\maketitle

\begin{abstract}
We present and analyze a model aiming at recovering as dynamical outcomes of tree-grass interactions the wide range of vegetation physiognomies observable in the savanna biome along rainfall gradients at regional/continental scales. The model is based on two ordinary differential equations (ODE), for woody and grass biomass. It is parameterized from literature and retains mathematical tractability, since we restricted it to the main processes, notably tree-grass asymmetric interactions (either facilitative or competitive) and the grass-fire feedback.  We used a fully qualitative analysis to derive all possible long term dynamics and express them in a bifurcation diagram in relation to mean annual rainfall and fire frequency. We delineated domains of monostability (forest, grassland, savanna), of bistability (e.g. forest-grassland or forest-savanna) and even tristability. Notably, we highlighted regions in which two savanna equilibria may be jointly stable (possibly in addition to forest or grassland). We verified that common knowledge about decreasing woody biomass with increasing fire frequency is recovered for all levels of rainfall, contrary to previous attempts using analogous ODE frameworks. Thus, this framework appears able to render more realistic and diversified outcomes than often thought of. Our model can help figure out the ongoing dynamics of savanna vegetation in large territories for which local data are sparse or absent. 
To explore the bifurcation diagram with different combinations of the model parameters, we have developed a user-friendly R-Shiny application freely available at : \url{https://gitlab.com/cirad-apps/tree-grass}.

\end{abstract}




\section{Introduction}\label{intro}
Savannas, as broadly defined as systems where tree and grass coexist
(\citet{Scholes1997}), occupy about $20\%$ of the Earth land
surface and are observed in a large range of Mean Annual
Precipitation (MAP). In Africa, they particularly occur between 100
mm and 1500 mm (and sometimes more) of total mean annual
precipitation (\citet{Lehmann2011},  \citet{Baudena2013}), that
is along a precipitation gradient leading from dense tropical forest
to desert. There is widespread evidence that fire and water
availability are variables which can exert determinant roles in
mixed tree-grass systems (\citet{Scholes1997}, \citet{Yatat2018} and references therein).
Empirical studies showed that vegetation properties such as biomass, leaf
area, net primary production, maximal tree height and annual maximum
standing crop of grasses vary along gradients of precipitation
(\citet{PenningDjiteye1982}, \citet{Abbadie2006lamto}). It is
widely accepted that water availability directly limits woody
vegetation in the driest part of the rainfall gradient, see e.g. \citet{Sankaran2005determinants}. Along the rest of this gradient, rainfall is
known to influence indirectly the fire regime through what can be
referred to as the grass-fire feedback
(\citet{Yatat2018}, \citet{Scholes2003convex} and references therein): grass biomass that grows during
rainfall periods is fuel for fires occurring in the dry months.
Sufficiently frequent and intense fires are known to prevent or at
least delay the development of woody vegetation
(\citet{Yatat2018},
\citet{Govender2006}), thereby preventing trees and shrubs to
depress grass production through competition for light and nutrients. The grass-fire feedback is
widely acknowledged in literature as a force able to counteract the
asymmetric competition of trees onto grasses, at least for climatic
conditions within the savanna biome that enables sufficient grass
production during wet months.

\par

Dynamical processes underlying savanna vegetation have been the
subject of many models. Some of them explicitly considered the
influence of soil water resource on the respective productions of
grass and woody vegetation components
(see the review of
\citet{Yatat2018}). Most of the models also incorporated the
grass-fire positive feedback, several of them distinguishing
fire-sensitive small trees and shrubs from non-sensitive large trees (\citet{Higgins2000fire},
\citet{Beckage2009}, \citet{Baudena2010},
\citet{Staver2011tree}, \citet{Yatat2014, Yatat2018}), while the rest stuck to the simplest formalism featuring just grass
and tree state variables (\citet{vanLangevelde2003},
\citet{DOdorico2006probabilistic}, \citet{Higgins2010stability},
\citet{Accatino2010tree}, \citet{Beckage2011grass},
\citet{YuDOdoricco2014ecohydrological}, \citet{Tchuinte2014}, see also the review of \citet{Yatat2018}). Models featuring the grass-fire feedbacks have
shown that complex physiognomies displaying tree-grass coexistence (i.e.
savannas) may be stable (\citet{vanLangevelde2003},
\citet{DOdorico2006probabilistic}, \citet{Baudena2010},
\citet{Accatino2010tree}, \citet{Yatat2014},
\citet{Tchuinte2014}) as well as more ``trivial" equilibria such as
desert, dense forest or open grassland. Some models also predict
alternative stable physiognomies under similar rainfall conditions
(\citet{Accatino2010tree}, \citet{Staver2011tree},
 \citet{Tchuinte2014}, \citet{Yatat2014, Yatat2018}) while field observations report contrasted savanna-forest mosaics at landscape scale (see Figure \ref{forest_savanna_mosaic}). 
However, the ability to predict, along the whole rainfall gradient, all the physiognomies that are suggested by observations as possible stable
or multi-stable outcomes was not fully mastered and established. Indeed, most models focused on specific contexts or questions and often feature parameters difficult to assess over large territories, especially in Africa (\citet{Accatino2010tree},
\citet{Higgins2010stability}, \citet{Baudena2010},
\citet{DeMichele2011}, \citet{Beckage2011grass},
\citet{YuDOdoricco2014ecohydrological}). Nonetheless, the \citet{Accatino2010tree}'s attempt was a seminal step in that direction but with some notable imperfections (see below).

\begin{figure}[h!]
    \centering
    \subfloat[][]{\includegraphics[scale=0.5]{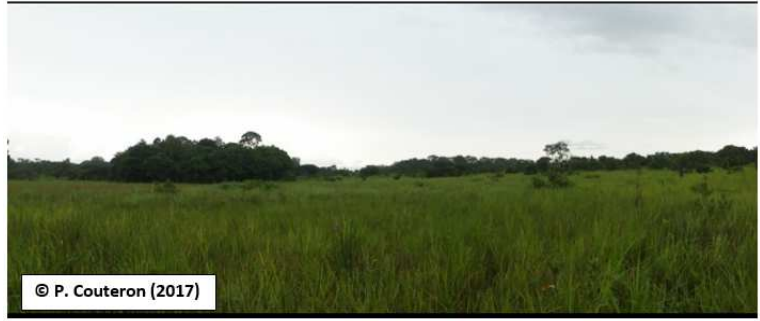}}
    \vspace{0.25cm}
    \subfloat[][]{\includegraphics[scale=0.43]{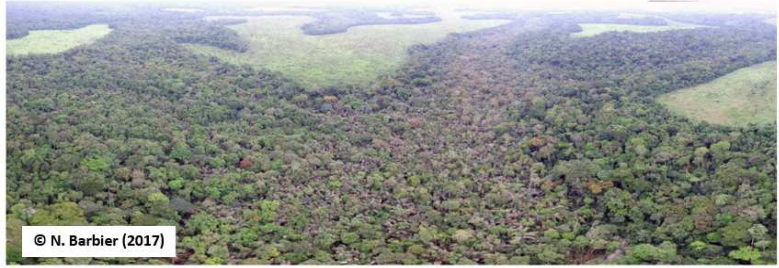}}
    \caption{{\scriptsize (a) Photo of forest--grassland boundary in
    Mpem \& Djim National Park, Central Cameroon. (b) An abrupt Forest--savanna (grassland) mosaics in Ayos, Cameroon.}}
    \label{forest_savanna_mosaic}
\end{figure}

The \citet{Accatino2010tree} model 
was pioneering in the sense
that it allowed these authors to provide a ``broad picture", by
delimiting stability domains for a variety of possible vegetation
equilibria as functions of gradients in rainfall and fire frequency.
This result was especially  interesting and the considered model
was sufficiently simple (two vegetation variables, i.e. grass and
tree covers) to provide analytical forecasts.  However, results from
\citet{Accatino2010tree} were questionable regarding the role of
fire return time. In fact, all over the rainfall gradient their
model predicted that increasing fire frequency would lead to an
increase in woody cover which contradicts empirical
knowledge on the subject. The features of the model that led to this
problem were barely debated in the ensuing publications. And more
recent papers instead either devised more complex models or shift to
stochastic modelling (see the review of \citet{Yatat2018}) that did not allow much
analytical exploration of their fundamental properties.

\par

 In this paper, we aim to account for a wide range of physiognomies and dynamical outcomes of the tree--grass interactions system at both regional and continental scales by
relying on a simple model that explicitly address some essential
processes that are: (i) limits put by rainfall on woody and grassy biomasses development, (ii) asymmetric interactions between woody and herbaceous
plant life forms, (iii) positive feedback between grass biomass and
fire intensity and, decreased fire impact with tree height.

\par

Starting from \citet{Yatat2018}, we explicitly express the growth of both
woody and herbaceous vegetation as functions of the mean annual rainfall, 
with the aim to study
model predictions in direct relation to rainfall and fire frequency gradients.
Through the present contribution we aim at extending and improving a framework for modelling
vegetation in the savanna biome through an ODE-based model, that is minimal (in terms of state variables and parameters), mathematically tractable and generic in the sense that its structure does not pertain to particular locations in the savanna biome. 

\par
 An idiosyncrasy of our minimalistic tree-grass
model is that we considered the fire-induced loss of woody
biomass by mean of two independent non-linear functions, namely
$\omega$ (see (\ref{omega_fction})) and $\vartheta$ (see
(\ref{theta_fction})). Introducing these two functions,
\citet{Tchuinte2017} showed that the previous model substantially
improve previously published results on tree-grass dynamical systems
(see also \citet{Yatat2018}). For example, they showed that
increasing fire return period systematically leads the system to
switch from grassland or savanna to forest (woody biomass build-up).
This result is entirely consistent with field observations
(\citet{Bond2005global}, \citet{Yatat2018} and references therein). From this sound basis, we introduced improvements in the model which are exposed in the present paper. Notably, we now let  influences of trees on grasses range from facilitation to competition according to climate.

The goal of the present paper is to show that the theoretical analysis of our minimalistic tree-grass
ODE model is able to provide, at  broad scales, an array of sensible predictions about possible vegetation
physiognomies that was not attained by tree--grass models of similar levels of complexity. Hence, relying only on qualitative results, we will construct a bifurcation diagram depicting the possible vegetation types along the rainfall vs. fire frequency gradients. Last but not least, in order to render our approach easy-to-use, we have developed a R-Shiny application to build the previous bifurcation diagram taking into account all the model parameters that can be changed easily according to the reader's wish.

\par

This paper is organized as follows. Section \ref{section2}  presents the ODE model. Section \ref{longtermbehavior} gives
the main theoretical results. Section \ref{AS} presents parameter ranges as well as results of the sensitivity analysis of the ODE model.   In section
\ref{section4}, the R-Shiny application is presented,  bifurcation diagrams in the rainfall-fire frequency space are given and numerical simulations are also provided to illustrate vegetation shifts in relation to rainfall and fire drivers as basis for discussions in section \ref{discussion}. Finally, in section \ref{conclusion}, we summarize the main results of this paper and how they can be improved or extended.

\section{The minimalistic ODE model formulation}\label{section2}

Our model features two coupled ordinary differential equations (eq. \eqref{swv_eq1} below) expressing the dynamics of tree and grass biomasses. Each equation entails a term of logistic growth (with parameters depending on MAP, section \ref{croissance-water}) and terms of biomass suppression by external agents (e.g. grazers or browsers) and fire. Coupling of the equations occurs because fire intensity experienced by woody biomass is a non-linear increasing function of grass biomass (see section \ref{fire-section}), while the grass biomass dynamics is asymmetrically influenced by woody biomass (see section \ref{assymmetrie}). The model presented here is built on a previous ODE framework that models fire-induced mortality on woody biomass by mean of two independent non-linear functions, namely $\omega$ (see \eqref{omega_fction}) and
$\vartheta$ (see \citet{Tchuinte2017}, \citet{Yatat2018}). The present contribution improves it by allowing both facilitative and competitive effects of trees on grasses. We thus take into account the fire-mediated negative feedback of grasses onto trees and the negative (in the case of competition) or positive (in the case of facilitation) feedback of grown-up trees on grasses. 

\subsection{Grass and tree biomass growths along the rainfall gradient}\label{croissance-water}
\subsubsection{Annual growths}

We assume that the annual productions of grasses and trees are non-linear and
saturating functions of MAP. Following
\citet{vandeKoppel1997}, \citet{Higgins2010stability} and
\citet{Nes2014tipping}, a Monod equation is judged adequate
to describe how limiting water resource modulates the maximal growth of both life forms (e.g.,  \citet{Whittaker1975},
see also \citet[Figure 4.6.3, page 191]{PenningDjiteye1982}).  We
assume that
$\displaystyle\frac{\gamma_{G}\textbf{W}}{b_{G}+\textbf{W}}$ and
$\displaystyle\frac{\gamma_{T}\textbf{W}}{b_{T}+\textbf{W}}$ are
annual biomass productions of grass and trees respectively, where
$\gamma_{G}$  and $\gamma_{T}$ (in yr$^{-1}$) express maximal
growths of grass and tree biomasses respectively. Half saturations
$b_{G}$ and $b_{T}$ (in mm.yr$^{-1}$) determine how quickly growth
increases with water availability.

\par

\citet{Accatino2010tree}
considered that vegetation growths are linear functions of soil
moisture, however, the nonlinear relationship between soil-water and
biomass production is widely observed in the field
(\citet{Mordelet1993influence}, \citet{Yatat2018} and references therein) as soon as the most favourable part of the rainfall gradient is taken into account.

\subsubsection{Carrying capacities}

We further assume that carrying capacities  of grass $K_{G}(\textbf{W})$ and
tree $K_{T}(\textbf{W})$ are  increasing and bounded functions of
water availability $\textbf{W}$. There are empirical field data sets (e.g.
\citet{UNESCO1981}, \citet{Sankaran2005determinants} and references therein) which expressed how maximum standing tree biomass increases with rainfall. Some more studies have
dealt with tree cover in relation to MAP at a continental or regional scale (see e.g.,  \citet{BuciniHanan2007}
and Figure 2 (a) in \citet{Favier2012abrupt} that observed increasing and saturating curves).  
To determine $K_T$, we combined field plot data reported in
\citet{Higgins2010stability} for the savanna side and \citet{Lewis2013AGBspatial} for the forest side (see also Figure \ref{swv_fig1}). To fit the data, we used the following function
$K_{T}(\textbf{W})=\dfrac{c_T}{1+d_{T}e^{-a_{T}\textbf{W}}}$, where
$c_T$ (in t.ha$^{-1}$) stands for the maximum value of the tree biomass
carrying capacity, $a_{T}$ (mm$^{-1}$yr) controls the steepness of
the curve, and $d_{T}$ controls the location of the inflection
point.  
We used the nonlinear quantile
regression (\citet{KoenkerPark1996}), as implemented in the ``quantreg" library of the R software \cite{R}. According to the 0.75$^{th}$ quantile regression (Figure \ref{swv_fig1} left, blue curve), we found $c_T=498.6$ t.ha$^{-1}$, $d_{T}=106.7$, and $a_{T}=0.0045$ mm$^{-1}$yr.

Concerning the  grass biomass standing crop, $K_G$, we used
empirical field data from  \citet{Braun1972a, Braun1972b},   \citet{MenautCesar1979} and
\citet{Abbadie2006lamto}.  We consider  the following function:
$K_{G}(\textbf{W})=\dfrac{c_G}{1+d_{G}e^{-a_{G}\textbf{W}}}$, where
$c_G$ (in t.ha$^{-1}$) denotes the maximum value of the grass
biomass carrying capacity,  $a_{G}$ (mm$^{-1}$yr) controls the
steepness of the curve, and $d_{G}$ controls the location of
the inflection point. We reached the following values: $c_G=17.06$
t.ha$^{-1}$, $d_{G}=14.73$, and $a_{G}=0.0029$ mm$^{-1}$yr for the 0.75$^{th}$ quantile regression (Figure \ref{swv_fig1} right, blue curve).

\begin{figure}[H]
	\centering
	\subfloat[][]{\includegraphics[scale=0.5]{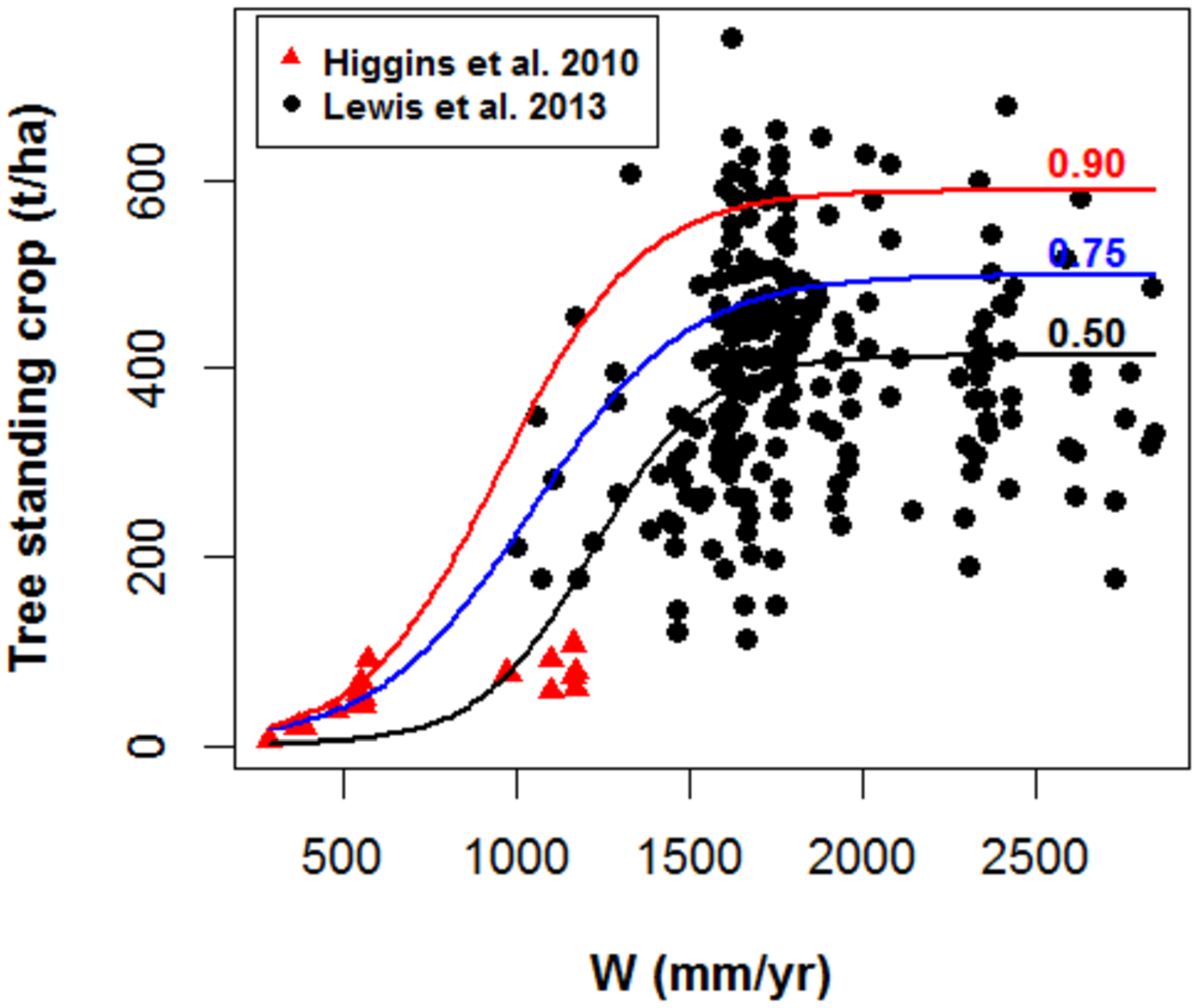}}
	\subfloat[][]{\includegraphics[scale=0.5]{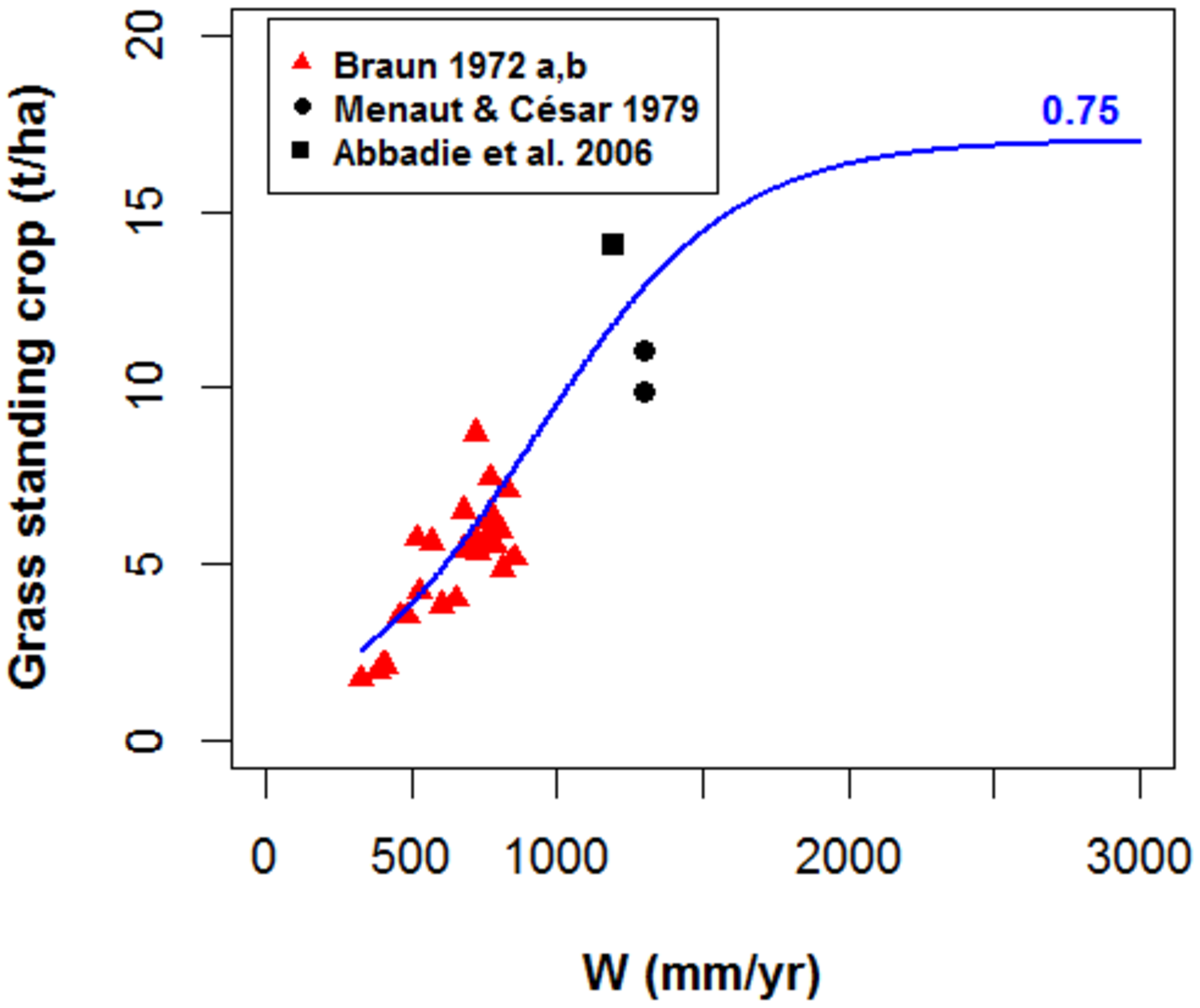}}
	\caption{{\scriptsize (a) Maximum standing tree biomass $K_{T}$  versus Mean Annual Rainfall.
			Data are drawn from figures in \citet{Higgins2010stability} and  \citet{Lewis2013AGBspatial}. Solid blue, red and black curves represent the quantile regression fits for 0.75$^{th}$, 0.9$^{th}$ and 0.5$^{th}$ quantiles, respectively. (b)
			Maximum grass biomass (standing crop) $K_{G}$ versus rainfall. Data are
			from  \citet{MenautCesar1979}, \citet{Braun1972a, Braun1972b} and  \citet{Abbadie2006lamto}.
			}}
	\label{swv_fig1}
\end{figure}

\subsection{Asymmetric tree-grass interactions}\label{assymmetrie}

Several studies, located
under different rainfall regimes, compared grass production under and outside a tree crown. The synthesis by Mordelet \& Le Roux (see \citet[page 156]{Abbadie2006lamto}) concluded that the relative production (within to outside crown) is a decreasing function of rainfall. This means that the impact of tree biomass on grass biomass ranges from possible facilitation, in arid and semi-arid parts of the rainfall gradient, to competition in the humid part with the tipping point located around a mean annual rainfall of \textit{ca.} 600 mm.yr$^{-1}$. However, despite 
empirical evidence possible facilitation has never been integrated in published tree-grass interactions models, even in those claiming genericity with respect to geographical location (see the review of \citet{Yatat2018}). In this contribution, we assume for the effect of tree biomass on grass biomass, a non-linear function of the mean annual rainfall, $\textbf{W}$ (in mm.yr$^{-1}$) named $\eta_{TG}(\textbf{W})$ (in (t.yr)$^{-1}$), that can take either negative values, meaning facilitation or positive values for competition. More specifically, 
\begin{equation}
   \eta_{TG}(\textbf{W})=a\times\tanh\left(\dfrac{\textbf{W}-b}{c}\right)+d
   \label{etaTG}
\end{equation}

where $b$ (in mm.yr$^{-1}$) controls the location of the inflection point, $c$ (in mm.yr$^{-1}$) controls the steepness of the curve.  The parameter $a$ (resp. $d$) (in (t.yr)$^{-1}$) shapes the minimal  facilitation (resp. maximal competition) level.  After re-interpretation of \citet[page 156]{Abbadie2006lamto}, \citet{Yatat2016} found $-0.0412$ as the minimal facilitation value for $\eta_{TG}$ and, $0.0913$ for the maximal competition value. 

\subsection{Grass biomass, fire intensity and fire-induced mortality}\label{fire-section}

\subsubsection{Fire intensity}
In savanna ecology it is overwhelmingly admitted that dried-up grass biomass is the main factor controlling both fire intensity and spreading capacity. Since our model is non-spatial, we combined these two properties of fire in a single, increasing function of grass-biomass (actually 'fire momentum', though we termed it 'fire intensity' for simplicity), expressing that whence average herbaceous biomass is in its highest range, fires both display the highest intensity and affect all the landscape. Conversely, low grass biomass due to aridity, grazing or tree competition, will make fires of low intensity and/or unable to reach all locations in a given year thereby decreasing the actual average frequency. We thus assume that the fire intensity noted $\omega$ is an increasing and bounded function (in [0,1]) of the grass biomass given as follows:
    \begin{equation}
    \omega(G)=\dfrac{G^{2}}{G^{2}+\alpha^{2}},
    \label{omega_fction}
    \end{equation}
where, $G$ (in t.ha$^{-1}$) is the grass biomass,
$\alpha$ (in t.ha$^{-1}$) is the value taken by $G$ when fire intensity is half its
maximum. Reader is also referred to \citet[section IV-B]{Yatat2018} for a detailed discussion about possible shapes of $\omega(G)$.

\subsubsection{Fire-induced woody biomass mortality}
For a given level of $\omega(G)$,  fire-induced tree/shrub mortality, noted $\vartheta$ is assumed to be a decreasing,
 non-linear function of tree biomass. Indeed, fires affect differently large and small trees since fires with high
 intensity (flame length $> ca.$ 2m) cause greater mortality of shrubs and topkill of
 trees while fires of lower intensity  (flame length $< ca.$ 2m)  topkill only
 shrubs and subshrubs (\citet{Yatat2018} and references therein). It
 is evident that tree biomass and total height are linked by increasing
 relationships. Therefore, we expressed $\vartheta$ as follows (\citet{Tchuinte2017}):
    \begin{equation}
    \vartheta(T)=\lambda_{fT}^{min} + (\lambda_{fT}^{max}-\lambda_{fT}^{min})e^{-pT},
    \label{theta_fction}
    \end{equation}
where, $T$ (t.ha$^{-1}$) stands for tree
biomass, $\lambda_{fT}^{min}$ (in yr$^{-1}$) is minimal lost portion of
tree biomass due to fire in configurations with a very large tree biomass,
$\lambda_{fT}^{max}$ (in yr$^{-1}$) is maximal loss of tree/shrub
biomass due to fire in open vegetation (e.g. for an isolated woody
individual having its crown within the flame zone), $p$ (in
t$^{-1}$) is proportional to the inverse of biomass suffering an
intermediate level of mortality.

\subsubsection{Fire-induced grass biomass mortality}

Fire-induced grass mortality is assumed to explicitly depend on the mean annual precipitation, noted $\textbf{W}$, because in arid and semi-arid locations, grass growth is low or very low due to insufficient rainfall and there is generally no continuous grass layer. Consequently, even if a fire occurs, it can not propagate and its impact on grass layer is therefore very limited. Conversely, in the humid part of the rainfall gradient, the fire-induced grass mortality is more important because grass layer is continuous and fire propagates easily. We express the fire-induced grass mortality as follows
\begin{equation}
\lambda_{fG}(\textbf{W})=\lambda_{fG}^{min} + (\lambda_{fG}^{max}-\lambda_{fG}^{min})\dfrac{\textbf{W}^z}{\textbf{W}^z+S^z}.
\label{lambda_fction}
\end{equation}
The parameter $z$ controls the shape for the function  $\lambda_{fG}(\textbf{W})$ while the value of $S$ (in mm.yr$^{-1}$) corresponds to the tipping point that separates low values to high values of the function $\lambda_{fG}(\textbf{W})$ along the mean annual rainfall gradient. $\lambda_{fG}^{min}$ and $\lambda_{fG}^{max}$ control the bounds of $\lambda_{fG}(\textbf{W})$.

\subsection{Full system}

Our resulting
minimalistic model is given by the set of
nonlinear ODE
(\ref{swv_eq1}). 

\begin{equation}
\left\{
\begin{array}{l}
\displaystyle \frac{dG}{dt}=\displaystyle\frac{\gamma_{G}\textbf{W}}{b_{G}+\textbf{W}}G\left(1-\displaystyle\frac{G}{K_{G}(\textbf{W})}\right)-\delta_{G}G-\eta_{TG}(\textbf{W})TG -\lambda_{fG}(\textbf{W})fG,\\
\\
\displaystyle\frac{dT}{dt}=\displaystyle\frac{\gamma_{T}\textbf{W}}{b_{T}+\textbf{W}}T\left(1-\displaystyle\frac{T}{K_{T}(\textbf{W})}\right)-\delta_{T}T-f\vartheta(T)\omega(G)T,\\
\\
G(0)=G_{0}, T(0)=T_{0},
\end{array}
\right.
\label{swv_eq1}
\end{equation}

where, $G$ and $T$  (in t.ha$^{-1}$) stand for grass and tree
biomasses respectively; $\delta_{G}$ and $\delta_{T}$
express, respectively, the rates of grass and tree biomasses loss by
herbivores (termites, grazing and/or browsing) or by human action. In our modelling, 
the $f$ (in yr$^{-1}$) parameter is taken as
constant multiplier of $\omega(G)$, and we interpret it
as a man-induced ``targeted" fire frequency (as for
instance in a fire management plan), which will not automatically
translate into actual frequency of fires of notable
intensity (because of $\omega(G)$). With this interpretation, the actual fire
regime may substantially differ from the targeted
one, as frequently observed in the field  (see for instance \citet{Diouf2012} in southern Niger). We therefore distinguish fire frequency from fire intensity because grass biomass controls fire spread  (see e.g. \citet{Govender2006}, \citet{McNaughton1992}, \citet{Yatat2018} and references therein).

\section{Long-term behavior of system \eqref{swv_eq1}: main results of the qualitative analysis}\label{longtermbehavior}

Our approach has kept the model amenable to a complete qualitative analysis of equilibria and stability thereof, as developed in the appendices.  
Equilibria embodying the long-term behavior of system (\ref{swv_eq1}) are summarized in Tables \ref{swv_tab_2} and \ref{swv_tab_2bis} in the case of competitive and facilitative influences of trees on grasses, respectively. Tables \ref{swv_tab_2}-\ref{swv_tab_2bis} result from the theoretical analysis of system (\ref{swv_eq1}) provided in \ref{section3}. For reader convenience, we recall in the following some key findings from the appendices. Set the following functions and thresholds:

\begin{equation} \left\{
\begin{array}{l}
g_{G}(\textbf{W})=\displaystyle\frac{\gamma_{G}\textbf{W}}{b_{G}+\textbf{W}},\\
g_{T}(\textbf{W})=\displaystyle\frac{\gamma_{T}\textbf{W}}{b_{T}+\textbf{W}},
\end{array}
\right.
\label{swv_growths}
\end{equation}

\begin{equation}
\left\{
\begin{array}{l}
\mathcal{R}^{1}_{\textbf{W}}=\dfrac{g_{T}(\textbf{W})}{\delta_{T}},\\
\mathcal{R}^{2}_{\textbf{W}}=\dfrac{g_{G}(\textbf{W})}{\delta_{G}+\lambda_{fG}(\textbf{W})f}.
\end{array}
\right.
\label{swv_R0_desert}
\end{equation}

Irrespective of the effect of trees on grasses (i.e. facilitation or competition), system (\ref{swv_eq1}) always has the following trivial equilibria:

\begin{itemize}
	\item a bare soil equilibrium, i.e. desert, $\textbf{E}_{0}=(0,0)'$.
	\item a forest equilibrium $\textbf{E}_{F}=(0,T^{*})'$ which exists when $\mathcal{R}^{1}_{\textbf{W}}>1$.
	\item a grassland equilibrium $\textbf{E}_{G}=(G^{*},0)'$ which exists when $\mathcal{R}^{2}_{\textbf{W}}>1$,
\end{itemize}
with the following notation:
\begin{equation}
\left\{
\begin{array}{l}
T^{*}=K_{T}(\textbf{W})\left(1-\dfrac{1}{\mathcal{R}^{1}_{\textbf{W}}}\right),\\
G^{*}=K_{G}(\textbf{W})\left(1-\dfrac{1}{\mathcal{R}^{2}_{\textbf{W}}}\right).
\end{array}
\right.
\label{swv_T_G}
\end{equation}

The novelty in this paper is considering both possible competitive ($\eta_{TG}(\textbf{W})>0$) and facilitative ($\eta_{TG}(\textbf{W})<0$) influences of trees on grasses and carrying out the qualitative analysis for both cases (see Tables \ref{swv_tab_2}-\ref{swv_tab_2bis}, Proposition \ref{proposition-competition}, \ref{section3}) 
that shows that this induces a variety of behaviors for system (\ref{swv_eq1}). 
Precisely, qualitative analyses allow us to efficiently explore all parts of the parameter space by relying on well-defined thresholds that delineate all outcomes of our model. Notably, we show that contrary to the competition case that only admits monostability or multi-stability of equilibria, the facilitation case additionally admits periodic solutions in time (limit cycle, Theorem \ref{Poicare-bendixson} in \ref{section3}). We will not further elaborate this theoretical result in the main text since we did not observe it for the ranges of parameters we investigated. 

 A savanna equilibrium $\textbf{E}_{S}=(G_{*},T_{*})'$ of system (\ref{swv_eq1}) features coexistence of both trees and grasses, and satisfies
\begin{equation}
\left\{
\begin{array}{lcl}
g_{G}(\textbf{W})\left(1-\displaystyle\frac{G_{*}}{K_{G}(\textbf{W})}\right)-(\delta_{G}+\lambda_{fG}(\textbf{W})f)-\eta_{TG}(\textbf{W})T_{*}=0,\\
\\
g_{T}(\textbf{W})\left(1-\displaystyle\frac{T_{*}}{K_{T}(\textbf{W})}\right)-\delta_{T}-f\vartheta(T_{*})\omega(G_{*})=0.\\
\end{array}
\right.
\label{app_eq1-bis100}
\end{equation}

We first consider the case of competition of trees on grasses and then the case of facilitation. Hence, Proposition \ref{proposition-competition} 
holds true on the basis of Theorem \ref{al_thm1} in \ref{al_AppendixA}, page \pageref{al_AppendixA}. 

\begin{propo}\label{proposition-competition}
\begin{enumerate}
    \item \textbf{Competition case}.  Assume that $\eta_{TG}(\textbf{W})>0$. Then system (\ref{swv_eq1}) may admit zero, one, two, three or four savanna equilibria.
    \item \textbf{Facilitation case}.  Assume that $\eta_{TG}(\textbf{W})<0$. Then system (\ref{swv_eq1}) may admit zero, one, two, three, four or five savanna equilibria.
    \item \textbf{Neutral case}. Assume that $\eta_{TG}(\textbf{W})=0$. Then system (\ref{swv_eq1}) may admit zero, one or two savanna equilibria.
\end{enumerate}
\end{propo}

 We also set
\begin{equation}
\mathcal{Q}_{F}=\dfrac{g_{G}(\textbf{W})-\eta_{TG}(\textbf{W})T^{*}}{\delta_{G}+\lambda_{fG}(\textbf{W})f},\hspace{0.35cm}
\mathcal{R}_{F}=\dfrac{g_{G}(\textbf{W})}{\eta_{TG}(\textbf{W})T^{*}+\delta_{G}+\lambda_{fG}(\textbf{W})f}
\hspace{0.35cm}\mbox{and}\hspace{0.35cm}
\mathcal{R}_{G}=\dfrac{g_{T}(\textbf{W})}{\delta_{T}+\lambda_{fT}^{max}f\omega(G^{*})}.
\label{swv_thresholds_F_G}
\end{equation}

Below, we give an approximated interpretation of the aforementioned thresholds. The aim is to favor an intuitive ecological understanding of our theoretical results in Tables \ref{swv_tab_2}-\ref{swv_tab_2bis}.
\begin{itemize}
	\item[(i)] $\mathcal{R}^{1}_{\textbf{W}}=\dfrac{g_{T}(\textbf{W})}{\delta_{T}}$: reflects the primary production of tree biomass relative to tree biomass loss by herbivory (termites, browsing) or human action.
	\item[(ii)] $\mathcal{R}^{2}_{\textbf{W}}=\dfrac{g_{G}(\textbf{W})}{\delta_{G}+\lambda_{fG}(\textbf{W})f}$: represents the primary production of grass biomass relative to fire-induced biomass loss and additional loss due to herbivory (termites, grazing) or human action.
	\item[(iii)] $\mathcal{R}_{F}=\dfrac{g_{G}(\textbf{W})}{\eta_{TG}(\textbf{W})T^{*}+\delta_{G}+\lambda_{fG}(\textbf{W})f}$: denotes the primary production of grass biomass,
	relative to grass biomass loss induced by fire, herbivory (grazing) or human action and to additional grass suppression due to tree competition, at the close forest equilibrium. $\mathcal{R}_{F}$ is defined when $\eta_{TG}(\textbf{W})\geq0$. 
	\item[(iv)] $\mathcal{Q}_{F}=\dfrac{g_{G}(\textbf{W})-\eta_{TG}(\textbf{W})T^{*}}{\delta_{G}+\lambda_{fG}(\textbf{W})f}$: denotes the primary production of grass biomass and the additional grass production due to tree facilitation, at the close forest equilibrium, relative to fire-induced grass biomass loss and additional grass suppression due to herbivory (grazing) or human action. $\mathcal{Q}_{F}$ is considered when $\eta_{TG}(\textbf{W})\leq0$. The larger $\mathcal{R}_{F}$ or $\mathcal{Q}_{F}$, the higher the potential of grass, experiencing competition or facilitation, to maintain at a coexistence state characterized by $T^*$.  
	\item[(v)] $\mathcal{R}_{G}=\dfrac{g_{T}(\textbf{W})}{\delta_{T}+\lambda_{fT}^{max}f\omega(G^{*})}$: is the primary production of tree biomass relative  to fire-induced biomass loss at the grassland equilibrium and additional loss due to herbivory (browsing) or human action. The larger $\mathcal{R}_{G}$, the higher the potential of tree growth to compensate biomass losses at a coexistence state characterized by $G^*$. 
\end{itemize}
\label{thresholds_ecolo_meaning_1}

The long-term behavior of system (\ref{swv_eq1}), in the case of tree vs. grass competition, is entirely determined by the previous thresholds. It is summarized in Table \ref{swv_tab_2} where more than one savanna equilibrium could simultaneously exist and be stable (as per symbol `$\dagger$', at least one savanna equilibrium and at most four). Conditions for the existence of savanna equilibria, in the competition case, are summarized in Table \ref{swv_tab_1}. Thresholds $\mathcal{R}_{*}^{1}$, $\mathcal{R}_{*}^{2}$ and $\mathcal{Q}_{*}^{2}$, related to the asymptotic stability of savanna equilibria, when they exist, are defined in \eqref{swv_thresholds_S}, page \pageref{swv_thresholds_S}.

\begin{table}[H]
        \begin{center}
            \renewcommand{\arraystretch}{1.2}
            \begin{tabular}{lccccccc}
                \cline{1-8}
                \multicolumn{5}{c}{\bf Thresholds} &  \multirow{2}{1.3cm}{\bf Stable} & \multirow{2}{1.3cm}{\bf Unstable} & \multirow{2}{0.7cm}{\bf Case}\\
                \cline{1-5}
                $\mathcal{R}^{1}_{\textbf{W}}$ ($\mathcal{R}^{2}_{\textbf{W}}$) &  $\mathcal{R}_{G}$ &  $\mathcal{R}_{F}$ & $\mathcal{R}_{*}^{1}$& $\mathcal{R}_{*}^{2}$ & & &  \\
                \hline
                $\leq1(\leq1)$ & ND & ND & ND & ND &$ \textbf{E}_{0}$  &  & $\textbf{I}$ \\
                \hline
                \multirow{7}{2cm}{$>1(>1)$} & $>1$ & $\leq1$ & \multirow{3}{0.25cm}{--} & \multirow{3}{0.65cm}{$<1$}  & $\textbf{E}_{F}$ & $\textbf{E}_{0}$, $\textbf{E}_{G}$, $\textbf{E}_{S}$ & $\textbf{II}$\\
                \cline{2-3} \cline{6-8} & $\leq1$ & $>1$ &  &  & $\textbf{E}_{G}$ & $\textbf{E}_{0}$, $\textbf{E}_{F}$, $\textbf{E}_{S}$ & $\textbf{III}$\\
                \cline{2-3} \cline{6-8} & $\leq1$ & $\leq1$ &  &  & $\textbf{E}_{G}$, $\textbf{E}_{F}$  & $\textbf{E}_{0}$, $\textbf{E}_{S}$  & $\textbf{IV}$\\
                \cline{2-8}
                &  $>1$ &$\leq1$ & \multirow{4}{0.65cm}{$<1$} & \multirow{4}{0.65cm}{$>1$} & $\textbf{E}_{F}$, $\textbf{E}_{S}$ & $\textbf{E}_{0}$, $\textbf{E}_{G}$ & $\textbf{V}^\dagger$\\
                \cline{2-3} \cline{6-8}
                &  $\leq1$ &$>1$ &  &  & $\textbf{E}_{G}$, $\textbf{E}_{S}$ & $\textbf{E}_{0}$, $\textbf{E}_{F}$ & $\textbf{VI}^\dagger$\\
                \cline{2-3} \cline{6-8}
                &  $>1$ &$>1$ &  &  &  $\textbf{E}_{S}$ & $\textbf{E}_{0}$, $\textbf{E}_{G}$, $\textbf{E}_{F}$ & $\textbf{VII}^\dagger$\\
                \cline{2-3} \cline{6-8}
                &  $\leq1$ &$\leq1$ &  &  &  $\textbf{E}_{F}$, $\textbf{E}_{G}$, $\textbf{E}_{S}$ & $\textbf{E}_{0}$  & $\textbf{VIII}^\dagger$\\
                \hline
            \end{tabular}
            \end{center}
\caption{Long-term dynamics of system
(\ref{swv_eq1}) when $\eta_{TG}(\textbf{W})\geq0$ (i.e. competition). `ND' stands for ``Not Defined" threshold. `$\dagger$' means that more than one savanna equilibrium (i.e. $\textbf{E}_{S}$) could be simultaneously stable. Precisely, at least one savanna equilibrium and at most four savanna equilibria could be simultaneously stable. 
}
\label{swv_tab_2}
\end{table}

Table \ref{swv_tab_2bis} summarizes the long-term behavior of system (\ref{swv_eq1}) in the case of tree vs. grass facilitation with possible existence of more than one savanna equilibrium. Precisely, at least one savanna equilibrium and at most five savanna equilibria could be simultaneously stable. See Table \ref{swv_tab_1bis1} for savanna equilibria existence conditions.

\begin{table}[H]
        \begin{center}
            \renewcommand{\arraystretch}{1.2}
            \begin{tabular}{lccccccc}
                \cline{1-8}
                \multicolumn{5}{c}{\bf Thresholds} &  \multirow{2}{1.3cm}{\bf Stable} & \multirow{2}{1.3cm}{\bf Unstable} & \multirow{2}{0.7cm}{\bf Case}\\
                \cline{1-5}
                $\mathcal{R}^{1}_{\textbf{W}}$ ($\mathcal{R}^{2}_{\textbf{W}}$) &  $\mathcal{R}_{G}$ &  $\mathcal{Q}_{F}$ & $\mathcal{R}_{*}^{1}$& $\mathcal{Q}_{*}^{2}$ & & &  \\
                \hline
                $\leq1(\leq1)$ & ND & ND & ND & ND &$ \textbf{E}_{0}$  &  & $\textbf{I}$ \\
                \hline
                \multirow{7}{2cm}{$>1(>1)$} & $>1$ & $\leq1$ & \multirow{3}{0.25cm}{--} & \multirow{3}{0.65cm}{$<1$}  & $\textbf{E}_{F}$ & $\textbf{E}_{0}$, $\textbf{E}_{G}$, $\textbf{E}_{S}$ & $\textbf{II}$\\
                \cline{2-3} \cline{6-8} & $\leq1$ & $>1$ &  &  & $\textbf{E}_{G}$ & $\textbf{E}_{0}$, $\textbf{E}_{F}$, $\textbf{E}_{S}$ & $\textbf{III}$\\
                \cline{2-3} \cline{6-8} & $\leq1$ & $\leq1$ &  &  & $\textbf{E}_{G}$, $\textbf{E}_{F}$  & $\textbf{E}_{0}$, $\textbf{E}_{S}$  & $\textbf{IV}$\\
                \cline{2-8}
                &  $>1$ &$\leq1$ & \multirow{4}{0.65cm}{$<1$} & \multirow{4}{0.65cm}{$>1$} & $\textbf{E}_{F}$, $\textbf{E}_{S}$ & $\textbf{E}_{0}$, $\textbf{E}_{G}$ & $\textbf{V}^\ddagger$\\
                \cline{2-3} \cline{6-8}
                &  $\leq1$ &$>1$ &  &  & $\textbf{E}_{G}$, $\textbf{E}_{S}$ & $\textbf{E}_{0}$, $\textbf{E}_{F}$ & $\textbf{VI}^\ddagger$\\
                \cline{2-3} \cline{6-8}
                &  $>1$ &$>1$ &  &  &  $\textbf{E}_{S}$ & $\textbf{E}_{0}$, $\textbf{E}_{G}$, $\textbf{E}_{F}$ & $\textbf{VII}^\ddagger$\\
                \cline{2-3} \cline{6-8}
                &  $\leq1$ &$\leq1$ &  &  &  $\textbf{E}_{F}$, $\textbf{E}_{G}$, $\textbf{E}_{S}$ & $\textbf{E}_{0}$  & $\textbf{VIII}^\ddagger$\\
                \cline{2-8}
                &  $>1$ &$>1$ & -- & $<1$ &  $\textbf{LC}$ & $\textbf{E}_{0}$,$\textbf{E}_{F}$, $\textbf{E}_{G}$, $\textbf{E}_{S}$  & $\textbf{IX}$\\
                \hline
            \end{tabular}
            \end{center}
\caption{Long-term dynamics of system
(\ref{swv_eq1}) when $\eta_{TG}(\textbf{W})<0$ (i.e. facilitation). The notation `$\ddagger$' means that more than one savanna equilibrium (i.e. $\textbf{E}_{S}$) could be simultaneously stable (at least one and at most five). `\textbf{LC}' stands for limit cycle that appears when all equilibria are unstable.}
\label{swv_tab_2bis}
\end{table}

\section{Parameter values and sensitivity analyses of model (\ref{swv_eq1})} \label{AS}

Interpretation of results from mathematical models of biological systems
is often complicated by the presence of uncertainties in
experimental data that are used to estimate parameter values
(\citet{Marino2008}). Moreover, some parameters are liable to vary in space, even in a given reference area. Sensitivity analysis (SA) is a method for measuring uncertainty in any type of complex model by identifying
critical inputs and quantifying how input uncertainty impacts model
outcomes. Different SA techniques exist (\citet{Marino2008} and
references therein). In this section we will perform partial rank
correlation coefficient (PRCC) and the extended Fourier amplitude
sensitivity test (eFAST) analysis in order to deal with both
cases of nonlinear but monotonic relationships between outputs and
inputs (i.e. PRCC) as well as nonlinear and non-monotonic trends
(eFAST). 

\par

The parameter ranges considered for this study are given in
Table \ref{intervalle-params}. Though the model aims to be qualitatively relevant for a large swath
of African situations, we particularly ground our choice of parameter values in a 
north-south gradient located at and around the $16$\textdegree E of longitude, and between $ca.$ $6$ and $10$\textdegree N of latitude (i.e., between $ca.$ $900$ to $1500$ mm.yr$^{-1}$ of MAP). This area goes from desert and the Sahel steppe in the north of lake Chad to the equatorial area in southern Cameroon and it spans the main vegetation physiognomies of Central Africa that include close canopy forest, grassland, savanna, forest-grassland and forest-savanna mosaics (see e.g. Figure \ref{forest_savanna_mosaic}).  Using longitude and latitude data, the MAP data were extracted from BIO12 (\url{http://www.worldclim.org/bioclim}, see also \citet{Hijmans2005}) using the ``raster" package of RStudio, version $1.1.383$ \cite{R}. 
 Retained parameter ranges originate from published literature (e.g. $f$: fire frequency, \textbf{W}: MAP), re-interpretations of empirical results (e.g. $\lambda_{fT}^{min}$: minimal lost portion of tree biomass due to fire in configurations with a very large tree biomass, $\lambda_{fT}^{max}$: maximal loss of tree/shrub biomass due to fire in open vegetation), expert-based knowledge (e.g. $\lambda_{fG}^{min}$: minimal fire-induced grass mortality, $\lambda_{fG}^{max}$: maximal fire-induced grass mortality) or by data fitting (e.g. $c_T$: maximum value of the tree biomass carrying capacity, $c_G$: maximum value of the grass biomass carrying capacity). It is to the best of our knowledge the first time that consistent responses curves (Figure \ref{swv_fig1}) are assessed from existing information all along the rainfall gradient. 

For the PRCC analysis (see Figure \ref{PRCC}), we used the PCC function (R software \cite{R}) and $1000$ bootstrap replicates, with a probability level of $0.95$ for (the bootstrap) confidence intervals. 
For the eFAST analysis (see Figure \ref{eFAST}), we used the FAST99 function (R software) with $7500$ runs. As expected, because of a large number of parameters (25), it took quite a long time.

eFast sensitivity analysis pointed towards the leading role of parameters relating to fire frequency, biomass growth $(\gamma_{G,T})$, biomass destruction $(\delta_{G,T})$. Logically, MAP (\textbf{W}) appears pervasive, especially for $T$. Maximal rate of grass suppression by fire is influential for both tree and grass biomass while maximal woody biomass suppression is not. For both variables, the $\alpha$ parameter, which is the critical grass biomass letting fire shift from low to high intensities (eq. \eqref{omega_fction}) appears of substantial influence (7th rank for both variables).

PRCC results provide some complementary insights.  Some parameters that tend to decrease grass biomass logically boost tree biomass and vice-versa, e.g. fire intensity, $\gamma_G$ vs. $\gamma_T$, $\delta_G$ vs. $\delta_T$. For both methods, MAP is of utmost importance for trees and fairly less for grass biomasses. Most of those parameters were already singled out by eFast but PRCC also underlined the roles of $p$ (tuning the decrease of fire impact with woody biomass, eq. \eqref{theta_fction}), $a_G$ and $\lambda_{fT}^{max}$. 
We may note that parameters related to equations (\ref{etaTG}) and (\ref{lambda_fction}) did not appear prominent in the sensitivity analysis, in spite of the important role that $\eta_{TG}$ (eq. (\ref{etaTG})) plays in the qualitative analysis.



\begin{table}[H]
{\footnotesize
\begin{center}
    \renewcommand{\arraystretch}{1.2}
\begin{tabular}{|c|c|c|c|c|}
  \hline
  Symbol & Unit & Baseline & Range & References \\
  \hline
  $c_T$ & t.ha$^{-1}$ & 430 & 423.8--523.4 & See text and Fig. \ref{swv_fig1}\\
  $a_T$ & yr$^{-1}$ & 0.004 & 0.0038--0.0054 & See text and Fig. \ref{swv_fig1}\\
  $d_T$ & -- & 107 & 78.26--167.34 & See text and Fig. \ref{swv_fig1}\\
  $c_G$ & t.ha$^{-1}$ & 20 & 12.3--21.82 & See text and Fig. \ref{swv_fig1}\\
  $a_G$ & yr$^{-1}$ & 0.0029 & 0.0023--0.0042 & See text and Fig. \ref{swv_fig1}\\
  $d_G$ & -- & 14.73 & 11.36--24.05 & See text and Fig. \ref{swv_fig1}\\
  $\gamma_T$ & yr$^{-1}$ & 1.5 & 1--3 & Estimated by revisiting \\
  &    &   &   & \citet{Stape2010}; \citet{Laclau2010}; \\
  & &   &   & \citet{Karmacharya1992}\\
  $b_T$ & mm.yr$^{-1}$ & 1100 & 900--1300 & \citet{Abbadie2006lamto} \\
  $\gamma_G$ & yr$^{-1}$ & 2.7 & 0.5--3.5 & \citet{Mordelet1995} \\
  $b_G$ & mm.yr$^{-1}$ & 500 & 400--650 & \citet{UNESCO1981} \\
  $\delta_T$ & yr$^{-1}$ & 0.1 & 0.015--0.3 & \citet{Hochberg1994influences};\\
  &   &   &   & \citet{Accatino2010tree} \\
  $\delta_G$ & yr$^{-1}$ & 0.1 & 0--0.6 & \citet{vanLangevelde2003} \\
  $\lambda_{fG}^{max}$ & -- & 0.4 & 0.2--0.7 & Expert-based value\\
  $\lambda_{fG}^{min}$ & -- & 0.005 & 0--0.1 & Expert-based value \\
  $S$ & mm.yr$^{-1}$ & 900 & 750-1100 & Expert-based value\\
  $z$ & -- & 8 & -- & Expert-based value\\
  $\lambda_{fT}^{min}$ & -- & 0.05 & 0--0.1 & Reinterpretation of\\ &   &   &   & \citet{Trollope2010};\\
    &   &   &   & see also \citet{Higgins2007topkill} \\
  $\lambda_{fT}^{max}$ & -- & 0.65 & 0.5--1 & Reinterpretation of\\
  &   &   &   &\citet{Trollope2010};\\
    &   &   &   & see also \citet{Higgins2007topkill} \\
  $p$ & t$^{-1}$ & 0.01 & 0.01--0.15 & Reinterpretation of\\ &   &   &   &\citet{Trollope2010} \\
  $\alpha$ & t.ha$^{-1}$ & 1 & 0.5--2.5 & \citet{Govender2006} \\
  $b$ & mm.yr$^{-1}$ & 600 & 500--700 &  Reinterpretation of\\
  &   &   &   & \citet{Mordelet1995}; \\
  &    &   &   & see also \citet{Abbadie2006lamto} \\
  $c$ & mm.yr$^{-1}$ & 120 & 75--150 & Assumed \\
  $a$ & (t.yr)$^{-1}$ & 0.01 & 0.001--0.01 & Reinterpretation of\\
  &   &   &   & \citet{Mordelet1995}; \\
  &    &   &   & see also \citet{Abbadie2006lamto} \\
  $d$ & (t.yr)$^{-1}$ & 0.0045 & 0.001--0.01 & Reinterpretation of\\
  &   &   &   & \citet{Mordelet1995}; \\
  &    &   &   & see also \citet{Abbadie2006lamto} \\
  $\textbf{W}$ & mm.yr$^{-1}$ & 1300 & 0--2000 & \citet{Menaut1991biomass}; \citet{Lewis2013AGBspatial} \\
  $f$ & yr$^{-1}$ & 1 & 0--2 & \citet{Higgins2010stability}; \citet{Accatino2010tree} \\
   \hline
\end{tabular}
    \caption{ {\small Parameter ranges. 
    }}\label{intervalle-params}
\end{center}}
\end{table}

\begin{figure}[H]
        \centering
        \subfloat[][\scriptsize LHS-PRCC sensitivity analysis when the reference output is the Grass biomass, $G$.]{\includegraphics[scale=0.75]{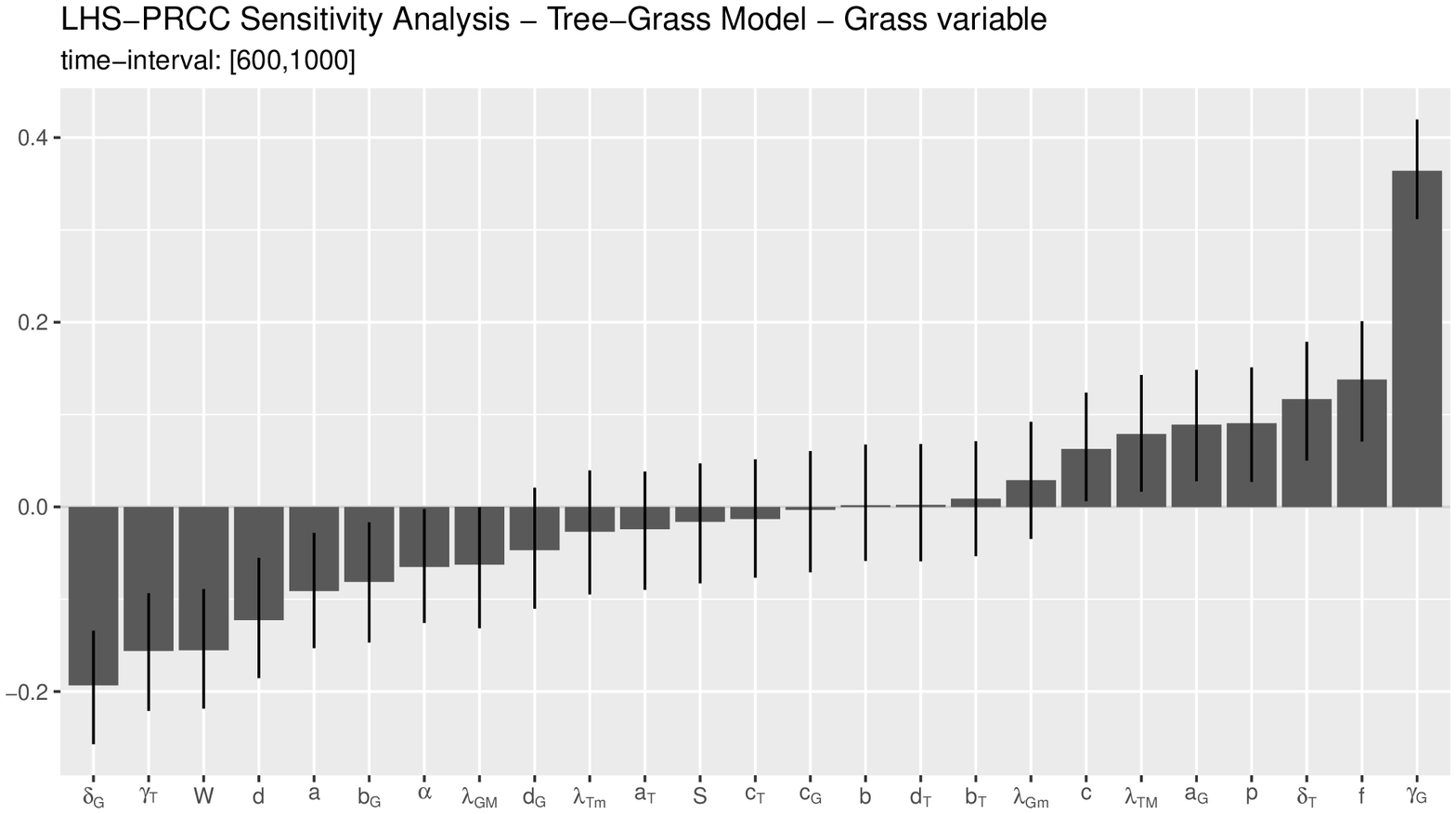}}\\
        \subfloat[][\scriptsize LHS-PRCC sensitivity analysis when the reference output is the Tree biomass, $T$.]{\includegraphics[scale=0.75]{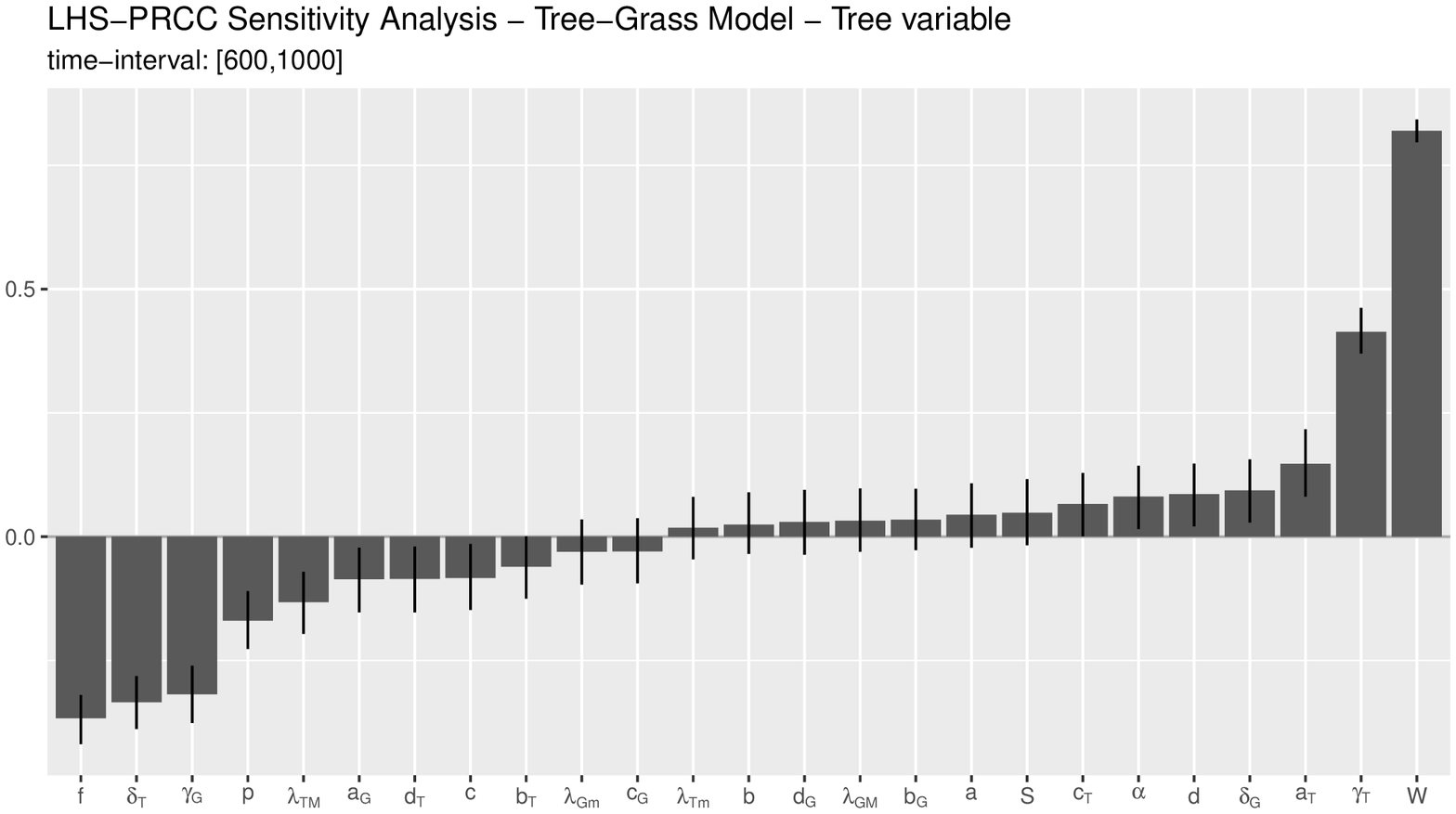}}
        \caption{{\scriptsize  LHS- PRCC Sensitivity Analysis. For simplicity, $\lambda_{GM}:= \lambda_{fG}^{max}$, $\lambda_{Gm}:= \lambda_{fG}^{min}$, $\lambda_{TM}:= \lambda_{fT}^{max}$ and $\lambda_{Tm}:= \lambda_{fT}^{min}$.}}
        \label{PRCC}
\end{figure}

\begin{figure}[H]
        \centering
        \subfloat[][eFAST sensitivity analysis when the reference output is the  Grass biomass, $G$]{\includegraphics[scale=0.75]{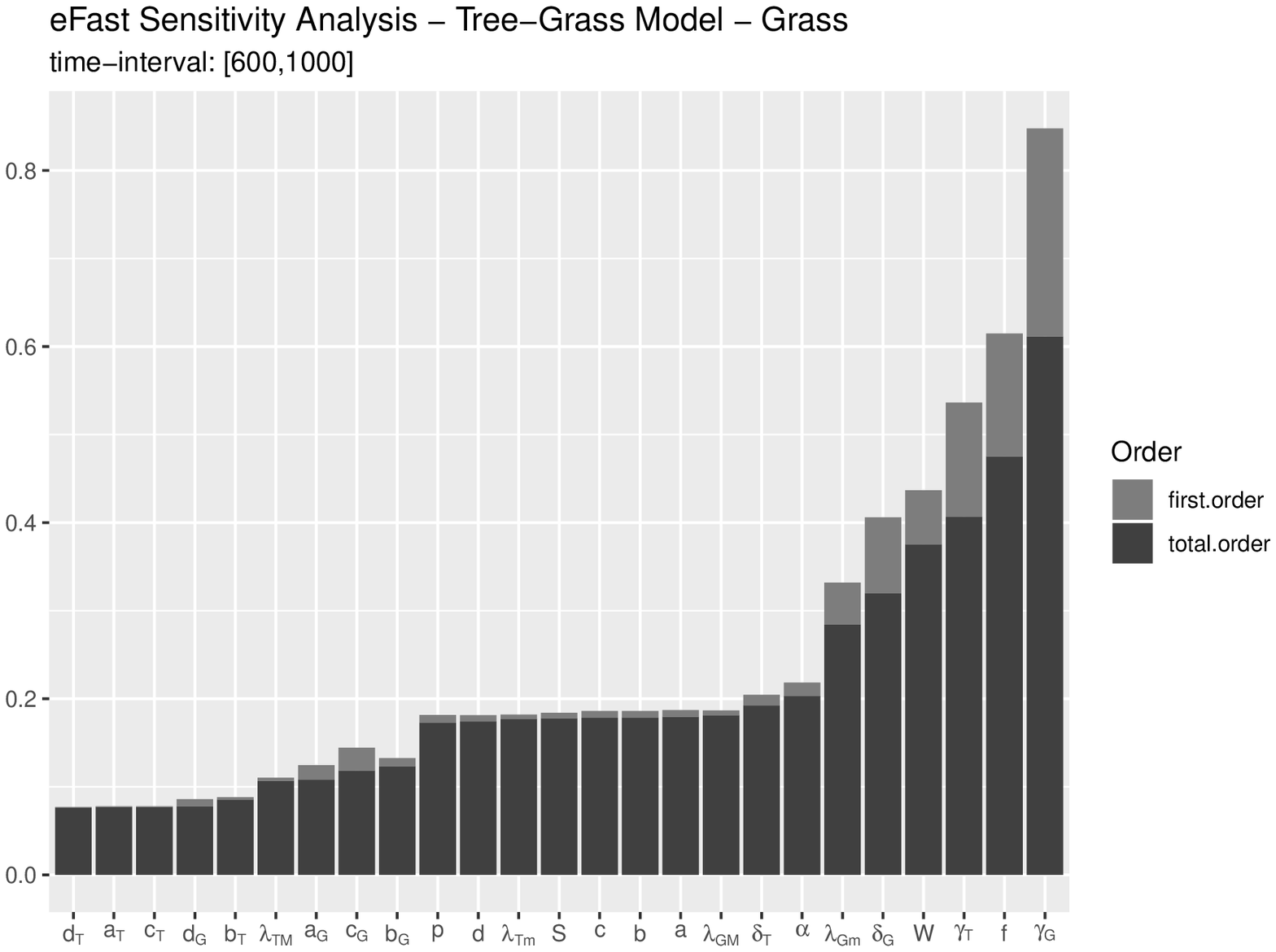}}\\
        \subfloat[][eFAST sensitivity analysis when the reference output is the Tree Biomass, $T$]{\includegraphics[scale=0.75]{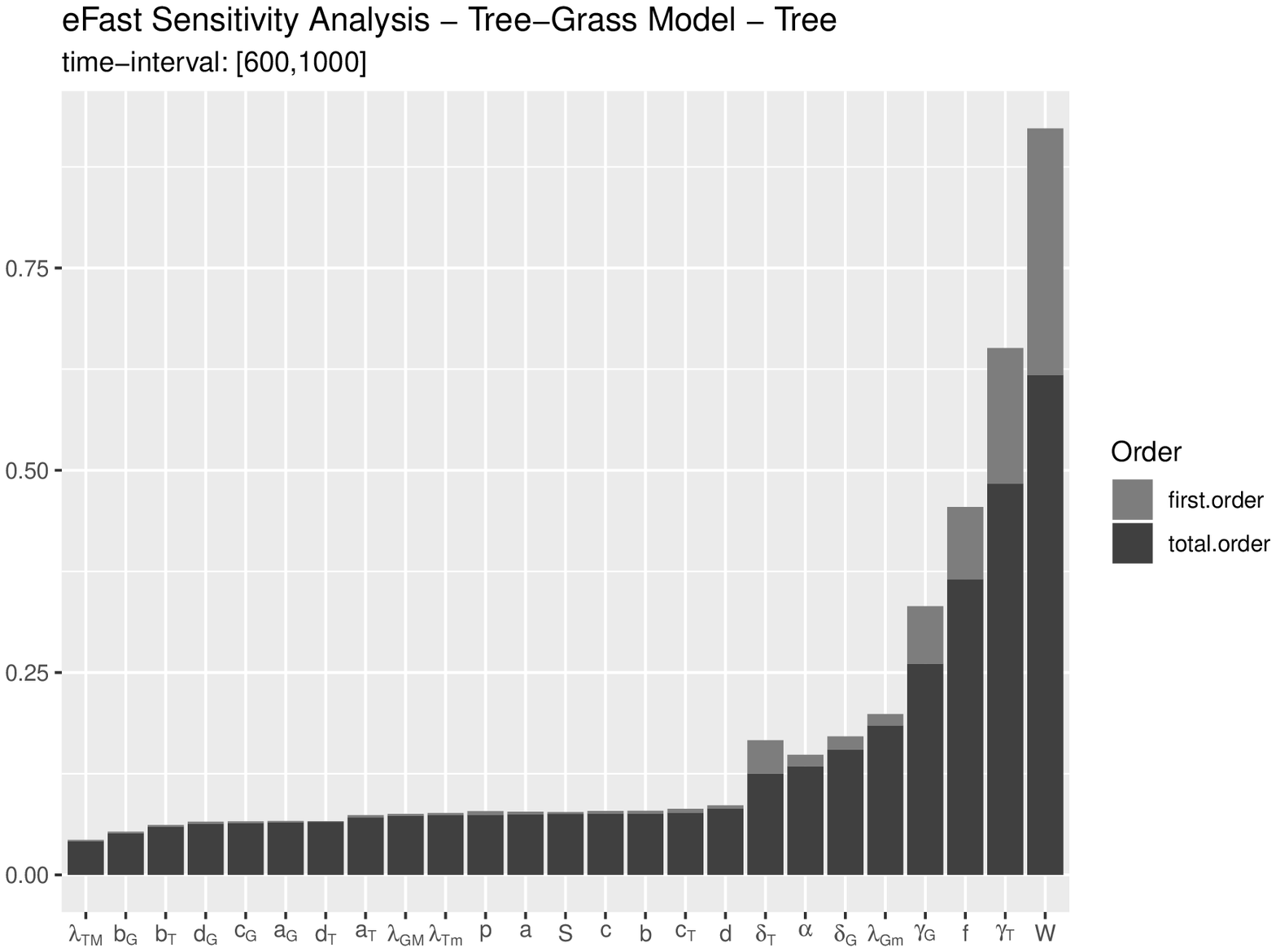}}
        \caption{{\scriptsize  e-FAST Sensitivity Analysis where $\lambda_{GM}:= \lambda_{fG}^{max}$, $\lambda_{Gm}:= \lambda_{fG}^{min}$, $\lambda_{TM}:= \lambda_{fT}^{max}$ and $\lambda_{Tm}:= \lambda_{fT}^{min}$.}}
        \label{eFAST}
\end{figure}

\section{Bifurcation diagrams and numerical simulations} \label{section4}

We first provide bifurcation diagrams, based on the thresholds computation for the following set of parameters (see Table \ref{params_figR2_fig1}). We secondly present numerical simulations (also based on Table \ref{params_figR2_fig1} values) to illustrate bifurcations in relation to mean annual rainfall (\textbf{W}) and fire frequency ($f$).

\begin{table}[H]
    {\footnotesize
        \begin{center}
            \renewcommand{\arraystretch}{1.2}
            \begin{tabular}{cccccc}
                \hline
                $c_{G}$,  t.ha$^{-1}$   & $c_{T}$,  t.ha$^{-1}$  & $b_{G}$, mm.yr$^{-1}$ & $b_{T}$, mm.yr$^{-1}$  &  $a_{G}$, yr$^{-1}$ & $a_{T}$, yr$^{-1}$ \\
                $20$ & $430$  & $500$ & $1100$ & $0.0029$ & $0.004$ \\
                \hline
                $d_{G}$, $-$    & $d_{T}$, $-$ & $\gamma_{G}$, yr$^{-1}$ & $\gamma_{T}$, yr$^{-1}$ & $\delta_{G}$, yr$^{-1}$ & $\delta_{T}$, yr$^{-1}$\\
                $14.73$     & $107$ & $2.7$ & $1.5$ & $0.1$ & $0.1$\\
                \hline
                $S$, $-$ & $\lambda_{fT}^{min}$, $-$ & $\lambda_{fT}^{max}$, $-$ & $p$, t$^{-1}$ & $\alpha$, t.ha$^{-1}$ &$z$\\
                $900$ & $0.05$ & $0.65$ & $0.01$ & $2.45$ & $8$  \\
                \hline
              $\lambda_{fG}^{min}$, $-$  & $\lambda_{fG}^{max}$, $-$ & $a$, t$^{-1}$yr$^{-1}$ & $b$, mm.yr$^{-1}$ & $c$, mm.yr$^{-1}$ & $d$, t$^{-1}$yr$^{-1}$ \\
              $0.005$ & $0.4$ & $0.01$ & $600$ & $120$ & $0.0045$ \\
                \hline
            \end{tabular}
            \end{center}
        \caption{ {\small Parameter values considered for simulations.}}\label{params_figR2_fig1}}
\end{table}

Thanks to the qualitative analysis of system (\ref{swv_eq1}) (see \ref{section3}), any version of the bifurcation diagrams (see for instance Figures \ref{bif_diagramme}-\ref{zoom_bif_diagramme}), in terms of the
fire frequency and the MAP, summarize the
outcomes of the ODE model (\ref{swv_eq1}). These bifurcation diagrams are obtained without simulations: they are produced with a
simple web application, called Tree-Grass (see Figure \ref{code-F-S}),  developed using R \cite{R}, shiny R package \cite{shiny} and plotly R package \cite{plotly}. The source code of this application is free to use and is licensed under a
Creative Commons Attribution-NonCommercial-ShareAlike 4.0
International License
(\url{https://creativecommons.org/licenses/by-nc-sa/4.0/}).
It is available at \url{https://gitlab.com/cirad-apps/tree-grass}. 
The user can modify the default parameters which are classified in
primary and secondary parameters according to the sensitivity
analysis (see section \ref{AS}).
They can be changed either manually via the interface or by uploading a csv file containing some custom set of parameters.
The ``Calculate" button launches the computation of the bifurcation diagram, using all qualitative thresholds, for the chosen parameters. 
 The Tree-Grass application allows to export the obtained bifurcation diagram and also the underlying set of parameters.

\begin{figure}[H]
    \centering
    \includegraphics[width = 0.8\textwidth]{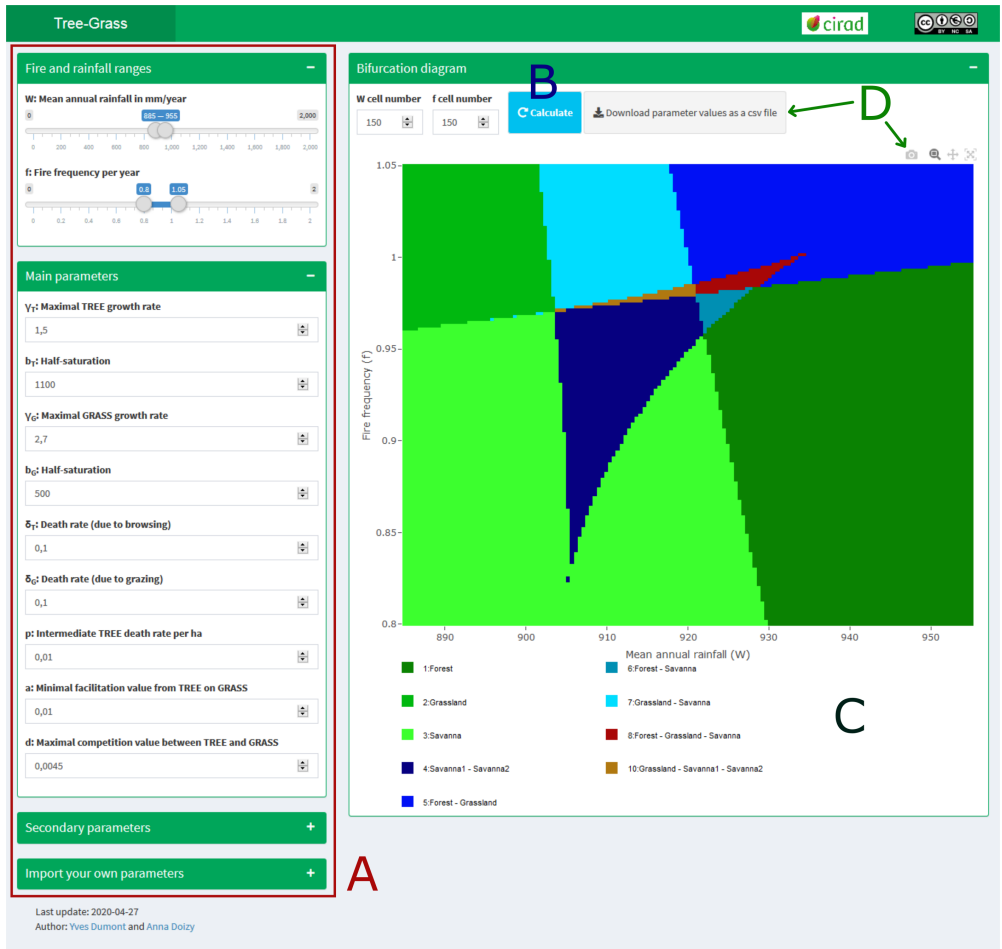}
    \caption{Illustration of the Tree-Grass application interface which permits the user to modify the model parameters (A), launch qualitative thresholds computation (B), get the resulting interactive bifurcation diagram (C) and export the outcomes
(D).}
    \label{code-F-S}
\end{figure}

\par
Figure \ref{bif_diagramme} depicts the outcomes of model
(\ref{swv_eq1}) depending on fire frequency ($f$) and mean annual rainfall (\textbf{W}). In relation to these two parameters, the system experiences both monostability and multi-stability situations involving desert, forest, grassland and savanna.
In the lowest part of the rainfall gradient,a stable bare soil (i.e. desert) is observed for all values of the fire frequency $f$. For a large stretch of the rainfall gradient, i.e. from $ca.$ \textbf{W}=100 mm.yr$^{-1}$ to $ca.$ \textbf{W}=950 mm.yr$^{-1}$, savannas are found to be stable but for high fire frequencies ($\sim>$0.85) they are nevertheless unlikely to be observed at landscape scale as long as MAP do not exceed 700 - 800 mm. Above this threshold, increasing the fire frequency is predicted to notably reduce tree biomass and induce a shift from monostable savanna to monostable grassland and even to multistable states. In the humid parts of the rainfall gradient (MAP$>$ 950-1000 mm), monostable forest is predicted for low values of the fire frequency while for very high fire frequencies, forest-grassland bistabilty is possible. Thanks to the nonlinear functions $\omega(G)$ and $\vartheta(T)$ several savanna equilibria may exist and may be simultaneously stable. For intermediate MAP values associated with very high fire frequency we moreover note a variety of multi-stable states, i.e. savanna-forest-grassland, savanna-savanna-grassland  tristability, savanna-savanna, savanna-grassland, forest-savanna and forest-grassland bistabilities. 

\begin{figure}[H]
    \centering
   \includegraphics[width = \textwidth]{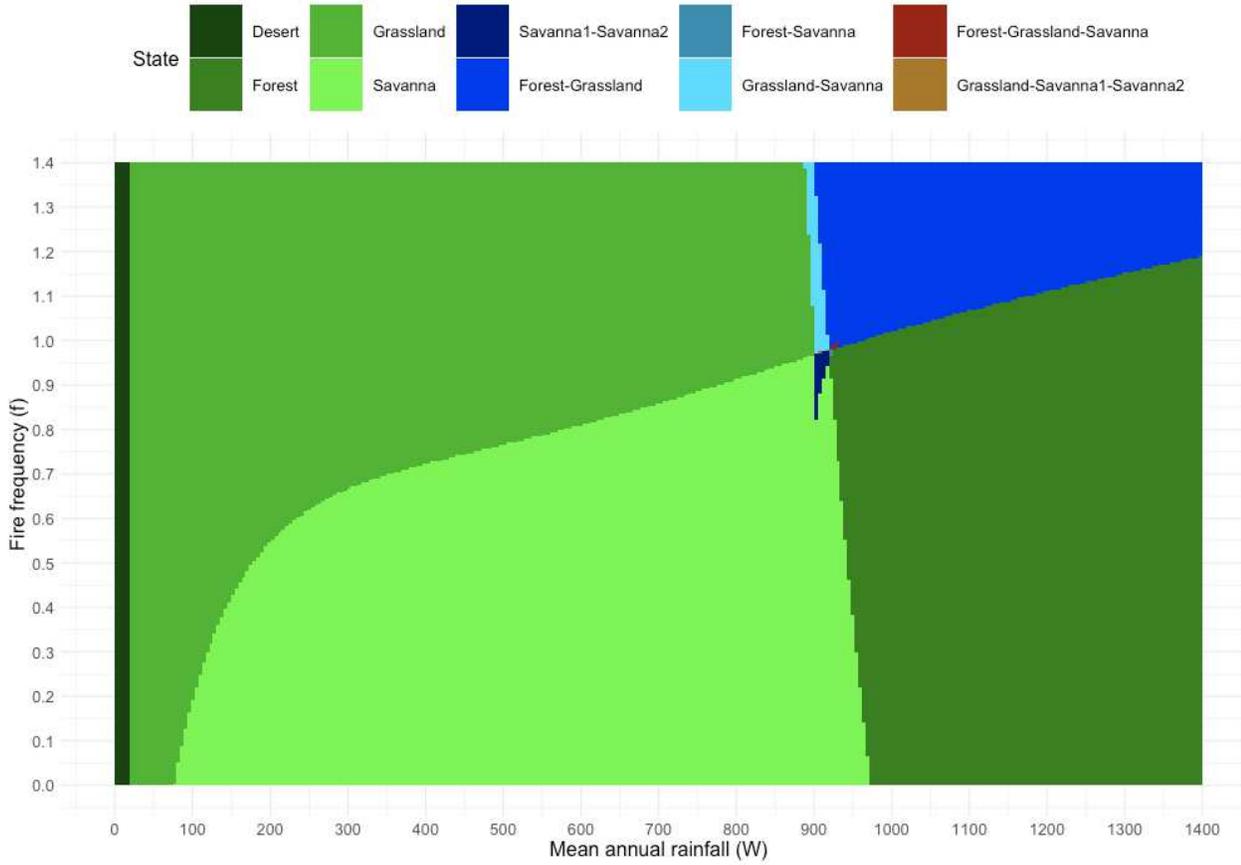}
\caption{{\scriptsize Bifurcation diagram of model
 (\ref{swv_eq1}) obtained along gradients of MAP
($\textbf{W}$) and fire frequency ($f$). Regions in the $\textbf{W}-f$ parameter space are delineated 
according to the thresholds of the qualitative analysis computed from the chosen parameters (Table \ref{params_figR2_fig1}) with monostable states (desert, forest, grassland and savanna), along with bistable states (savanna-savanna, forest-grassland), a forest-savanna, grassland-savanna. A zoom of bistable and tristable states is presented in Figure \ref{zoom_bif_diagramme}, page \pageref{zoom_bif_diagramme}}}
        \label{bif_diagramme}
\end{figure}

Figure \ref{zoom_bif_diagramme} zooms in the bifurcation diagram presented in Figure \ref{bif_diagramme} where we let \textbf{W} to range from 900 mm.yr$^{-1}$ to 940 mm.yr$^{-1}$ and the fire frequency $f$ to range from 0.8 yr$^{-1}$ to 1.05 yr$^{-1}$. The zooming highlights the multistable states that are not visible (savanna-savanna-grassland; savanna-forest) or barely apparent (forest-savanna-grassland) in Figure \ref{bif_diagramme}.

\begin{figure}[H]
    \centering
   \includegraphics[width = 1\textwidth]{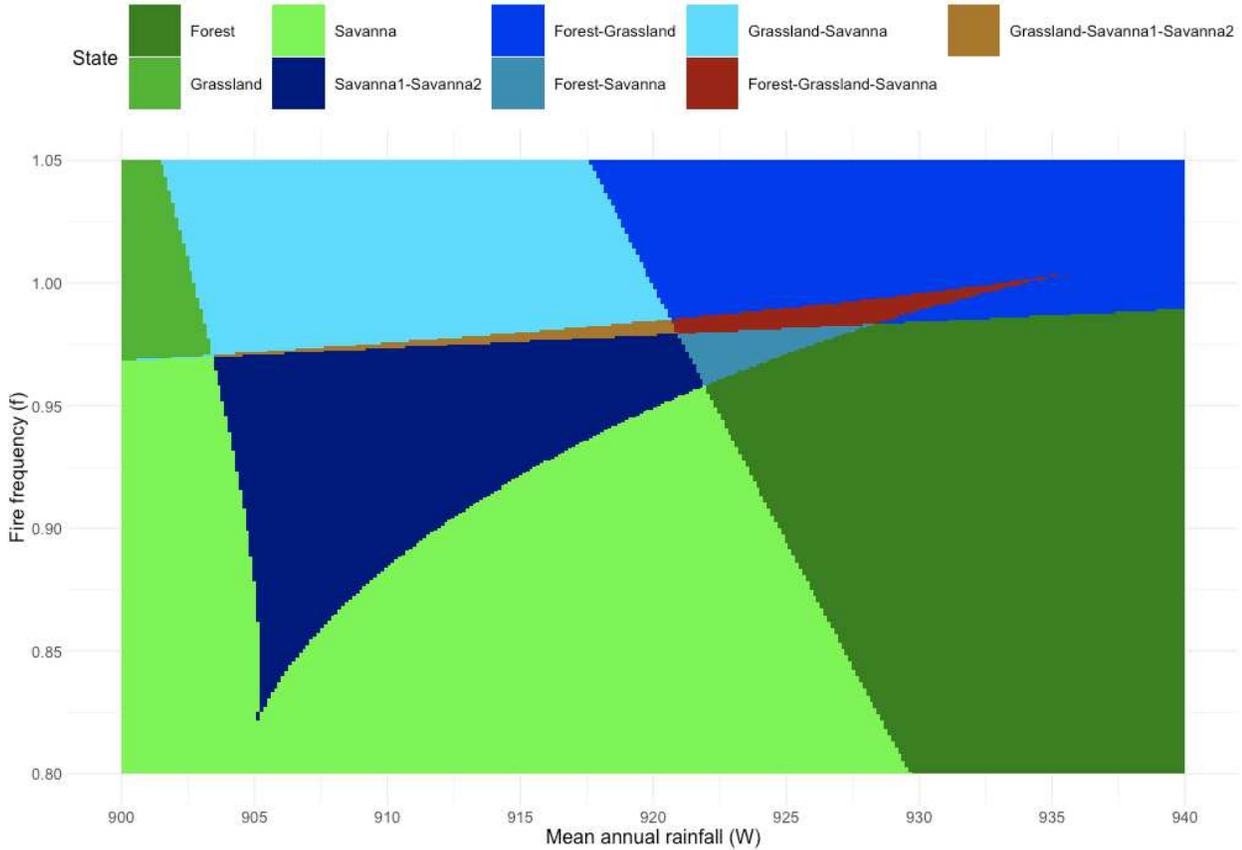}
    \caption{{\scriptsize  Zooming in the bifurcation diagram presented in Figure \ref{bif_diagramme}  to emphasize multi-stable states of limited extent in the $W-f$ parameter space.}}
    \label{zoom_bif_diagramme}
\end{figure}

\comment{
\begin{figure}[H]
    \centering
   \includegraphics[width = 0.8 \textwidth]{legend.eps}
    \caption{{\scriptsize  Color legend describing all possible long-term dynamics of model \eqref{swv_eq1} (see also Figures \ref{bif_diagramme}-\ref{zoom_bif_diagramme}).}}
    \label{legende}
\end{figure}
}
\par
We used simulations of model (\ref{swv_eq1}) and phase portraits of the two state variables to illustrate transitions in the part of the $\textbf{W}$ vs. $f$ parameter space where several multistable configurations were found.
 Figure \ref{bif_W-G_GS_GF} shows a transition from grassland monostability (panel (a)) to forest monostability (panel (f)) with intermediate stages of grassland-savanna bistability (panel (b)), savanna-savanna-grassland tristability (panels (c)), savanna-forest-grassland tristability (panel (d)) and  savanna-forest bistability (panels (e)), as the mean annual rainfall $\textbf{W}$ increases from 902 to 930 mm.yr$^{-1}$ while the fire frequency is kept fixed ($f$=0.98). We note here that a high woody biomass savanna equilibrium (slightly less than 100 $t/ha$) appears in panels (b) and (c) that may be interpreted as open forest with very low, yet perpetuating grass biomass. The woody biomass of the stable forest equilibrium in panels (d), (e) and (f) is just above the value found for the high biomass stable savanna equilibrium. Here the bifurcation owing to a slight increase in mean-annual rainfall entails the final suppression of grass biomass by tree cover competition. We also note that the area of grassland stability in panels (c) and (d) is restricted to a tiny domain of the phase space and cannot be reached for simulations starting from very low levels of woody biomass (especially in (d)).

In Figure \ref{bif_savanna_forest_grassalnd} we depict a transition due to $f$ while the mean annual rainfall is kept constant at $\textbf{W}=920$. It illustrates a shift from a monostable high woody biomass savanna state to a forest-grassland bistability as fire frequency $f$ increases from $f=0.9$ (panel (a)) to $f=1.05$ (panel (e)). Precisely, it shows a transition from savanna monostability (panel (a)) to savanna-savanna bistability  (panel (b)), then savanna-savanna-grassland tristability  (panel (c)) and savanna-grassland bistability (panel (d)) as the fire frequency increases. The woody biomass of the high level savanna equilibrium is of $ca.$ 100 $t/ha$ as in the previous figure, while the slight increase in fire frequency decreases the woody biomass of the lower level savanna equilibrium from $ca.$ 40 $t/ha$ (in (b)) down to 20 $t/ha$ in (c) (this panel being the same as in the previous figure).

From this simulation-based illustration we verify that increasing fire return period for a given rainfall level systematically implies an increase in woody biomass, as classically observed in the field (\citet{Bond2005global}, \citet{Bond2010beyond}, \citet{Mitchard2013woody}).

\begin{figure}[H]
\hspace{-0.6cm}
    \subfloat[][Grassland monostability]{  \includegraphics[scale=0.6]{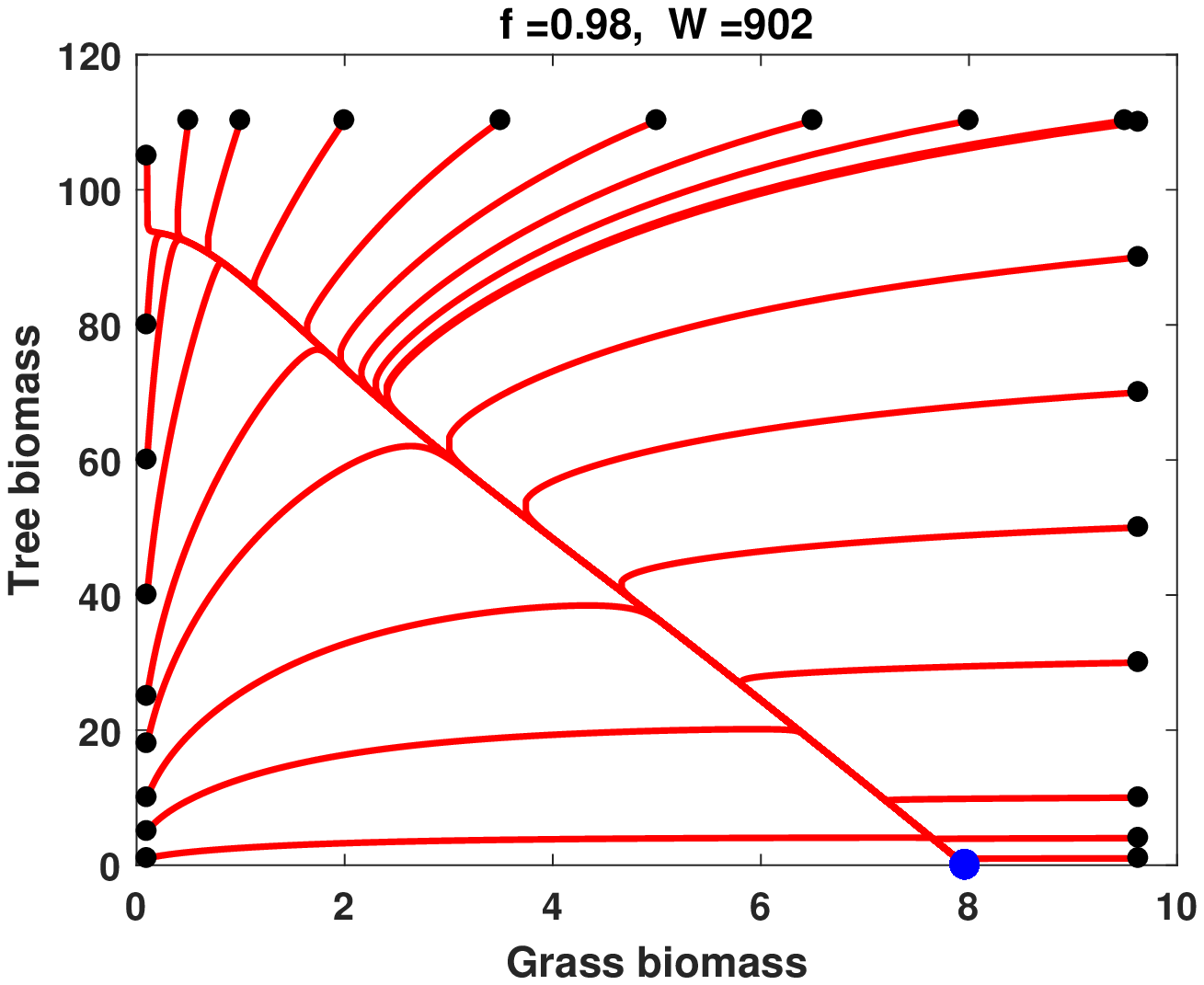}}
    \subfloat[][Savanna-grassland bistability]{  \includegraphics[scale=0.6]{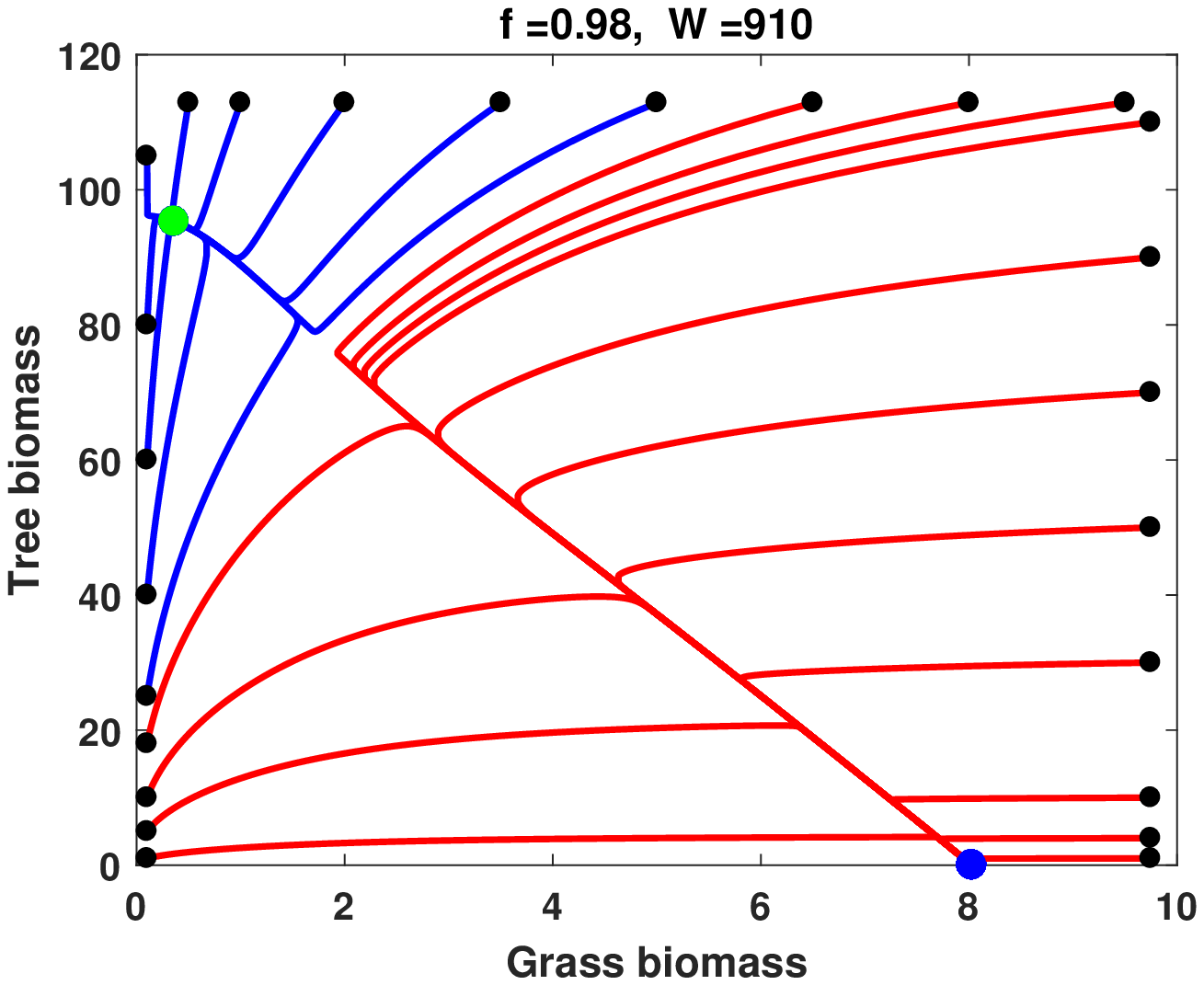}}
    
\hspace{-0.6cm}
    \subfloat[][Savanna-savanna-grassland tristability]{  \includegraphics[scale=0.6]{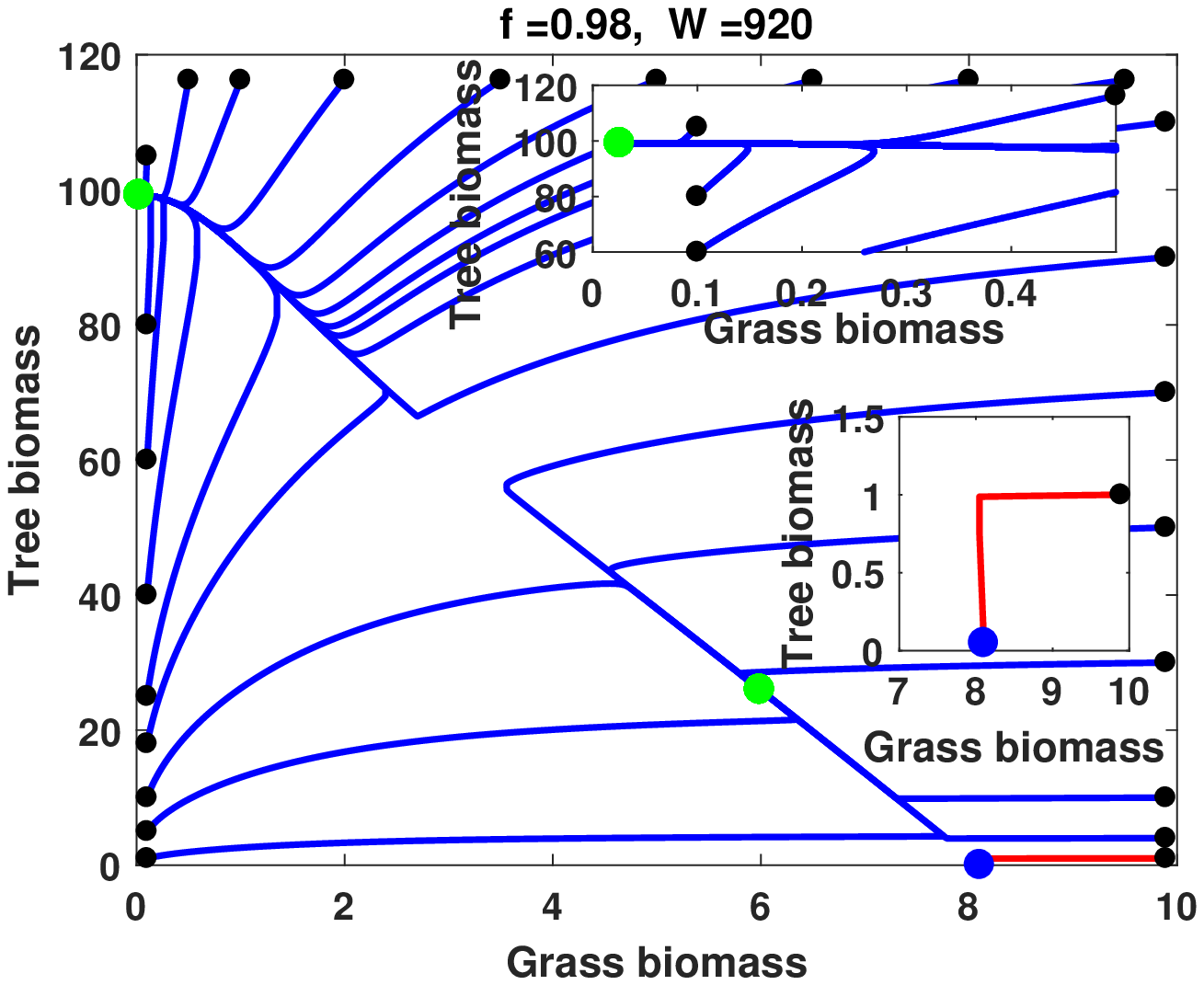}}
    \subfloat[][Forest-savanna-grassland tristability]{  \includegraphics[scale=0.6]{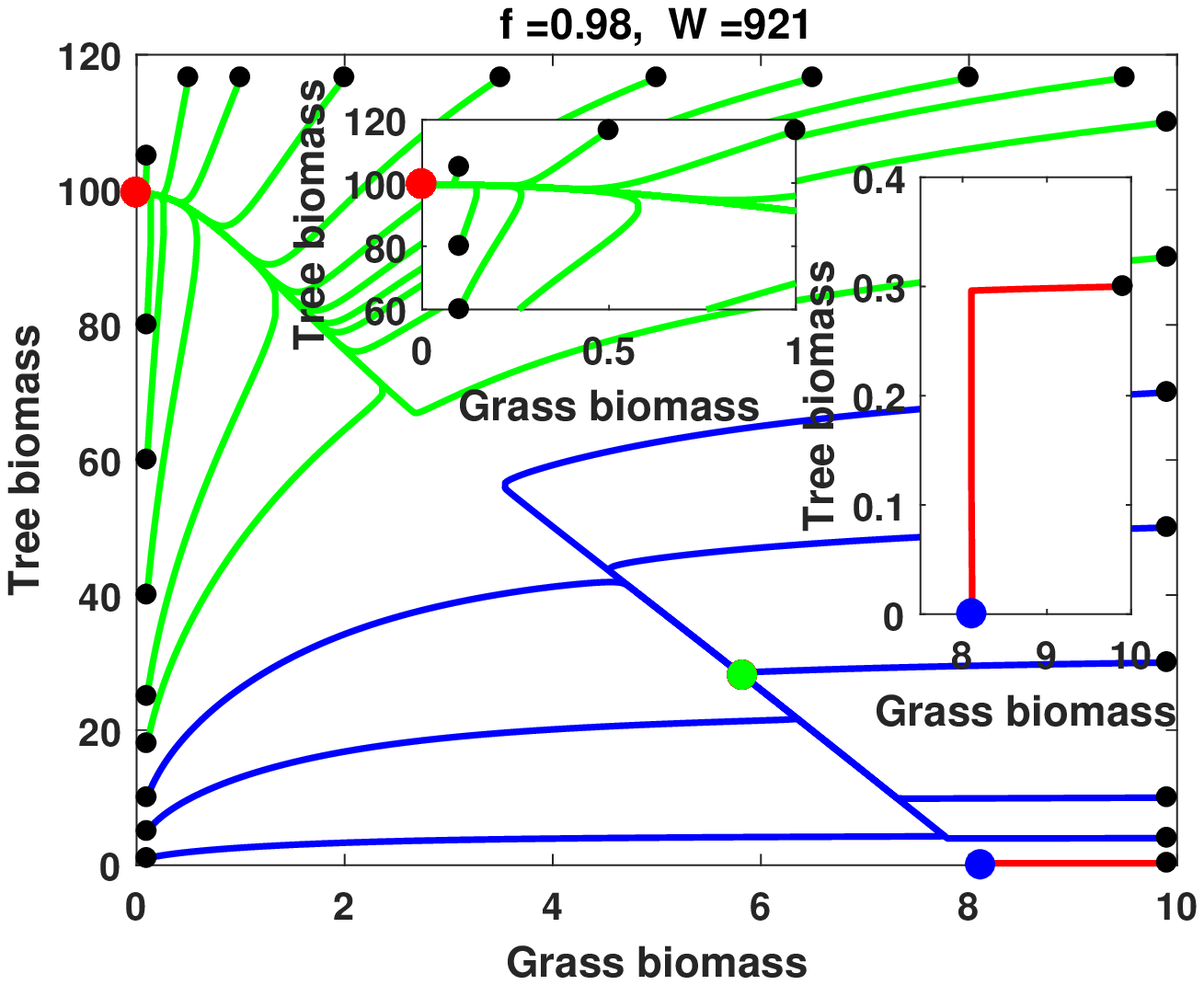}}
    
\hspace{-0.6cm}
    \subfloat[][Savanna-forest bistability]{  \includegraphics[scale=0.6]{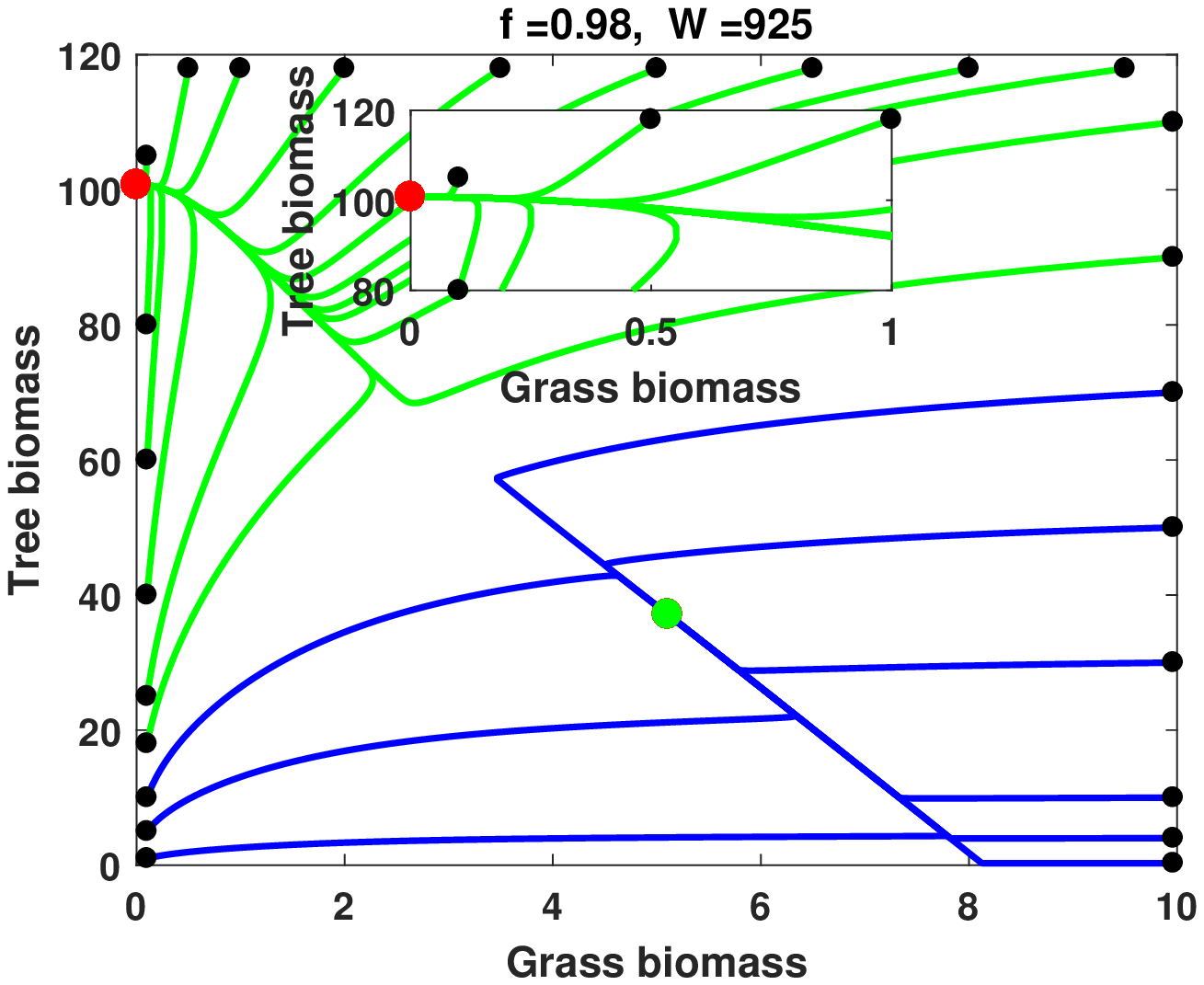}}
    \subfloat[][Forest monostability]{  \includegraphics[scale=0.6]{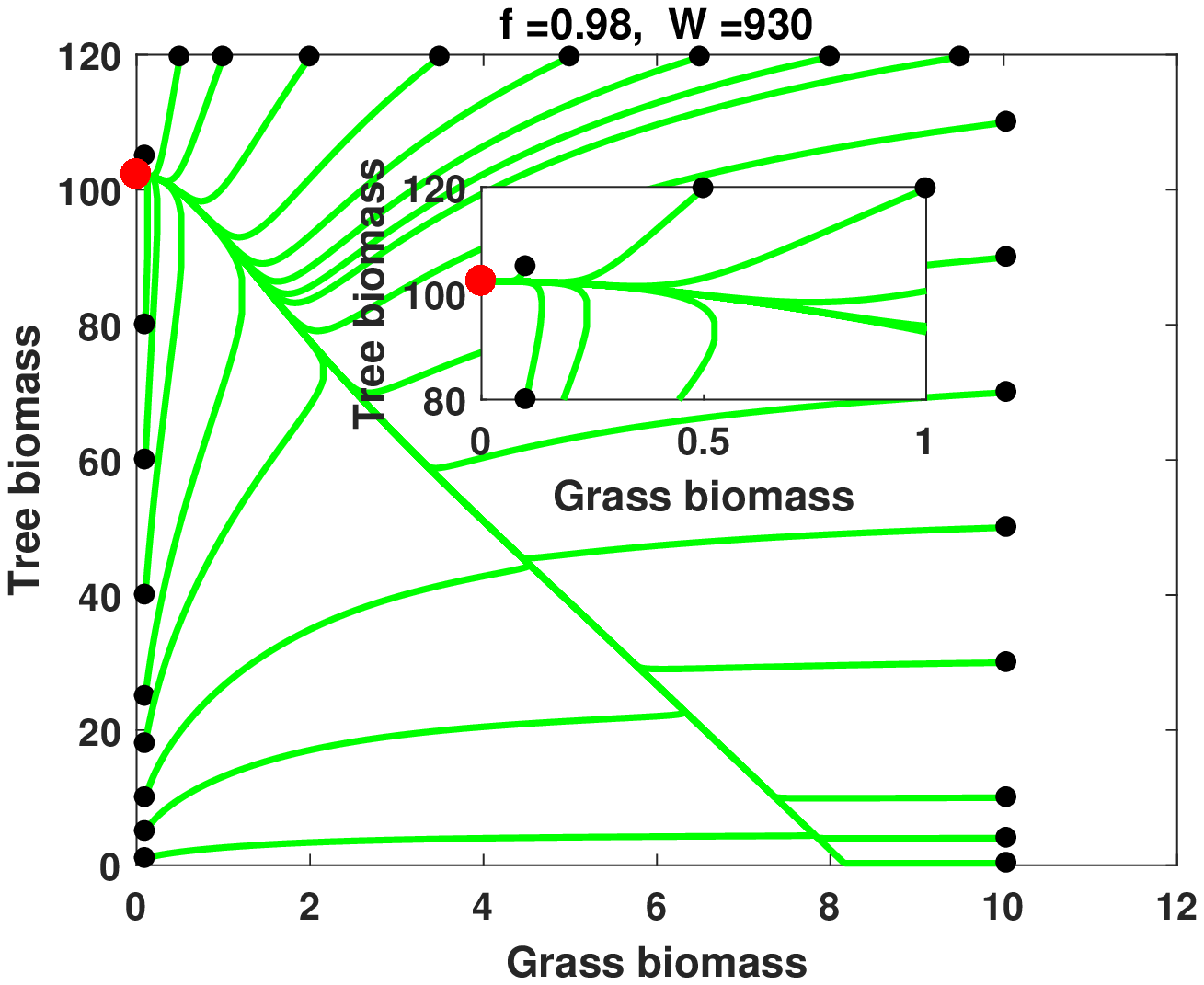}}
    \caption{{\scriptsize  Phase diagrams for grass and woody biomasses (in $t/ha$) illustrating from simulations of model (\ref{swv_eq1}) a transition from grassland monostability (panel (a)) to forest monostability (panel (f)) due  to  an increase in mean annual rainfall \textbf{W}. Black dots represent simulation starting points in phase space, the green dot stands for the stable savanna, the red dot denotes the stable forest while the blue dot represents the stable grassland. Insets magnify the model behavior around equilibria.}}
    \label{bif_W-G_GS_GF}
\end{figure}

\begin{figure}[H]
\hspace{-0.6cm}  
\subfloat[][Savanna monostability]{ \includegraphics[scale=0.6]{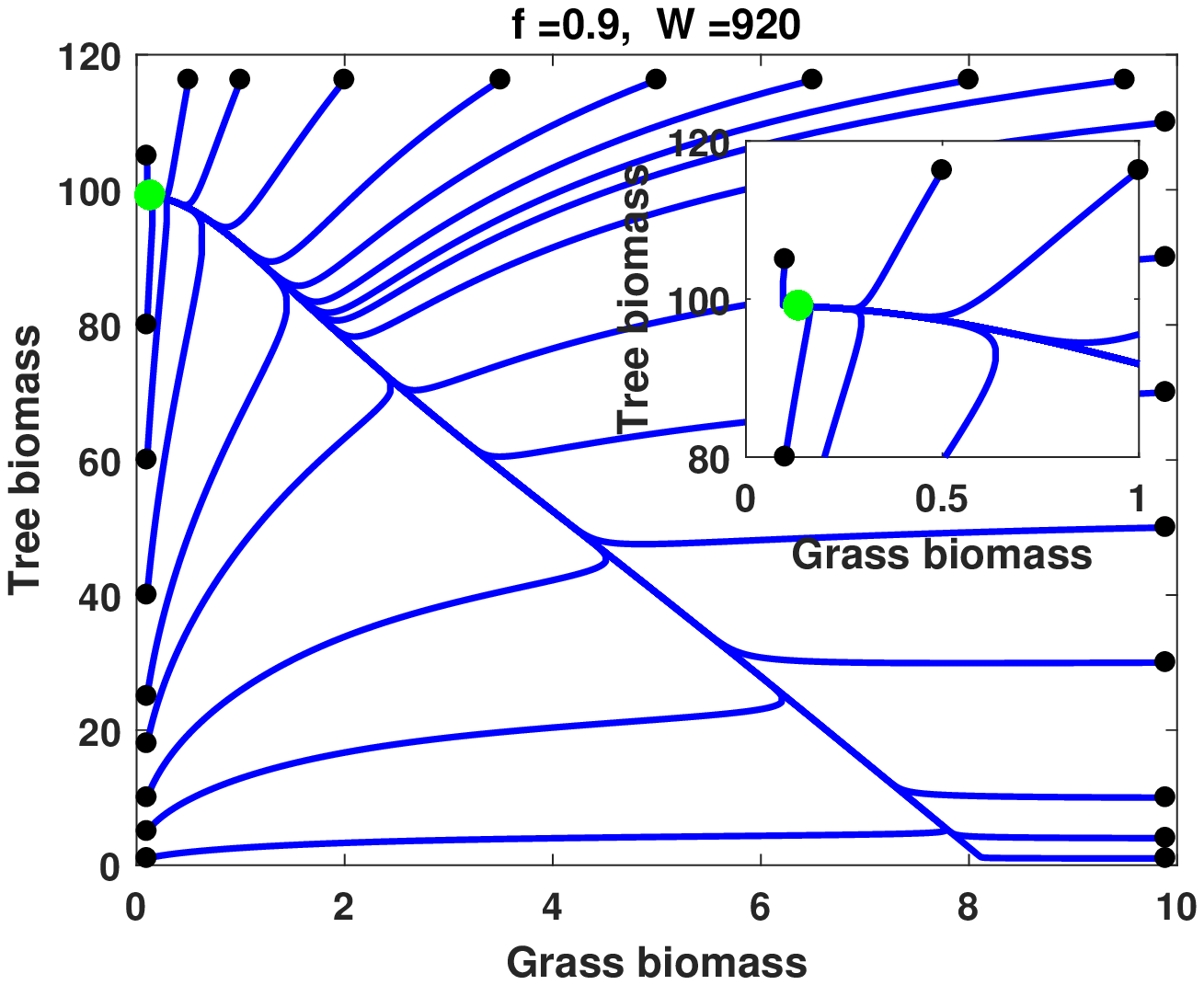}}
\subfloat[][Savanna-savanna bistability]{ \includegraphics[scale=0.6]{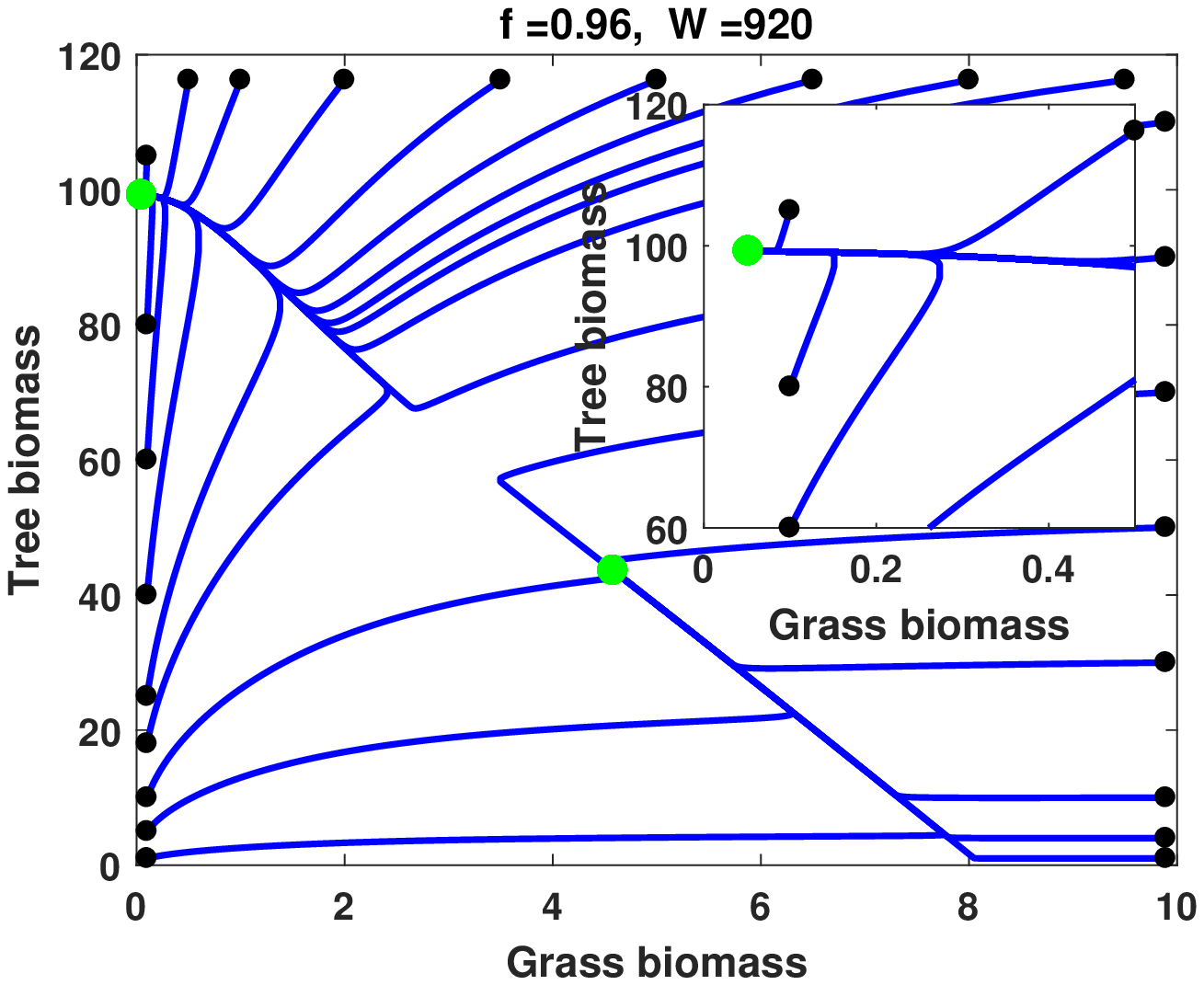}}

\hspace{-0.6cm} 
\subfloat[][Savanna-savanna-grassland tristability]{  \includegraphics[scale=0.6]{savanna1-savanna2-grassland_new.eps}}
  \subfloat[][Savanna-grassland bistability]{ \includegraphics[scale=0.6]{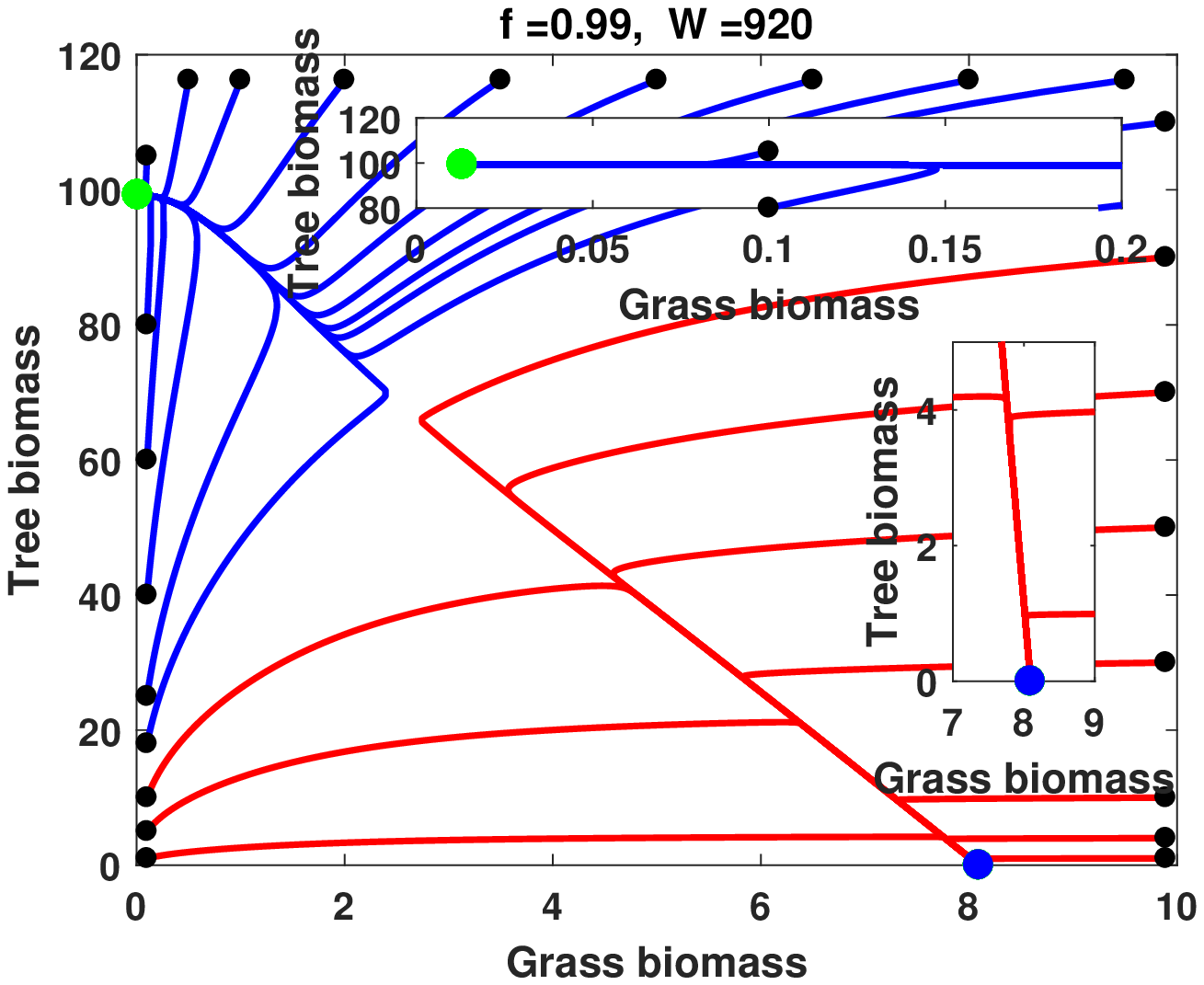}} 
 
\hspace{3.5cm} 
\subfloat[][Forest-grassland bistability]{  \includegraphics[scale=0.6]{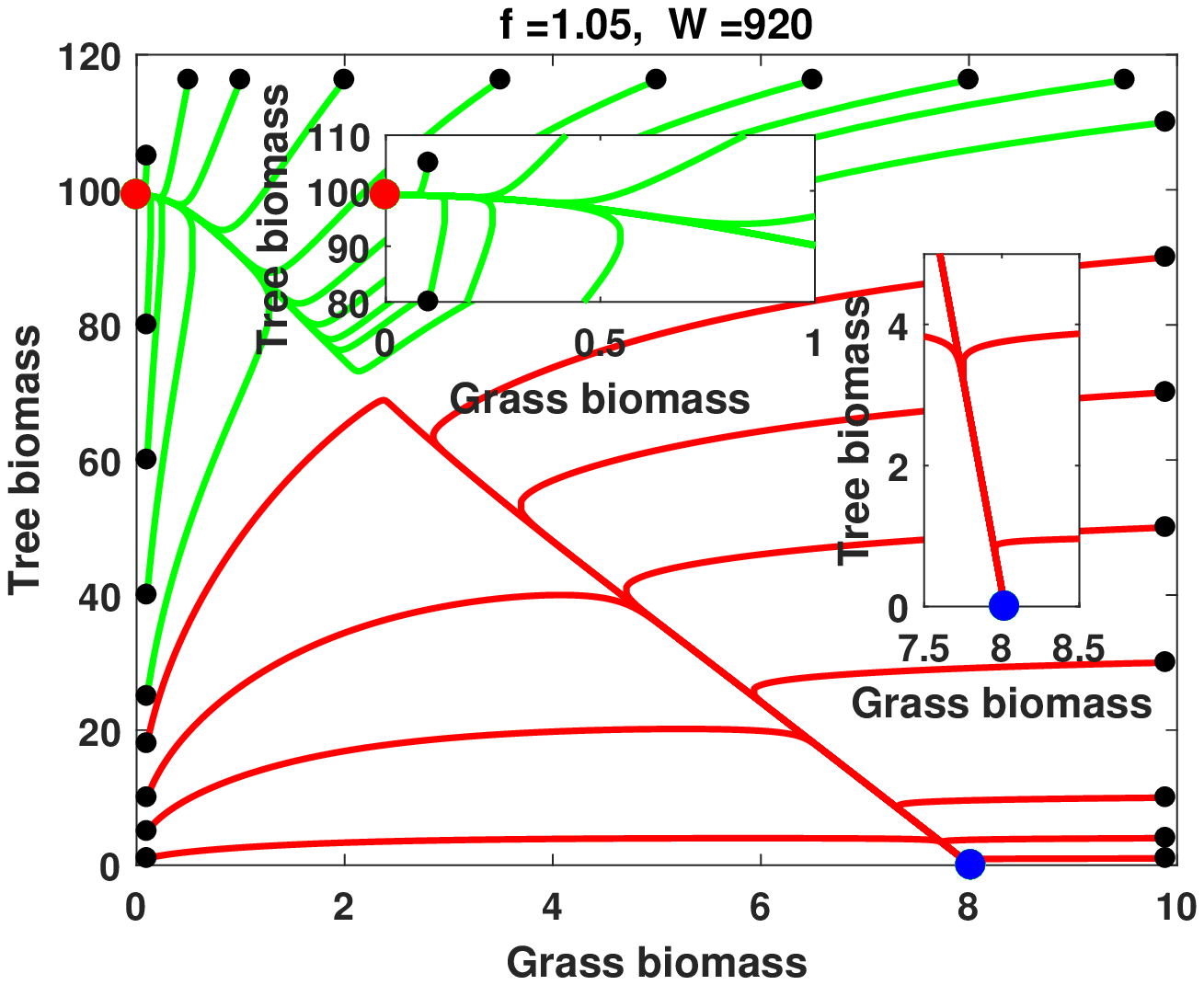}}
\caption{{\scriptsize  Phase portrait for grass and tree biomasses (in $t/ha$) illustrating from simulations of model (\ref{swv_eq1}) the transition from a high woody biomass monostable savanna (panel (a)) to forest-grassland bistability (panel (e)) due  to increasing fire frequency $f$ while keeping constant the mean annual rainfall at $\textbf{W}=920$ mm.yr$^{-1}$.  Black dots represent the starting points of simulations, the green dot stands for the stable savanna, the red dot denotes stable forest while the blue dot represents the stable grassland. Insets magnify the model behavior around equilibria.}}
   \label{bif_savanna_forest_grassalnd}
\end{figure}

\section{Discussion}\label{discussion}

The present line of modelling aimed at demonstrating that meaningful and diversified outcomes can be expected from parsimonious, mathematically tractable models of grassy and woody biomasses interactions in the savanna biome. On the basis of a simple ODE framework, sensible results were indeed reached regarding how vegetation physiognomies change in relation to MAP and fire frequency. The model is liable to predict `trivial' equilibria, i.e. desert, grassland and forest as well as coexistence savanna equilibria (up to five of them), that is the main physiognomies encountered along the rainfall gradients of inter-tropical zones. The qualitative analysis also defined several ecological thresholds that delineate regions of monostability, bistability and tristability involving these equilibria. The bifurcation diagram allows us verifying that shifts between regions induced by increasing fire frequency do not favour the woody component of vegetation: monostable forest gives place to forest-grassland bistability, savanna shifts into grassland, ... This may sound trivial with respect to common experience (\citet{Bond2005global},\citet{Bond2010beyond}, \citet{Favier2012abrupt},\citet{Jeffery2014}), though it is not established to our knowledge that any other model of only two state variables is able to render this fundamental behavior.
For this, the introduction in earlier versions of two independent non-linear functions $\omega(G)$ (see (\ref{omega_fction})) and
$\vartheta(T)$ (see (\ref{theta_fction})) was decisive. Moreover, thanks to $\omega(G)$ and $\vartheta(T)$, more than one savanna coexistence equilibrium may exist for system (\ref{swv_eq1}), while at least two of them may be simultaneously stable. We also found that the model can yield a variety of bistable (forest-savanna, forest-grassland, savanna-grassland) and tristable patterns (forest-savanna-grassland, forest-savanna-savanna and grassland-savanna-savanna). That relatively simple ODE models can lead to complex behaviours has already been highlighted, notably by \citet{Touboul2018} though they used three to four state variables. 
\par
For the set of parameters we used to compute the bifurcation diagram, we found bi- and tristability occurring in sensible situations, with fire frequencies approaching one fire per year and MAP values of (900 - 1000 mm per year) close to those reported as allowing forest to take over savanna in the absence of frequent fires. We however acknowledge that multistable patterns only cover a limited area in the bifurcation diagram while some of them may disappear upon minor changes in some parameters. Similarly, forest-savanna-savanna-grassland quadristability and a limit cycle (linked to facilitation) were proven to be theoretically possible. But neither of them was observed for the ranges of parameters we deemed plausible and this recall the gap between theoretically-possible complexity of dynamical outcomes (as underlined in e.g. \citet{Touboul2018}) and what is actually observable from reasonable parameter ranges. 
Bistable situations involving grassland (as alternative to savanna or forest) seems the most robust but they are linked to very high fire frequencies (above one fire per year) of questionable realism. Under humid equatorial climates, landscape mosaics juxtaposing both grassland-like and forest vegetation are widely observed in places were fire frequencies can exceed one per year because of two dry seasons (\citet{Walters2010}, \citet{Jeffery2014}). But lower intensity and impact on woody stems is reported for too frequent fires (\citet{Walters2010}), while real grasslands in the corresponding landscapes are often associated to seasonal water-logging. There is thus no agreement that fire alone can ensure grassland stability under humid climates.  
\par
More generally, whether complex multistable situations may actually occur in spite of inherent temporal variability of climate and environmental factors is a fully open question. Some observations however suggest that we should not a priori rule them out and that ability to predict their conditions of occurrence on analytical grounds is a desirable property for a model. For instance, an analysis by \citet{Favier2012abrupt} on remote sensing data along a general transect in Central Africa reported (for 3--4$^\circ$ north latitude range), a distribution of woody cover values featuring three modes, namely very low values resembling grassland, large values around 80\% cover indicating forests and intermediate cover values around 40\% suggesting dense savannas. Possible multistabiltiy of equilibria also means that shifts from one stable state to another may often be less abrupt and spectacular that hypothesized from existing models and that trajectories of vegetation change may be more complex than often thought of (\citet{Yatat2014, Yatat2018}). In ecology, the theory of alternative stable states has been to date mostly invoked in relation to bistability of contrasted vegetation types i.e. forest vs. grassland or vs. savanna (assumed of low cover). Consequently, transitions between alternative stable states are frequently termed as abrupt or catastrophic shifts (\citet{Pausas2020}, \citet{Scheffer2001, Scheffer2015}, \citet{Scheffer2003}, \citet{Staver2011tree}, \citet{Favier2012abrupt}, \citet{Yatat2018}) and were therefore deemed unrealistic by some other authors. But we illustrate here that bistability may involve less contrasted states, as well. Notably, we highlighted here the possible existence of two savanna equilibria among which one of high woody biomass, that may be interpreted as dense woodland or open forest. Indeed, the corresponding woody biomass of slightly less than 100 $t/ha$ is in the upper range of values reported for the miombo woodlands, while the associated very low grass biomass is not at odd with most miombo reported figures of less than 2 $t/ha$ (\citet[p. 24-26]{Frost1996}). The area of bistability for the two savanna equilibria is notable in our bifurcation diagram. This finding may echo the long-lasting, unsettled debate about whether open or dry forests, among which the miombos should be considered as transient or {stable states (\citet[p. 50]{Frost1996})}.  

\par
We made here two additions to the model presented by \citet{Yatat2018}. First, we allow the parameter ($\eta_{TG}$) depicting the asymmetric influence of trees on grasses to depend on the biogeographical context through MAP (eq. \eqref{etaTG}). One novelty in the present paper is to provide complete qualitative analyses of the consequences of this choice and we show that shifting from competition to facilitation with decreasing MAP, as empirically evidenced (\citet[page 156]{Abbadie2006lamto}), substantially enriches the possible outcomes of the model. This variety of results illustrates the potential of the ODE framework. Second, we let the fire-induced mortality of grasses non-linearly decrease with annual rainfall instead of being constant (eq. \eqref{lambda_fction}) in order to avoid possible nonsensical results in the dry stretch of the MAP gradient for which fire is known to be absent or a negligible.  
 Our ODE model differs fundamentally from existing
tree-grass models in that MAP is explicit in the parameters of biomass logistic growths. We made a first assessment of these parameters all over the MAP gradient using published results, while there was no previous synthesis about MAP influence on potential maximal woody biomass that encompassed both savannas (as in \citet{Higgins2000fire}) and forests (as in \citet{Lewis2006}). Some  existing models considered rainfall through an additional state variable of soil moisture (see the review of \citet{Yatat2018}) leading to additional parameters and more complex mathematical systems. But there is no real need for a third equation about soil moisture since its
dynamics is very rapid compared to change in vegetation (\citet{Barbier2008}), leading to systems in which the fast soil moisture variable can be eliminated (\citet{Martinez2013spatial}).

\citet{Accatino2016a} questioned the assumption
according to which the parameter $f$ of fire frequency could be constant and 
independent of vegetation characteristics, as in most published ODE models and in some non-ODE ones. In fact, if most fires start from human ignition
(\citet{Favier2004modelling}, \citet{Govender2006},
\citet{Archibald2009}), fires are strongly constrained by available grass fuel and its distribution across space (\citet{Archibald2010}). Here, we keep fire frequency $f$ constant but we interpret it as a man-induced `targeted' fire frequency as for
instance in a fire management plan of a protected area or a ranch. This frequency which will not automatically translate into
actual frequency of fires of notable intensity in all places of a landscape. We therefore modulate $f$ by $\omega(G)$ which will stay in its low branch as long as grass biomass is not of sufficient quantity. For instance, in the `W' National Park in southern Niger, a one-year frequency was targeted by the fire management plan, but the actual average frequency was assessed at 0.7 year$^{-1}$ by a seven-year remote-sensing survey (\citet{Diouf2012}). At the scale of the entire Serengeti National Park (Tanzania), \citet{McNaughton1992} reported that the burnt area fraction (i.e. average fire frequency) dramatically decreased in the 70s due to grass biomass depletion by soaring herbivores populations, although ignition regime by neighbouring communities remained probably unchanged. We thereby split fire frequency from final fire impact on woody biomass in a multiplicative way $(f\times\omega(G)\times\vartheta(T))$. This modelling choice expresses the well-known fact that grass biomass controls both fire spread and local fire intensity which impacts differently small and large woody individuals (\citet{Govender2006}, \citet{McNaughton1992} and references therein). Thus, the $\omega(G)$ function (bounded in [0,1]) is meant to integrate the difficult spreading of fire and thereby modulates the overall, external forcing (i.e. $f$) applied on the tree-grass system. 

From a more general standpoint, ODE approaches have been
criticized by several authors who questioned the modelling of fire as a permanent forcing that continuously removes fractions of fire sensitive biomass all over the year (\citet{Higgins2000fire}, \citet{Baudena2010}, \citet{Beckage2011grass}, \citet{Accatino2013Humid, Accatino2016a}, \citet{Tchuinte2016, Tchuinte2017}, \citet{Yatat2016, Yatat2018}). Indeed, the time between two successive fires is generally long (several months or even years, see \citet[Table IV]{Yatat2018}). Hence, fire may rather be considered as an instantaneous perturbation of the savanna ecosystem (\citet{Yatat2018}). Some authors advocated stochastic modelling of fire occurrences while keeping the continuous-time differential equation framework for vegetation growth and direct interactions between plant forms (\citet{Baudena2010}, \citet{Beckage2011grass}) or using a time-discrete
model (\citet{Higgins2000fire}, \citet{Accatino2013Humid, Accatino2016a}, \citet{Touboul2018}). However, a drawback of most of
these time-discrete stochastic
models (\citet{Higgins2000fire}, \citet{Baudena2010}, \citet{Beckage2011grass}) is that they are less amenable to analytical approaches and often even barely tractable. This is a problem because outcomes spanning limited areas in parameter space (as for the multistable states we evidenced) may be missed by simulations if no qualitative result is available to pinpoint their existence.  
Another line of thought relies on the modelling of fires as impulsive events (\citet{Yatat2016, Yatat2018}, \citet{YatatDumont2018}, \citet{Tchuinte2016, Tchuinte2017} and references therein). This leads to the impulsive differential equation (IDE) framework. To some extends, IDE based models can be seen as a trade-off between realism (discrete nature of fire occurrences) and mathematical tractability (like in the present ODE models). In earlier works, there was no qualitative criteria for the savanna equilibria using IDE models (\citet{Tchuinte2017, Yatat2016}). Last but not least, some processes that likely impact the stability of savanna vegetation including fire spread, seed dispersal and thus tree establishment are spatially structured (see \citet{Li2019} and references therein). Consequently, it is obvious that one cannot expect mean-field or spatially implicit savanna models to accurately reproduce the dynamics of complex, mosaic-like landscapes, even through aggregated values of the two simple state variables we used here.

\section{Conclusion}
\label{conclusion}
In this paper, we presented and analyzed an improved version a `minimalistic' tree-grass
model that addresses the influence of fire and rainfall (MAP) in
tree-grass ecosystems. The model is minimalistic in terms of state variables and parameters, by only explicitly addressing essential processes that are: logistic growth of woody and grassy biomasses, asymmetric direct interactions thereof (both MAP-modulated), positive grass-fire feedback and decreased fire impact on large woody biomass.

The model is fully mathematically tractable and is sufficient to produce a realistic bifurcation diagram rendering the `big picture' of vegetation physiognomies in the savanna biome. Reaching as meaningful results over complete rainfall and fire gradients with less parameters seems challenging. Tractability is important because it allows us to efficiently explore all parts of the parameter space and be sure that interesting situations, notably linked to multi-stability, are not missed as it may happen if only relying on computer simulations. Since well-defined thresholds delineate all outcomes of our model, we can rapidly re-draw the bifurcation diagram after changing parameters, as to better adapt the model response to specific contexts or to integrate improved knowledge on some parameters. We moreover propose a R-Shiny application, ``Tree-Grass", to let ecologists easily explore consequences of modifying parameter values. Results of sensitivity analysis provided in this paper may also guide such explorations and suggest priorities for further data acquisition.

This work can be improved and extended in several ways. One could consider MAP together with potential evapotranspiration (PET) instead of MAP alone as to render that under cooler climates (e.g. in Eastern and Austral Africa) limits in the bifurcation diagrams may shift towards lower MAP values. Adjusting parameter values to more specific reference data sets is also needed to better agree with any given biogeographical context. In contexts where parameters are fairly mastered, spatially explicit approaches are desirable. A former version of the present model has already inspired a spatially explicit model featuring local propagation of grass and tree biomass (e.g. clonal reproduction) through diffusion operators, taking into account continuous fire (\citet{Yatat2018ECOM}), or impulsive periodic fire (\citet{YatatDumont2018,Banasiak2019}). Several studies (e.g. \citet{Borgogno2009}, \citet{Lefever2009}) have also fruitfully modelled non-local plant-plant interactions using kernel operators in reference to arid patterned vegetation and single state variable models (undifferentiated vegetation biomass). Such kernels could be introduced in our model as to embody distance-dependent interactions between grassy and woody biomass in presence of fires. Spatially explicit versions of the present model are desirable, for instance to better address the dynamics of savanna-forest mosaics found under humid climates and investigate the stability of particular landscape features such as localized structures (e.g. groves, \citet{Lejeune2002}) or abrupt boundaries (\citet{Yatat2018ECOM}, \citet{Wuyts2019}) that are of particular relevance to understand the dynamics of forest-savanna mosaics in the face of global change.

\section*{Supplementary materials}
Tree-Grass app source code is freely available at \url{https://gitlab.com/cirad-apps/tree-grass}.

\section*{Acknowledgements}
VY and YD were supported by the DST/NRF SARChI Chair in Mathematical Models and Methods in Biosciences and Bioengineering at the University of Pretoria (grant 82770).
This work benefitted from ongoing field investigation in Cameroon supported by Nachtigal Hydropower Company (Contract n$^o$ C006C007-DES-2017).


\section*{Bibliography}
\bibliography{Biblio_new_model}

\appendix

\section{Analytical results of system (\ref{swv_eq1})} \label{section3}

Here, both competition and facilitation are considered and, will be theoretically analyzed. For reader convenience, we will explicitly state whether we are in the competition or in the facilitation case. To favor the readability of the paper, key theoretical results will be stated in this appendix but their proofs will be given in subsequent appendices.

The right-hand side of system (\ref{swv_eq1}) is $\mathcal{C}^{1}(\R^2)$
i.e., continuously differentiable on $\R^2$. Then, from the Cauchy-Lipschitz
theorem,  system (\ref{swv_eq1}) has a unique maximal solution. From
the ecological point of view, since the variables  of system
(\ref{swv_eq1}) represent biomasses,   each variable must stay
nonnegative and must be bounded during the time evolution (i.e., the
system is said to be biologically well-posed). Note that a solution
with initial conditions in $\R^{2}_{+}$ stays in
$\R^{2}_{+}$ since it can not cut the $y$-axis (vertical null
line) and the $x$-axis (horizontal null line).

In the case of
competition of tree biomass on grass biomass, i.e.
$\eta_{TG}(\textbf{W})\geq0$, we define the subset of solutions
\begin{displaymath}
\Gamma_{\eta_{TG}(\textbf{W})\geq0}=\left\{(G,T)'\in \R^{2}_{+}
: G\leq K_{G}(\textbf{W}), T\leq K_{T}(\textbf{W})\right\}.
\end{displaymath}

In the case of facilitation of tree biomass on grass
biomass, i.e. $\eta_{TG}(\textbf{W})<0$, we consider the subset of solutions
\begin{displaymath}
\Gamma_{\eta_{TG}(\textbf{W})<0}=\left\{(G,T)'\in \R^{2}_{+}
: G\leq
K_{G}(\textbf{W})\times\dfrac{\dfrac{\gamma_{G}\textbf{W}}{b_{G}+\textbf{W}}-\eta_{TG}(\textbf{W})K_T(\textbf{W})}{\dfrac{\gamma_{G}\textbf{W}}{b_{G}+\textbf{W}}},
T\leq K_{T}(\textbf{W})\right\}.
\end{displaymath}

It is straightforward to verify that the subsets $\Gamma_{\eta_{TG}(\textbf{W})\geq0}$ and $\Gamma_{\eta_{TG}(\textbf{W})<0}$ are positively invariant with respect to system (\ref{swv_eq1}). It means that any solutions of (\ref{swv_eq1}) starting in $\Gamma_{\eta_{TG}(\textbf{W})\geq0}$ or $\Gamma_{\eta_{TG}(\textbf{W})<0}$ will remain inside. In other words, any solutions initiated in $\Gamma_{\eta_{TG}(\textbf{W})\geq0}$ or $\Gamma_{\eta_{TG}(\textbf{W})<0}$ will stay nonnegative and bounded.

\subsection{Existence of equilibria}
System (\ref{swv_eq1}) always has the following trivial equilibria: a bare soil equilibrium, i.e. desert, $\textbf{E}_{0}=(0,0)'$; a forest equilibrium $\textbf{E}_{F}=(0,T^{*})'$ which exists when $\mathcal{R}^{1}_{\textbf{W}}>1$; a grassland equilibrium $\textbf{E}_{G}=(G^{*},0)'$ which exists when $\mathcal{R}^{2}_{\textbf{W}}>1$. Existence of savanna equilibria $\textbf{E}_{S}=(G_{*},T_{*})'$ follows from Theorem \ref{al_thm1} in \ref{al_AppendixA}, page \pageref{al_AppendixA}.

\subsection{Stability analysis }
\subsubsection{Stability of equilibria }

In the case $\eta_{TG}(\textbf{W})\geq0$, system (\ref{swv_eq1}) is a planar, competitive and dissipative system. Hence, based on \citet[Theorem 2.2, page 35]{Smith2008}, we deduce that solutions of system (\ref{swv_eq1}) will always converge toward an equilibrium point. That is, no stable limit cycles may exist for system (\ref{swv_eq1}) when $\eta_{TG}(\textbf{W})\geq0$.
Recall that $\mathcal{Q}_{F}$, $\mathcal{R}_{F}$ and $\mathcal{R}_{G}$ are given by \eqref{swv_thresholds_F_G}, page \pageref{swv_thresholds_F_G}. Straightforward computations lead to the following Theorem
\ref{al_thm2} 
that deals with hyperbolic equilibria; that is, none of the eigenvalues of the Jacobian matrix computed at an equilibrium has a null real part (\citet[Definition 1.2.6]{Wiggins2003}). Hence, conclusions of Theorem \ref{al_thm2} follow from \citet[Theorem 1.2.5]{Wiggins2003} and its proof is omitted:
\begin{thm} (Stability of trivial equilibria: the hyperbolic case)
	\begin{itemize}
		\item[(1)] The desert equilibrium $\textbf{E}_{0}=(0, 0)'$ is locally asymptotically stable (LAS) in $\R^2_+$ when $\mathcal{R}^{1}_{\textbf{W}}<1$ and $\mathcal{R}^{2}_{\textbf{W}}<1$ while it is unstable whenever $\mathcal{R}^{1}_{\textbf{W}}>1$ or $\mathcal{R}^{2}_{\textbf{W}}>1$.
		\item[(2)] The grassland equilibrium $\textbf{E}_{G}=(G^{*},0)'$  is  LAS in $\R^2_+$ when $\mathcal{R}_{G}<1$ while it is unstable if $\mathcal{R}_{G}>1$.
		\item[(3)] \begin{description}
			\item[a. \textbf{Competition or Neutrality case}.] When $\eta_{TG}(\textbf{W})\geq0$, the forest equilibrium  $\textbf{E}_{F}=(0, T^{*})'$  is  LAS in $\R^2_+$ whenever $\mathcal{R}_{F}<1$ while it is unstable when $\mathcal{R}_{F}>1$
			\item[b. \textbf{Facilitation case}.] When $\eta_{TG}(\textbf{W})<0$, the forest equilibrium  $\textbf{E}_{F}=(0, T^{*})'$  is  LAS in $\R^2_+$ whenever
			$\mathcal{Q}_{F}<1$ and it is unstable if $\mathcal{Q}_{F}>1$.
		\end{description}
	\end{itemize}
	
	\label{al_thm2}
\end{thm}

In Theorem \ref{al_thm2}, the threshold $\mathcal{R}^{1}_{\textbf{W}}$,  $\mathcal{R}^{2}_{\textbf{W}}$, $\mathcal{R}_{G}$, $\mathcal{R}_{F}$ or $\mathcal{Q}_{F}$ is either lower or greater than one. However, from a direct computation of Jacobian matrix at $\textbf{E}_{0}$, $\textbf{E}_{F}$ or $\textbf{E}_{G}$ one deduces that if any the previous thresholds is equal to one then the corresponding equilibrium becomes non-hyperbolic. In that case, Theorem \ref{al_thm2-bis} is valid.

\begin{thm} (Stability of trivial equilibria: the non-hyperbolic case)
	\begin{itemize}
		\item[(1)] The desert equilibrium $\textbf{E}_{0}=(0, 0)'$ is LAS in $\R^2_+$ when ($\mathcal{R}^{1}_{\textbf{W}}<1$ and $\mathcal{R}^{2}_{\textbf{W}}=1$), ($\mathcal{R}^{1}_{\textbf{W}}=1$ and $\mathcal{R}^{2}_{\textbf{W}}=1$) or ($\mathcal{R}^{1}_{\textbf{W}}=1$ and $\mathcal{R}^{2}_{\textbf{W}}<1$).
		\item[(2)] The grassland equilibrium $\textbf{E}_{G}=(G^{*},0)'$  is  LAS in $\R^2_+$ when $\mathcal{R}_{G}=1$.
		\item[(3)] 
		\begin{description}
			\item[a. \textbf{Competition or Neutrality case}.] When $\eta_{TG}(\textbf{W})\geq0$, the forest equilibrium  $\textbf{E}_{F}=(0, T^{*})'$  is  LAS in $\R^2_+$ whenever $\mathcal{R}_{F}=1$.
			\item[b. \textbf{Facilitation case}.] When $\eta_{TG}(\textbf{W})<0$, the forest equilibrium  $\textbf{E}_{F}=(0, T^{*})'$  is  LAS in $\R^2_+$ whenever
			$\mathcal{Q}_{F}=1$.
		\end{description}
	\end{itemize}
	\label{al_thm2-bis}
\end{thm}

\begin{proof}
	See \ref{proff-al_thm2-bis}, page \pageref{proff-al_thm2-bis}.
\end{proof}

\begin{rmq}
	The existence of $\textbf{E}_{F}$ or $\textbf{E}_{G}$ destabilizes the desert equilibrium $\textbf{E}_{0}$. Hence, there is no bistability between vegetation and bare soil.	
	\label{swv_rmq_1}
\end{rmq}

Let $\textbf{E}_{S}=(G_{*}, T_{*})'$ be a savanna equilibrium given
by Proposition \ref{proposition-competition}. 
 If there exist several savanna
equilibria, for each of them, we define the three following threshold:

\begin{equation}
\left\{
\begin{array}{l}
\mathcal{R}^{1}_{*}=\dfrac{-f\omega(G_{*})T_{*}\vartheta^{'}(T_{*})}
{\left(\dfrac{g_{G}(\textbf{W})}{K_{G}(\textbf{W})}G_{*}+\dfrac{g_{T}(\textbf{W})}{K_{T}(\textbf{W})}T_{*}\right)},
\\
\mathcal{R}^{2}_{*}=\dfrac{\dfrac{g_{G}(\textbf{W})g_{T}(\textbf{W})}{K_{G}(\textbf{W})K_{T}(\textbf{W})}}
{\left(-f\dfrac{g_{G}(\textbf{W})}{K_{G}(\textbf{W})}\omega(G_*)\vartheta'(T_*)+f\eta_{TG}(\textbf{W})\vartheta(T_*)\omega'(G_*)
	\right)}, \quad \mbox{when} \quad \eta_{TG}(\textbf{W})\geq0,\\
\mathcal{Q}^{2}_{*}=\dfrac{\dfrac{g_{G}(\textbf{W})g_{T}(\textbf{W})}{K_{G}(\textbf{W})K_{T}(\textbf{W})}-f\eta_{TG}(\textbf{W})\vartheta(T_*)\omega'(G_*)}
{-f\dfrac{g_{G}(\textbf{W})}{K_{G}(\textbf{W})}\omega(G_*)\vartheta'(T_*)}, \quad \mbox{when} \quad \eta_{TG}(\textbf{W})<0.\\
\end{array}
\right.
\label{swv_thresholds_S}
\end{equation}

Concerning the stability of savanna equilibria, the following theorem holds:

\begin{thm} (Stability of the savanna equilibrium)
	\begin{description}
		\item[a. \textbf{Competition or Neutrality case}.] Assume that $\eta_{TG}(\textbf{W})\geq0$. Then, the savanna equilibrium $\textbf{E}_{S}=(G_{*}, T_{*})'$  is locally
		asymptotically stable whenever $\mathcal{R}^{1}_{*}<1$ and
		$\mathcal{R}^{2}_{*}>1$.
		\item[b. \textbf{Facilitation case}.] Assume that $\eta_{TG}(\textbf{W})<0$. Then, the savanna equilibrium $\textbf{E}_{S}=(G_{*}, T_{*})'$  is locally
		asymptotically stable whenever $\mathcal{R}^{1}_{*}<1$ and
		$\mathcal{Q}^{2}_{*}>1$.
	\end{description}
	\label{al_thm3}
\end{thm}

\begin{proof}
	See \ref{al_AppendixB}, page \pageref{al_AppendixB}.
\end{proof}

\begin{rmq}Multi-stability of savanna equilibria.\\
	It should be noted that several savanna equilibria may
	simultaneously verify requirements of Theorem \ref{al_thm3}. This
	case is the so-called multi-stability situations involving several
	savanna equilibria.
\end{rmq}

%

\subsubsection{Limit cycle and the Hopf bifurcation}
When $\eta_{TG}(\textbf{W})<0$, i.e. in the facilitation case, system (\ref{swv_eq1}) is a planar and dissipative system but it is no longer a competitive system. Hence,
Theorem \ref{Poicare-bendixson}, that ensures the existence of a limit
cycle, follows from the Poincar\'e-Bendixson theorem, see e.g.
\citet[Theorem 1.20]{Augier2010}.

\begin{thm}
	In the case where all equilibria of system (\ref{swv_eq1}) are
	unstable, then one of the following holds true:
	\begin{itemize}
		\item[(i)] Solutions of system (\ref{swv_eq1}) all converge
		toward a periodic solution.
		\item[(ii)] System (\ref{swv_eq1}) admits a limit cycle like homoclinic or
		heteroclinic cycle.
	\end{itemize}
	\label{Poicare-bendixson}
\end{thm}

In the following, we deal with the case where $\eta_{TG}(\textbf{W})<0$ and a periodic solution bifurcates from a savanna equilibrium.
Assume that, for the savanna equilibrium point
$\textbf{E}_{S}=(G_{*}, T_{*})'$, one has 
\begin{equation}\label{condtion1}
\mathcal{Q}^2_*>1.
\end{equation}
Following Theorem \ref{al_thm3}-\textbf{b}, $\textbf{E}_{S}=(G_{*}, T_{*})'$ is LAS if, in addition to \eqref{condtion1}, one has $\mathcal{R}^1_*<1$. Therefore, even with \eqref{condtion1} satisfies, it can be concluded
that when $\eta_{TG}(\textbf{W})<0$, the savanna equilibrium point $\textbf{E}_{S}=(G_{*}, T_{*})'$
may lose its stability through a Hopf bifurcation under certain
parametric conditions. Considering the fire frequency $f$ as a bifurcation parameter,
one can compute the threshold value
$$f
=f_h=-\dfrac{\left(\dfrac{g_{G}(\textbf{W})}{K_{G}(\textbf{W})}G_{*}+\dfrac{g_{T}(\textbf{W})}{K_{T}(\textbf{W})}T_{*}\right)}
{\omega(G_{*})T_{*}\vartheta^{'}(T_{*})},$$ which satisfies
\begin{equation}\label{condition-trace}
\left.\mathcal{R}^1_*\right|_{f=f_h}=1.
\end{equation}
Assume also that the following condition holds true
\begin{equation}\label{condition-determinant}
\begin{array}{l}
\left.\mathcal{Q}^2_*\right|_{f=f_h}>1, \quad \mbox{when}\quad
\eta_{TG}(\textbf{W})<0.
\end{array}
\end{equation}
The transversality condition for the Hopf bifurcation is
\begin{equation}\label{condition-transversality}
\left.\dfrac{d}{df}(tr(J_*))\right|_{f=f_h}=-\omega(G_{*,f})T_{*,f}\vartheta'(T_{*,f})>0
\end{equation}
where $tr(J_*)$ is given by (\ref{trace-jacobien}), page
\pageref{trace-jacobien}, and $G_{*,f}$, $T_{*,f}$ indicate the
functionality of the components of the positive savanna equilibrium
$\textbf{E}_{S}=(G_{*}, T_{*})'$ with respect to the parameter $f$. 
\par

Hence, the savanna equilibrium $\textbf{E}_{S}=(G_{*}, T_{*})'$
loses its stability through the Hopf bifurcation when conditions
(\ref{condition-trace}) and (\ref{condition-determinant}) are
satisfied simultaneously. 
\par

Now we calculate the Lyapunov number to determine the nature
of Hopf-bifurcating periodic solutions. 

\begin{thm} (Hopf bifurcation)\label{Hopf-Lyapunov}\\
	Assume that the savanna equilibrium $\textbf{E}_{S}=(G_{*}, T_{*})'$
	exists and that requirements (\ref{condition-trace}) and
	(\ref{condition-determinant}) are satisfied. Hence, there exists a real number $\sigma$ such that, if $\sigma\neq0$, then
	a Hopf bifurcation occurs at $\textbf{E}_{S}=(G_{*}, T_{*})'$ for
	system (\ref{swv_eq1}) at the bifurcation value $f=f_h$. In
	particular
	\begin{itemize}
		\item[(i)] If $\sigma<0$ then, the savanna equilibrium $\textbf{E}_{S}=(G_{*}, T_{*})'$
		destabilizes through a supercritical Hopf bifurcation. That is, a unique stable limit cycle bifurcates
		from $\textbf{E}_{S}=(G_{*}, T_{*})'$.
		\item[(ii)] If $\sigma>0$ then, the Hopf bifurcation is
		subcritical. That is, a unique unstable limit cycle bifurcates
		from $\textbf{E}_{S}=(G_{*}, T_{*})'$.
	\end{itemize}
\end{thm}

\begin{proof}
	See \ref{AppendixD}.
\end{proof}


\begin{rmq}(The case $f=0$)\\
	The particular case where there is no fires in system (\ref{swv_eq1}); that is, when $f=0$, straightforward computations lead to the following conclusions:
	\begin{itemize}
		\item[(i)]  The unique savanna equilibrium $\textbf{E}_{S}=(G_{*}, T_{*})'$ is such that $T_*=T^*$ and \begin{itemize}
			\item[(a)] when $\eta_{TG}(\textbf{W})\geq0$, $G_*=K_G(\textbf{W})\left(1-\dfrac{1}{\mathcal{R}_{F,f=0}}\right)$ where $\mathcal{R}_{F,f=0}$ is computed from $\mathcal{R}_{F}$ with $f=0$.
			\item[(b)] when $\eta_{TG}(\textbf{W})<0$, $G_*=\dfrac{\delta_GK_G(\textbf{W})}{g_G(\textbf{W})}\left(\mathcal{Q}_{F,f=0}-1\right)$ where $\mathcal{Q}_{F,f=0}$ is computed from $\mathcal{Q}_{F}$ with $f=0$.
		\end{itemize}
		\item[(ii)] The threshold $\mathcal{R}_{G}$ is such that 
		$\mathcal{R}_{G}=\mathcal{R}_{\textbf{W}}^1$. 
		\item[(iii)]Grassland-forest, grassland-savanna and forest-savanna bistabilies can not occur. 
		\item[(iv)] The function $B(T,G)=\dfrac{1}{TG}$ is a Dulac's function for system (\ref{swv_eq1}). Hence, system (\ref{swv_eq1}) does not admit a closed orbit such that periodic solutions, homoclinic or heteroclinic cycles. 
	\end{itemize}
\end{rmq}

\section{Existence of a savanna equilibria}
\label{al_AppendixA}

Let us set:
\begin{equation}\label{definition-a-b}
\begin{array}{lcl}
\mathcal{A} &=&\dfrac{g_{T}(\textbf{W})}{K_{T}(\textbf{W})}T^{*},  \\
\mathcal{B}&=&\dfrac{g_{G}(\textbf{W})g_{T}(\textbf{W})}{\eta_{TG}(\textbf{W})K_{G}(\textbf{W})K_{T}(\textbf{W})}G^{*},\\
\mathcal{C}&=&\dfrac{\mathcal{B}}{G^{*}},\\ 
\mathcal{D}&=&f\lambda_{fT}^{min},\\
\lambda&=&f(\lambda_{fT}^{max}-\lambda_{fT}^{min})\times
e^{-p\dfrac{g_{G}(\textbf{W})G^{*}}{\eta_{TG}(\textbf{W})K_{G}(\textbf{W})}},\\
\alpha_0&=&p\dfrac{g_{G}(\textbf{W})}{\eta_{TG}(\textbf{W})K_{G}(\textbf{W})}
\end{array}
\end{equation}

where $T^*$ and $G^{*}$ are given by (\ref{swv_T_G}).

The existence of positive savanna equilibria is given in Theorem \ref{al_thm1}.
\begin{thm} (Existence of savanna equilibria)\\
	A savanna equilibrium $\textbf{E}_{S}=(G_{*},T_{*})'$ satisfies
	
	\begin{equation}
	\left\{
	\begin{array}{lcl}
	g_{G}(\textbf{W})\left(1-\displaystyle\frac{G_{*}}{K_{G}(\textbf{W})}\right)-(\delta_{G}+\lambda_{fG}f)-\eta_{TG}(\textbf{W})T_{*}=0,\\
	\\
	g_{T}(\textbf{W})\left(1-\displaystyle\frac{T_{*}}{K_{T}(\textbf{W})}\right)-\delta_{T}-f\vartheta(T_{*})\omega(G_{*})=0.\\
	\end{array}
	\right.
	\label{app_eq1-bis}
	\end{equation}
	
	Using the first equation of (\ref{app_eq1-bis}), we have
	
	\begin{equation}
	T_{*}=\dfrac{g_{G}(\textbf{W})}{\eta_{TG}(\textbf{W})K_{G}(\textbf{W})}(G^{*}-G_{*}).
	\label{app_eq2-bis}
	\end{equation}
	From (\ref{app_eq2-bis}) we deduce that,
	a condition to have a (positive) savanna equilibrium in the case $\eta_{TG}(\textbf{W})>0$ is:
	\begin{equation}
	G^{*}>G_{*}.
	\label{app_eq2*-bis}
	\end{equation}
	When $\eta_{TG}(\textbf{W})<0$, savanna equilibria are computed with positive $G_*$ such that $T_*$ is also positive.
	Substituting (\ref{app_eq2-bis}) in the second equation of (\ref{app_eq1-bis}) leads that $G_*$ must satisfy: 
	
	\begin{equation}
	\mathcal{C}G_{*}^{3}-\lambda G_{*}^{2}e^{\alpha_0 G_{*}}
	+(\mathcal{A}-\mathcal{B}-\mathcal{D})G_{*}^{2}+\mathcal{C}\alpha^{2}G_{*}+(\mathcal{A}-\mathcal{B})\alpha^{2}=0.
	\label{app_eq7-bis}
	\end{equation}
	
	Table \ref{swv_tab_1} summarizes the conditions of existence of positive solutions $G_*$ of (\ref{app_eq7-bis}), when $\eta_{TG}(\textbf{W})>0$, and that verify (\ref{app_eq2*-bis}). Hence, its summarizes the conditions of existence of savanna equilibria in the case of tree biomass vs. grass biomass competition.    
	
	\begin{table}[H]
		\begin{center}
			\renewcommand{\arraystretch}{1}
			\begin{tabular}{|c|c|c|c|c|}
				\hline
				$\eta_{TG}(\textbf{W})$   &   $\mathcal{C}-\lambda\alpha_0$   & $\mathcal{A}-\mathcal{B}-\mathcal{D}-\lambda$ & $\mathcal{A}-\mathcal{B}$ & Number of savanna equilibria \\
				\hline
				&    \multirow{2}{1cm}{$<0$} & \multirow{2}{1cm}{$<0$} & $<0$ & $0$, $1$ or $2$\\
				\cline{4-5}
				&    & &$>0$ & $0$ or $1$\\
				\cline{3-5}
				&    & $>0$ &$>0$ & $0$ or $1$\\
				\cline{2-5}
				$>0$   &    \multirow{4}{1cm}{$>0$} &\multirow{2}{1cm}{$-$}& $<0$ & $0$, $1$ or $2$\\
				\cline{4-5}
				&    & &$>0$ & $0$ or $1$\\
				\cline{3-5}
				&    & $>0$& $>0$ & $0$ or $1$ \\
				\cline{3-5}
				&    & \multirow{2}{1cm}{$<0$} &$<0$ & $0$, $1$, $2$,  $3$ or $4$ \\
				\cline{4-5}
				&    &  &$>0$ & $0$, $1$, $2$ or  $3$ \\
				\hline
			\end{tabular}
		\end{center}
		\caption{Existence of savanna equilibria in the case of the tree biomass vs. grass biomass competition. ``$-$" stands for any value.}\label{swv_tab_1}
	\end{table}
	
	Table \ref{swv_tab_1bis1} summarizes the conditions of existence of positive solutions $G_*$ of (\ref{app_eq7-bis}), when $\eta_{TG}(\textbf{W})<0$, and that are such that $T_*>0$ (see (\ref{app_eq2-bis})). Hence, its summarizes the conditions of existence of savanna equilibria in the case of tree biomass vs. grass biomass facilitation.    
	
	\begin{table}[H]
		\begin{center}
			\renewcommand{\arraystretch}{1}
			\begin{tabular}{|c|c|c|c|c|}
				\hline
				$\eta_{TG}(\textbf{W})$   &   $\mathcal{C}-\lambda\alpha_0$   & $\mathcal{A}-\mathcal{B}-\mathcal{D}-\lambda$ & $\mathcal{A}-\mathcal{B}$ & Number of savanna equilibria \\
				\hline
				&    \multirow{2}{1cm}{$<0$} & \multirow{2}{1cm}{$<0$} & $<0$ & $0$, $1$ or $2$\\
				\cline{4-5}
				&    & &$>0$ & $0$, $1$, $2$ or $3$\\
				\cline{3-5}
				&    & $>0$ &$>0$ & $0$, $2$, $3$, $4$ or $5$\\
				\cline{2-5}
				$<0$   &    \multirow{4}{1cm}{$>0$} &\multirow{2}{1cm}{$-$}& $<0$ & $0$\\
				\cline{4-5}
				&    & &$>0$ & $0$ or $1$\\
				\cline{3-5}
				&    & $>0$& $>0$ & $0$, $1$, $2$ or $3$ \\
				\cline{3-5}
				&    & \multirow{2}{1cm}{$<0$} &$<0$ & $0$, $1$ or $2$\\
				\cline{4-5}
				&    &  &$>0$ & $0$, $1$, $2$ or  $3$ \\
				\cline{2-5}
				& \multirow{2}{1cm}{$-$}& $>0$ & $>0$ & $0$, $1$, $2$ or $3$\\
				\cline{3-5}
				& &       \multirow{2}{1cm}{$<0$} &$<0$ & $0$\\
				\cline{4-5}
				& & & $>0$ & $0$ or $1$\\
				\hline
			\end{tabular}
		\end{center}
		\caption{Existence of savanna equilibria in the case of tree biomass vs. grass biomass facilitation. ``$-$" stands for any value.}\label{swv_tab_1bis1}
	\end{table}   
	
	When $\eta_{TG}(\textbf{W})=0$, one has
	\begin{equation}
	\left\{
	\begin{array}{l}
	G_* = G^*, \\
	T^*-T_*-\displaystyle\frac{K_T(\textbf{W})}{g_T(\textbf{W})}f\vartheta(T_*)\omega(G^*)=0. 
	\end{array}
	\right.
	\end{equation}
	Let us set
	\begin{equation}
	\begin{array}{l}
	u =\displaystyle\frac{K_T(\textbf{W})}{g_T(\textbf{W})}f\omega(G^*)\lambda_{fT}^{min}, \\
	v =\displaystyle\frac{K_T(\textbf{W})}{g_T(\textbf{W})}f\omega(G^*)(\lambda_{fT}^{max}-\lambda_{fT}^{min}),\\
	J(T)=T^*-T-u-ve^{-pT}.
	\end{array}
	\end{equation}
	Hence,
	\begin{enumerate}
		\item[1.] if $-1+pv>0$ then, there may exist $0$, $1$ or $2$ savanna equilibria.
		\item[2.] if $-1+pv\leq0$ then, there may exist $0$ or $1$ savanna equilibrium.
	\end{enumerate}
	\label{al_thm1}
\end{thm}

\begin{proof}
From system (\ref{swv_eq1}),  a savanna equilibrium $\textbf{E}_{S}=(G_{*},T_{*})'$ satisfies

\begin{equation}
\left\{
\begin{array}{lcl}
g_{G}(\textbf{W})\left(1-\displaystyle\frac{G_{*}}{K_{G}(\textbf{W})}\right)-(\delta_{G}+\lambda_{fG}f)-\eta_{TG}(\textbf{W})T_{*}=0,\\
\\
g_{T}(\textbf{W})\left(1-\displaystyle\frac{T_{*}}{K_{T}(\textbf{W})}\right)-\delta_{T}-f\vartheta(T_{*})\omega(G_{*})=0.\\
\end{array}
\right.
\label{app_eq1}
\end{equation}

Using the first equation of (\ref{app_eq1}), we have

\begin{equation}
T_{*}=\dfrac{1}{\eta_{TG}(\textbf{W})}\left(g_{G}(\textbf{W})-(\delta_{G}+\lambda_{fG}f)-\dfrac{g_{G}(\textbf{W})}{K_{G}(\textbf{W})}G_{*}\right)=\dfrac{g_{G}(\textbf{W})}{\eta_{TG}(\textbf{W})K_{G}(\textbf{W})}(G^{*}-G_{*}).
\label{app_eq2}
\end{equation}

Substituting (\ref{app_eq2}) in the second equation of (\ref{app_eq1}) gives

\begin{equation}
\dfrac{(g_{T}(\textbf{W})-\delta_{T})-\dfrac{g_{G}(\textbf{W})g_{T}(\textbf{W})}{\eta_{TG}(\textbf{W})K_{G}(\textbf{W})K_{T}(\textbf{W})}(G^{*}-G_{*})}{\omega(G_{*})}=f\vartheta(T_{*}).
\label{app_eq3}
\end{equation}

From (\ref{app_eq3}), introducing the expression of $\omega(G)$, we have

\begin{equation}
\dfrac{\dfrac{g_{T}(\textbf{W})}{K_{T}(\textbf{W})}T^{*}-\dfrac{g_{G}(\textbf{W})g_{T}(\textbf{W})}{\eta_{TG}(\textbf{W})K_{G}(\textbf{W})K_{T}(\textbf{W})}G^{*}+\dfrac{g_{G}(\textbf{W})g_{T}(\textbf{W})}{\eta_{TG}(\textbf{W})K_{G}(\textbf{W})K_{T}(\textbf{W})}G_{*}}{\dfrac{G_{*}^{2}}{G_{*}^{2}+\alpha^{2}}}=f\vartheta(T_{*}),
\label{app_eq4}
\end{equation}
where

\begin{equation}
f\vartheta(T_{*})=f\lambda_{fT}^{min}+f(\lambda_{fT}^{max}-\lambda_{fT}^{min})\times
e^{-p\dfrac{g_{G}(\textbf{W})G^{*}}{\eta_{TG}(\textbf{W})K_{G}(\textbf{W})}}\times
e^{p\dfrac{g_{G}(\textbf{W})G_{*}}{\eta_{TG}(\textbf{W})K_{G}(\textbf{W})}}.
\label{app_eq5}
\end{equation}

From (\ref{app_eq4}) and (\ref{app_eq5}) we have:

\begin{equation}
(\mathcal{A}-\mathcal{B}+\mathcal{C}G_{*})\left(1+\dfrac{\alpha^{2}}{G_{*}^{2}}\right)=\mathcal{D}+\lambda
e^{\alpha_0 G_{*}}, \label{app_eq6}
\end{equation}
where,\\
$\mathcal{A}=\dfrac{g_{T}(\textbf{W})}{K_{T}(\textbf{W})}T^{*}$,
$\mathcal{B}=\dfrac{g_{G}(\textbf{W})g_{T}(\textbf{W})}{\eta_{TG}(\textbf{W})K_{G}(\textbf{W})K_{T}(\textbf{W})}G^{*}$,
$\mathcal{C}=\dfrac{\mathcal{B}}{G^{*}}$, $\mathcal{D}=f\lambda_{fT}^{min}$,
$\lambda=f(\lambda_{fT}^{max}-\lambda_{fT}^{min})\times
e^{-p\dfrac{g_{G}(\textbf{W})G^{*}}{\eta_{TG}(\textbf{W})K_{G}(\textbf{W})}}$
and
$\alpha_0=p\dfrac{g_{G}(\textbf{W})}{\eta_{TG}(\textbf{W})K_{G}(\textbf{W})}$.\par

From equation (\ref{app_eq6}) we have

\begin{equation}
\mathcal{C}G_{*}^{3}-\lambda G_{*}^{2}e^{\alpha_0 G_{*}}
+(\mathcal{A}-\mathcal{B}-\mathcal{D})G_{*}^{2}+\mathcal{C}\alpha^{2}G_{*}+(\mathcal{A}-\mathcal{B})\alpha^{2}=0.
\label{app_eq7}
\end{equation}

Set    $H(G_{*})=\mathcal{C}G_{*}^{3}-\lambda G_{*}^{2}e^{\alpha_0 G_{*}}
+(\mathcal{A}-\mathcal{B}-\mathcal{D})G_{*}^{2}+c\alpha^{2}G_{*}+(\mathcal{A}-\mathcal{B})\alpha^{2}$.  $H$ is a
function of one variable $G_{*}\in ]0,+\infty[$.  To find the number
of real positive roots of $H(G_{*})$, we will use the intermediate
values theorem  which is generally good for investigating real roots
of  differentiable and monotonous functions.
\par 
\noindent Below, we distinguish several cases.

\noindent \textbf{Case 1:} $\eta_{TG}(\textbf{W})>0$.\\
From the first equation of (\ref{app_eq1}) and from (\ref{app_eq2}) note that
one of the conditions to have a plausible savanna equilibrium is:
\begin{equation}
0<G^* \quad \mbox{and} \quad G^{*}>G_{*}.
\label{app_eq2*}
\end{equation}

In addition, we have

\begin{equation}
\left\{
\begin{array}{lcl}
\lim\limits_{G_{*}\longrightarrow 0}H(G_{*})=(\mathcal{A}-\mathcal{B})\alpha^{2},\\
\lim\limits_{G_{*}\longrightarrow +\infty}H(G_{*})=-\infty.
\end{array}
\right.
\label{app_eq8}
\end{equation}

\comment{
\tcb{When $\eta_{TG}(\textbf{W})<0$ and $\mathcal{R}^2_{\textbf{W}}>1$, one has $(a-b)\alpha^{2}>0$. Hence,
based on the intermediate values theorem for continuous functions,
we deduce that there exists at least one positive zeros of $H$,
named $G_{*0}$. Therefore, there exists at least one savanna
equilibrium $\textbf{E}_{*0}=(G_{*0},T_{*0})$ whenever $G_{*0}\in]0,
G^*[$. Here and in the sequel, we recall that the tree's component
of stated savanna equilibrium or equilibria, is given by
(\ref{app_eq2}). For the rest of the proof, we assume that
$\eta_{TG}(\textbf{W})>0$.}
}

The derivative of $H$ is $H^{'}(G_{*})=3\mathcal{C}G_{*}^{2}-\lambda(\alpha_0
G_{*}^{2}+2G_{*})e^{\alpha_0 G_{*}}+2(\mathcal{A}-\mathcal{B}-\mathcal{D})G_{*}+\mathcal{C}\alpha^{2}.$ We
have

\begin{equation}
\left\{
\begin{array}{lcl}
\lim\limits_{G_{*}\longrightarrow 0}H^{'}(G_{*})=\mathcal{C}\alpha^{2}>0,\\
\lim\limits_{G_{*}\longrightarrow +\infty}H^{'}(G_{*})=-\infty.
\end{array}
\right.
\label{app_eq9}
\end{equation}

Denote by $H^{(2)}$ the derivative of $H^{'}$. We have
$$H^{(2)}(G_{*})=6\mathcal{C}G_{*}-\lambda(\alpha_0^{2} G_{*}^{2}+4\alpha_0
G_{*}+2)e^{\alpha_0 G_{*}}+2(\mathcal{A}-\mathcal{B}-\mathcal{D}).$$ The limits of $H^{(2)}(G_{*})$
at $0$ and $+\infty$ are:

\begin{equation}
\left\{
\begin{array}{lcl}
\lim\limits_{G_{*}\longrightarrow 0}H^{(2)}(G_{*})=2(\mathcal{A}-\mathcal{B}-\mathcal{D}-\lambda),\\
\lim\limits_{G_{*}\longrightarrow +\infty}H^{(2)}(G_{*})=-\infty.
\end{array}
\right.
\label{app_eq10}
\end{equation}

Denote by $H^{(3)}$ the derivative of $H^{(2)}$. We have
$$H^{(3)}(G_{*})=6\mathcal{C}-\lambda(\alpha_0^{3} G_{*}^{2}+6\alpha_0^{2}
G_{*}+6\alpha_0)e^{\alpha_0 G_{*}}$$ and

\begin{equation}
\left\{
\begin{array}{lcl}
\lim\limits_{G_{*}\longrightarrow 0}H^{(3)}(G_{*})=6(\mathcal{C}-\lambda\alpha_0),\\
\lim\limits_{G_{*}\longrightarrow +\infty}H^{(3)}(G_{*})=-\infty.
\end{array}
\right.
\label{app_eq11}
\end{equation}

We have $H^{(4)}(G_{*})=-\lambda(\alpha_0^{4}
G_{*}^{2}+8\alpha_0^{3} G_{*}+12\alpha_0)e^{\alpha_0 G_{*}}<0$. It
implies that $H^{(3)}$ decreases.

\begin{enumerate}
    \item[(I)] If $\mathcal{C}-\lambda\alpha_0\leq0$, then $H^{(3)}\leq0$. It means that $H^{(2)}$ decreases.
    \begin{itemize}
        \item[1)] If $\mathcal{A}-\mathcal{B}-\mathcal{D}-\lambda\leq0$, then $H^{(2)}\leq0$. It implies that $H^{'}$ decreases. Using (\ref{app_eq9}) and the intermediate values theorem, there exists a unique  $G_{*1}\in ]0,+\infty[$ such that $H^{'}(G_{*1})=0$.
        \begin{itemize}
            \item[a)] If $H(G_{*1})<0$, then there is no plausible savanna equilibrium.
            \item[b)] If $H(G_{*1})>0$ and $\mathcal{A}>\mathcal{B}$, then, at most, there exists a unique savanna equilibrium $\textbf{E}_{*}=(G_{*},T_{*})'$
            whenever $G_{*1}< G^{*}$ and, such that $G_{*}\in]G_{*1}, G^{*}[$.
            \item[c)] If $H(G_{*1})>0$ and $\mathcal{A}<\mathcal{B}$, then, at most,  there are two savanna equilibria: $\textbf{E}^{1}_{*}=(G^{1}_{*},T^{1}_{*})'$ and $\textbf{E}^{2}_{*}=(G^{2}_{*},T^{2}_{*})'$
            whenever $G_{*1}< G^{*}$ and, such that $G^{1}_{*}\in]0, G_{*1}[$, $G^{2}_{*}\in]G_{*1}, G^{*}[$.
        \end{itemize}
        \item[2)] If $\mathcal{A}-\mathcal{B}-\mathcal{D}-\lambda>0$, then using (\ref{app_eq10}) and the intermediate values theorem, there exists a unique
        $G_{*2}\in ]0,+\infty[$ such that $H^{(2)}(G_{*2})=0$. From (\ref{app_eq9}) we have $H^{'}(G_{*2})>0$.
        Then using (\ref{app_eq9}) and the intermediate values theorem, there exists a unique  $G_{*3}\in ]G_{*2},+\infty[$
        such that  $H^{'}(G_{*3})=0$. Similarly as in $1)$ we have the following results.
        \begin{itemize}
            \item[a)] If $H(G_{*3})<0$, then there is no plausible savanna equilibrium.
            \item[b)] If $H(G_{*3})>0$ and $\mathcal{A}>\mathcal{B}$, then, at most, there exists a unique savanna equilibrium $\textbf{E}_{**}=(G_{**},T_{**})'$
            whenever $G_{*3}<G^{*}$ and, such that $G_{**}\in]G_{*3}, G^{*}[$.
            \item[c)] The remaining case is $H(G_{*3})>0$ and $\mathcal{A}<\mathcal{B}$. However it
            is unfeasible since $\mathcal{A}-\mathcal{B}-\mathcal{D}-\lambda>0$.
        \end{itemize}
    \end{itemize}
    \item[(II)] If $\mathcal{C}-\lambda\alpha_0>0$,  then using (\ref{app_eq11}) and the intermediate values theorem, there
    exists a unique  $\bar{G}_{*1}\in ]0,+\infty[$ such that $H^{(3)}(\bar{G}_{*1})=0$.
    \begin{itemize}
        \item[1)] If $H^{(2)}(\bar{G}_{*1})<0$, then $H^{(2)}(G_{*})<0$. It implies that $H^{'}$ decreases. Using (\ref{app_eq9}) and the intermediate values theorem, there exists a unique  $\bar{G}_{*2}\in ]0,+\infty[$ such that $H^{'}(\bar{G}_{*2})=0$.
        \begin{itemize}
            \item[a)] If $H(\bar{G}_{*2})<0$, then  there is no plausible savanna equilibrium.
            \item[b)] If $H(\bar{G}_{*2})>0$ and $\mathcal{A}>\mathcal{B}$, then, at most, there exists a unique savanna equilibrium $\bar{\textbf{E}}_{*}=(\bar{G}_{*},\bar{T}_{*})'$
            whenever $\bar{G}_{*2}<G^{*}$ and, such that $\bar{G}_{*}\in]\bar{G}_{*2}, G^{*}[$.
            \item[c)] If $H(\bar{G}_{*2})>0$ and $\mathcal{A}<\mathcal{B}$, then, at most,  there are two savanna equilibria: $\bar{\textbf{E}}^{1}_{*}=(\bar{G}^{1}_{*},\bar{T}^{1}_{*})'$ and $\bar{\textbf{E}}^{2}_{*}=(\bar{G}^{2}_{*},\bar{T}^{2}_{*})'$
            whenever $\bar{G}_{*2}<G^{*}$ and, such that $\bar{G}^{1}_{*}\in]0, \bar{G}_{*2}[$ and $\bar{G}^{2}_{*}\in]\bar{G}_{*2}, G^{*}[$.
        \end{itemize}
        \item[2)] If $H^{(2)}(\bar{G}_{*1})>0$ and $\mathcal{A}-\mathcal{B}-\mathcal{D}-\lambda>0$, then using (\ref{app_eq10}) and the intermediate values theorem, there exists a unique
         $\bar{G}_{*3}\in ]\bar{G}_{*1},+\infty[$ such that $H^{(2)}(\bar{G}_{*3})=0$.  Using (\ref{app_eq9}) there exists a unique $\bar{G}_{*4}\in ]\bar{G}_{*3},+\infty[$ such that $H^{'}(\bar{G}_{*4})=0$.
        \begin{itemize}
            \item[a)] If $H(\bar{G}_{*4})<0$, then  there is no plausible savanna equilibrium.
            \item[b)] If $H(\bar{G}_{*4})>0$ and $\mathcal{A}>\mathcal{B}$, then, at most, there exists a unique savanna equilibrium $\bar{\textbf{E}}_{**}=(\bar{G}_{**},\bar{T}_{**})'$
            whenever $\bar{G}_{*4}<G^{*}$ and, such that $\bar{G}_{**}\in]\bar{G}_{*4}, G^{*}[$.
            \item[c)] The remaining case is $H(\bar{G}_{*4})>0$ and $\mathcal{A}<\mathcal{B}$. However it
            is unfeasible since $\mathcal{A}-\mathcal{B}-\mathcal{D}-\lambda>0$.
        \end{itemize}
        \item[3)] If $H^{(2)}(\bar{G}_{*1})>0$ and $\mathcal{A}-\mathcal{B}-\mathcal{D}-\lambda<0$, then using (\ref{app_eq10}) and the intermediate values theorem there are
         $\bar{G}_{*5}\in]0, \bar{G}_{*1}[$ and $\bar{G}_{*6}\in]\bar{G}_{*1}, +\infty[$ such that $H^{(2)}(\bar{G}_{*5})=0=H^{(2)}(\bar{G}_{*6})$.
        \begin{itemize}
            \item[a)] If $H^{'}(\bar{G}_{*5})>0$ and $H^{'}(\bar{G}_{*6})>0$, then using (\ref{app_eq9}) and the intermediate value theorem there exists a unique $\bar{G}_{*7}\in]\bar{G}_{*6}, +\infty[$ such that $H^{'}(\bar{G}_{*7})=0$.
            \begin{itemize}
                \item[1.] If  $H(\bar{G}_{*7})<0$, then there is no plausible savanna equilibrium.
                \item[2.] If $H(\bar{G}_{*7})>0$ and $\mathcal{A}>\mathcal{B}$, then, at most, there exists a unique savanna equilibrium $\bar{\textbf{E}}_{***}=(\bar{G}_{***},\bar{T}_{***})'$
                whenever $\bar{G}_{*7}<G^{*}$ and, such that $\bar{G}_{***}\in]\bar{G}_{*7}, G^{*}[$.
                \item[3.] If $H(\bar{G}_{*7})>0$ and $\mathcal{A}<\mathcal{B}$, then, at most,  there are two savanna equilibria: $\bar{\textbf{E}}^{1}_{***}=(\bar{G}^{1}_{***},\bar{T}^{1}_{***})'$ and $\bar{\textbf{E}}^{2}_{***}=(\bar{G}^{2}_{***},\bar{T}^{2}_{***})'$
                whenever $\bar{G}_{*7}<G^{*}$ and, such that $\bar{G}^{1}_{***}\in]0, \bar{G}_{*7}[$, $\bar{G}^{2}_{***}\in]\bar{G}_{*7}, G^{*}[$.
            \end{itemize}

            \item[b)] If $H^{'}(\bar{G}_{*5})<0$ and $H^{'}(\bar{G}_{*6})>0$, then using (\ref{app_eq9}) and the intermediate value theorem
            there are $\bar{G}_{*8}\in]0, \bar{G}_{*5}[$, $\bar{G}_{*9}\in]\bar{G}_{*5}, \bar{G}_{*6}[$ and $\bar{G}_{*10}\in]\bar{G}_{*6}, +\infty[$
            such that
            $H^{'}(\bar{G}_{*8})=H^{'}(\bar{G}_{*9})=H^{'}(\bar{G}_{*10})=0$. Based on (\ref{app_eq8}), one deduces that $\bar{G}_{*8}$
            and $\bar{G}_{*10}$ are two local maxima while $\bar{G}_{*9}$ is a local minimum.
            Once more,  using (\ref{app_eq8}) and the intermediate values theorem we have:
            \begin{itemize}
                \item[1.] If
                $\max(H(\bar{G}_{*8},H(\bar{G}_{*10})))<0$, then
                there is no plausible savanna equilibrium.
                %
                \item[2.] If $H(\bar{G}_{*9})>0$ and $\mathcal{A}<\mathcal{B}$, then, at most, there exist two savanna equilibria: $\bar{\textbf{E}}^{1}_{****}=(\bar{G}^{1}_{****},\bar{T}^{1}_{****})'$ and $\bar{\textbf{E}}^{2}_{****}=(\bar{G}^{2}_{****},\bar{T}^{2}_{****})'$
                whenever $\bar{G}_{*10}<G^{*}$ and, such that $\bar{G}^{1}_{****}\in]0, \bar{G}_{*8}[$, $\bar{G}^{2}_{****}\in]\bar{G}_{*10}, G^{*}[$.
                %
                \item[3.] If $H(\bar{G}_{*9})>0$ and $\mathcal{A}>\mathcal{B}$, then, at most,  there is a unique savanna equilibrium  $\bar{\textbf{E}}_{****}=(\bar{G}_{****},\bar{T}_{****})'$
                whenever $\bar{G}_{*10}<G^{*}$ and, such that $\bar{G}_{****}\in]\bar{G}_{*10}, G^{*}[$.
                %
                \item[4.] If $H(\bar{G}_{*8})<0$ and
                $H(\bar{G}_{*10})>0$, then, at most, there exist two savanna equilibria: $\bar{\textbf{E}}^{1}_{****}=(\bar{G}^{1}_{****},\bar{T}^{1}_{****})'$ and $\bar{\textbf{E}}^{2}_{****}=(\bar{G}^{2}_{****},\bar{T}^{2}_{****})'$
                whenever $\bar{G}_{*10}<G^{*}$ and, such that $\bar{G}^{1}_{****}\in]\bar{G}_{*9}, \bar{G}_{*10}[$, $\bar{G}^{2}_{****}\in]\bar{G}_{*10}, G^{*}[$.
                %
                \item[5.] If $H(\bar{G}_{*8})>0$, $H(\bar{G}_{*10})<0$ and
                $\mathcal{A}<\mathcal{B}$, then, at most, there exist two savanna equilibria: $\bar{\textbf{E}}^{1}_{****}=(\bar{G}^{1}_{****},\bar{T}^{1}_{****})'$ and $\bar{\textbf{E}}^{2}_{****}=(\bar{G}^{2}_{****},\bar{T}^{2}_{****})'$
                whenever $\bar{G}_{*9}<G^{*}$ and, such that $\bar{G}^{1}_{****}\in]0, \bar{G}_{*8}[$, $\bar{G}^{2}_{****}\in]\bar{G}_{*8}, \bar{G}_{*9}[$.
                %
                \item[6.] If $H(\bar{G}_{*8})>0$, $H(\bar{G}_{*10})<0$ and
                $\mathcal{A}>\mathcal{B}$, then, at most, there exist a unique savanna equilibrium: $\bar{\textbf{E}}_{****}=(\bar{G}_{****},\bar{T}_{****})'$
                whenever $\bar{G}_{*9}<G^{*}$ and, such that $\bar{G}_{****}\in]\bar{G}_{*8}, \bar{G}_{*9}[$.
                \item[7.] If $\min(H(\bar{G}_{*8}),H(\bar{G}_{*10}))>0$, $H(\bar{G}_{*9})<0$ and
                $\mathcal{A}<\mathcal{B}$, then, at most, there are four savanna equilibria: $\bar{\textbf{E}}^{1}_{*****}=(\bar{G}^{1}_{*****},\bar{T}^{1}_{*****})'$,
                $\bar{\textbf{E}}^{2}_{*****}=(\bar{G}^{2}_{****},\bar{T}^{2}_{*****})'$,
                $\bar{\textbf{E}}^{3}_{*****}=(\bar{G}^{3}_{****},\bar{T}^{3}_{*****})'$
                and  $\bar{\textbf{E}}^{4}_{*****}=(\bar{G}^{4}_{****},\bar{T}^{4}_{*****})'$
                whenever $\bar{G}_{*10}<G^{*}$ and, such that $\bar{G}^{1}_{*****}\in]0, \bar{G}_{*8}[$, $\bar{G}^{2}_{****}\in]\bar{G}_{*8}, \bar{G}_{*9}[$,
                $\bar{G}^{3}_{****}\in]\bar{G}_{*9}, \bar{G}_{*10}[$ and $\bar{G}^{4}_{****}\in]\bar{G}_{*10}, G^{*}[$.
                \item[8.] If $\min(H(\bar{G}_{*8}),H(\bar{G}_{*10}))>0$, $H(\bar{G}_{*9})<0$ and
                $\mathcal{A}>\mathcal{B}$, then, at most, there are three savanna equilibria: $\bar{\textbf{E}}^{1}_{*****}=(\bar{G}^{1}_{*****},\bar{T}^{1}_{*****})'$,
                $\bar{\textbf{E}}^{2}_{*****}=(\bar{G}^{2}_{****},\bar{T}^{2}_{*****})'$
                and
                $\bar{\textbf{E}}^{3}_{*****}=(\bar{G}^{3}_{****},\bar{T}^{3}_{*****})'$
                whenever $\bar{G}_{*10}<G^{*}$ and, such that $\bar{G}^{1}_{****}\in]\bar{G}_{*8}, \bar{G}_{*9}[$,
                $\bar{G}^{2}_{****}\in]\bar{G}_{*9}, \bar{G}_{*10}[$ and $\bar{G}^{3}_{****}\in]\bar{G}_{*10}, G^{*}[$.
            \end{itemize}
            \item[c)] If $H^{'}(\bar{G}_{*5})<0$ and $H^{'}(\bar{G}_{*6})<0$, then using (\ref{app_eq9}) and the
            intermediate values theorem there exists a unique $\bar{G}_{*11}\in]0, G_{*5}[$ such that $H^{'}(\bar{G}_{*11})=0.$ Using (\ref{app_eq8})     and the intermediate value theorem we have:
            \begin{itemize}
                \item[1.] If $H(\bar{G}_{*11})<0$, then  there is no plausible savanna equilibrium.
                \item[2.] If $H(\bar{G}_{*11})>0$ and $\mathcal{A}>\mathcal{B}$, then, at most, there exists a unique savanna equilibrium $\bar{\textbf{E}}=(\bar{G},\bar{T})'$
                whenever $\bar{G}_{*11}<G^{*}$ and, such that $\bar{G}\in]\bar{G}_{*11}, G^{*}[$.
                \item[3.] If $H(\bar{G}_{*11})>0$ and $\mathcal{A}<\mathcal{B}$, then, at most,  there are two savanna equilibria: $\bar{\textbf{E}}^{1}=(\bar{G}^{1},\bar{T}^{1})'$
                and $\bar{\textbf{E}}^{2}=(\bar{G}^{2},\bar{T}^{2})'$
                whenever $\bar{G}_{*11}<G^{*}$ and, such that $\bar{G}^{1}\in]0, \bar{G}_{*11}[$ and $\bar{G}^{2}\in]\bar{G}_{*11},
                G^{*}[$.
            \end{itemize}
        \end{itemize}
    \end{itemize}
\end{enumerate}
This ends the case $\eta_{TG}(\textbf{W})>0$ or the competition case. In the sequel, we assume that $\eta_{TG}(\textbf{W})<0$; that is the facilitation case. 

\vspace{1cm}

\noindent \textbf{Case 2:} $\eta_{TG}(\textbf{W})<0$.\\
Recall that the tree component's of a savanna equilibrium is given by (\ref{app_eq2}). Hence, in the sequel, a plausible savanna equilibrium is given by a positive $G_*$ which is a zero of the function $H$ and which is such that $T_*$ defined by (\ref{app_eq2}) is positive.
In this case, one has $\mathcal{A}>0$, $\mathcal{C}<0$, $\mathcal{D}>0$, $\lambda>0$ and $\alpha_0<0$.

Let us set  $K(G)=\alpha_0^4G^2+8\alpha_0^3G+12\alpha_0$ such that $H^{(4)}(G_*)=-\lambda K(G_*) e^{\alpha_0G_*}$. One has $K''(G)=2\alpha_0^4>0$ and $K'(0)=8\alpha_0^3<0$. Hence, there exists a unique $\G_{1*}\in \R_+$ such that $K'(\G_{1*})=0$ and $K$ is decreasing on $[0,\G_{1*}]$ and, $K$ is increasing on $[\G_{1*}, +\infty)$. Since $K(0)=12\alpha_0<0$ and $\lim\limits_{G\rightarrow+\infty}K(G)=+\infty$, there exists a unique $\G_{1**}\in(\G_{1*},+\infty)$ such that $K(\G_{1**})=0$. Thus, $K(G)\leq0$ on $[0,\G_{1**}]$ and $K(G)>0$ on $[\G_{1**},+\infty)$. In other words,  $H^{(4)}(G_*)\geq0$ on $[0,\G_{1**}]$ and $H^{(4)}(G_*)<0$ on $[\G_{1**},+\infty)$. Hence, $H^{(3)}$ is increasing on $[0,\G_{1**}]$ and $H^{(3)}$ is decreasing on $[\G_{1**},+\infty)$. One has $H^{(3)}(0)=6(\mathcal{C}-\lambda\alpha_0)$ and $\lim\limits_{G\rightarrow+\infty}H^{(3)}(G)=6\mathcal{C}<0.$

\begin{enumerate}
    \item[(I)]Assume that $H^{(3)}(\G_{1**})\leq0$.
    \begin{enumerate}
        \item[1)] Assume that $H^{(2)}(0)=2(\mathcal{A}-\mathcal{B}-\mathcal{D}-\lambda)\leq0$. Since $H^{(1)}(0)=\mathcal{C}\alpha^2<0$ and $\lim\limits_{G\rightarrow+\infty}H^{(1)}(G)=-\infty$, then $H^{(1)}(G)<0$ on $\R_+$; i.e. $H$ is decreasing on $\R_+$. 
        \begin{enumerate}
            \item[a)] If $\mathcal{A}-\mathcal{B}<0$ i.e. $H(0)=(\mathcal{A}-\mathcal{B})\alpha^2<0$, then no plausible savanna equilibria exist.
            \item[b)] If $\mathcal{A}-\mathcal{B}>0$ i.e. $H(0)=(\mathcal{A}-\mathcal{B})\alpha^2>0$, then there exists a unique $G_{*1}\in[0,+\infty)$ such that $H(G_{*1}=0$. Hence, there exists at most one savanna equilibrium $\textbf{E}_*=(G_{*1},T_{*1})'$ whenever $T_{*1}>0$, where $T_{*1}$ is computed from (\ref{app_eq2}).
        \end{enumerate}
        \item[2)] Assume that $2(\mathcal{A}-\mathcal{B}-\mathcal{D}-\lambda)>0$. Note that, in this case, $\mathcal{A}-\mathcal{B}>0$. Then, there exists a unique $\G_{3*}\in\R_+$ such that $H^{(2)}(\G_{3*})=0$, $H^{(1)}$ is increasing on $[0, \G_{3*}]$ and is decreasing on $(\G_{3*},+\infty)$.  Since $H^{(1)}(0)=\mathcal{C}\alpha^2<0$ and $\lim\limits_{G\rightarrow+\infty}H^{(1)}(G)=-\infty$, we have two sub-cases.
        \begin{enumerate}
            \item[a)] Assume that $H^{(1)}(\G_{3*})\leq0$. Since $H(0)=(\mathcal{A}-\mathcal{B})\alpha^2>0$ and $\lim\limits_{G\rightarrow+\infty}H(G)=-\infty$, then there exists a unique $G_{*1}\in[0,+\infty)$ such that $H(G_{*1}=0$. Hence, there exists at most one savanna equilibrium $\textbf{E}_*=(G_{*1},T_{*1})'$ whenever $T_{*1}>0$.
            \item[b)] Assume that $H^{(1)}(\G_{3*})>0$. Then, there exist $\G_{3**}\in(0,\G_{3*})$ and $\G_{3***}\in(\G_{3*},+\infty)$ that are zeros of $H^{(1)}$.
            \begin{enumerate}
                \item[i)] If $\min(H(\G_{3**}), H(\G_{3***}))>0$ then there exists at most one savanna equilibrium $\textbf{E}_*=(G_{*1},T_{*1})'$ whenever $G_{*1}\in(\G_{3***},+\infty)$, $H(G_{*1})=0$ and $T_{*1}>0$.
                \item[ii)] If $H(\G_{3**})<0$ and $H(\G_{3***})>0$, then there exist at most three savanna equilibria $\textbf{E}_*^i=(G_{*i},T_{*i})'$ whenever $G_{*1}\in(0,\G_{3**})$, $G_{*2}\in(\G_{3**},\G_{3***})$, $G_{*3}\in(\G_{3***},+\infty)$, $H(G_{*i})=0$ and $T_{*i}>0$, $i=1,2,3$.
                \item[iii)] If $\max(H(\G_{3**}), H(\G_{3***}))<0$ then there exists at most one savanna equilibrium $\textbf{E}_*=(G_{*1},T_{*1})'$ whenever $G_{*1}\in(0,\G_{3**})$, $H(G_{*1})=0$ and $T_{*1}>0$.
            \end{enumerate}
        \end{enumerate}
    \end{enumerate}
    \item[(II)] Assume that $H^{(3)}(\G_{1**})>0$ and $\mathcal{C}-\lambda\alpha_0>0$. Then there exists a unique $\G_{1***}\in(\G_{1**},+\infty)$, zero of $H^{(3)}$.
    \begin{enumerate}
        \item[1)] Assume that $H^{(2)}(\G_{1***})\leq0$. Since $H^{(1)}(0)=\mathcal{C}\alpha^2<0$, then $H^{(1)}(G)<0$ on $\R_+$.
        \begin{enumerate}
            \item[a)] If $\mathcal{A}-\mathcal{B}<0$, then no plausible savanna equilibria exist.
            \item[b)] If $\mathcal{A}-\mathcal{B}>0$, then there exists a unique $G_{*1}\in[0,+\infty)$ such that $H(G_{*1}=0$. Hence, there exists at most one savanna equilibrium $\textbf{E}_*=(G_{*1},T_{*1})'$ whenever $T_{*1}>0$.
        \end{enumerate}
        \item[2)] Assume that $H^{(2)}(\G_{1***})>0$ and $\mathcal{A}-\mathcal{B}-\mathcal{D}-\lambda>0$. Then there exists a unique $\G_{4*}\in(\G_{1***},+\infty)$ such that $H^{(2)}(\G_{4*})=0$. One has two sub-cases.
        \begin{enumerate}
            \item[a)] Assume that $H^{(1)}(\G_{4*})\leq0$. Since $H(0)=(\mathcal{A}-\mathcal{B})\alpha^2>0$ and $\lim\limits_{G\rightarrow+\infty}H(G)=-\infty$, then there exists a unique $G_{*1}\in[0,+\infty)$ such that $H(G_{*1}=0$. Hence, there exists at most one savanna equilibrium $\textbf{E}_*=(G_{*1},T_{*1})'$ whenever $T_{*1}>0$.
            \item[b)] Assume that $H^{(1)}(\G_{4*})>0$. Then, there exist $\G_{5*}\in(0,\G_{4*})$ and $\G_{5**}\in(\G_{4*},+\infty)$ that are zeros of $H^{(1)}$.
            \begin{enumerate}
                \item[i)] If $\min(H(\G_{5*}), H(\G_{5**}))>0$ then there exists at most one savanna equilibrium $\textbf{E}_*=(G_{*1},T_{*1})'$ whenever $G_{*1}\in(\G_{5**},+\infty)$, $H(G_{*1})=0$ and $T_{*1}>0$.
                \item[ii)] If $H(\G_{5*})<0$ and $H(\G_{5**})>0$, then there exist at most three savanna equilibria $\textbf{E}_*^i=(G_{*i},T_{*i})'$ whenever $G_{*1}\in(0,\G_{5*})$, $G_{*2}\in(\G_{5*},\G_{5**})$, $G_{*3}\in(\G_{5**},+\infty)$, $T_{*i}>0$ and $H(G_{*i})=0$, $i=1,2,3$.
                \item[iii)] If $\max(H(\G_{5*}), H(\G_{5**}))<0$ then there exists at most one savanna equilibrium $\textbf{E}_*=(G_{*1},T_{*1})'$ whenever $G_{*1}\in(0,\G_{5*})$, $H(G_{*1})=0$ and $T_{*1}>0$.
            \end{enumerate} 
        \end{enumerate}
        \item[3)] Assume that $H^{(2)}(\G_{1***})>0$ and $\mathcal{A}-\mathcal{B}-\mathcal{D}-\lambda<0$. Then there exist $\G_{4**}\in(0,\G_{1***})$ and  $\G_{4***}\in(\G_{1***},+\infty)$ such that $H^{(2)}(\G_{4**})=H^{(2)}(\G_{4***})=0$. One has three sub-cases.
        \begin{enumerate}
            \item[a)] Assume that $H^{(1)}(\G_{4***})\leq0$. Then $H^{(1)}(G)\leq0$ on $\R_+$. Note that $\lim\limits_{G\rightarrow+\infty}H(G)=-\infty$.
            \begin{enumerate}
                \item[i)] If $\mathcal{A}-\mathcal{B}<0$, then no plausible savanna equilibria exist.
                \item[ii)] If $\mathcal{A}-\mathcal{B}>0$, then there exists a unique $G_{*1}\in[0,+\infty)$ such that $H(G_{*1}=0$. Hence, there exists at most one savanna equilibrium $\textbf{E}_*=(G_{*1},T_{*1})'$ whenever $T_{*1}>0$.
            \end{enumerate}
            \comment{
            Since $H(0)=(a-b)\alpha^2>0$ and $\lim\limits_{G\rightarrow+\infty}H(G)=-\infty$, then there exists a unique $G_{*1}\in[0,+\infty)$ such that $H(G_{*1}=0$. Hence, there exists at most one savanna equilibrium $\textbf{E}_*=(G_{*1},T_{*1})$ whenever $G_{*1}>G^*$.
            }
            \item[b)] Assume that $H^{(1)}(\G_{4***})>0$ and $\mathcal{A}-\mathcal{B}>0$. Then, there exist $\G_{6*}\in(\G_{4**}, \G_{4***})$ and $\G_{6**}\in(\G_{4***},+\infty)$ that are zeros of $H^{(1)}$.
            \begin{enumerate}
                \item[i)] If $\min(H(\G_{6*}), H(\G_{6**}))>0$ then there exists at most one savanna equilibrium $\textbf{E}_*=(G_{*1},T_{*1})'$ whenever $G_{*1}\in(\G_{6**},+\infty)$, $H(G_{*1})=0$ and $T_{*1}>0$.
                \item[ii)] If $H(\G_{6*})<0$ and $H(\G_{6**})>0$, then there exist at most three savanna equilibria $\textbf{E}_*^i=(G_{*i},T_{*i})'$ whenever $G_{*1}\in(0,\G_{6*})$, $G_{*2}\in(\G_{6*},\G_{6**})$, $G_{*3}\in(\G_{6**},+\infty)$, $H(G_{*i})=0$ and $T_{*i}>0$, $i=1,2,3$.
                \item[iii)] If $\max(H(\G_{6*}), H(\G_{6**}))<0$ then there exists at most one savanna equilibrium $\textbf{E}_*=(G_{*1},T_{*1})'$ whenever $G_{*1}\in(0,\G_{6*})$, $H(G_{*1})=0$ and $T_{*1}>0$.
            \end{enumerate}
            \item[c)] Assume that $H^{(1)}(\G_{4***})>0$ and $\mathcal{A}-\mathcal{B}<0$. Then, there exist $\G_{6*}\in(\G_{4**}, \G_{4***})$ and $\G_{6**}\in(\G_{4***},+\infty)$ that are zeros of $H^{(1)}$.
            \begin{enumerate}
                \item[i)] If $H(\G_{6**})<0$, then there is no plausible savanna equilibria.
                \item[ii)] If $H(\G_{6**})>0$, then there exist at most two savanna equilibria $\textbf{E}_*^i=(G_{*i},T_{*i})'$ whenever $G_{*1}\in(\G_{6*},\G_{6**})$, $G_{*2}\in(\G_{6**},+\infty)$, $H(G_{*i})=0$, $T_{*i}>0$, $i=1,2$.
            \end{enumerate}
        \end{enumerate}
    \end{enumerate}
    \item[(III)] Assume that $H^{(3)}(\G_{1**})>0$ and $\mathcal{C}-\lambda\alpha_0<0$. Then there exist $\G_{2*}\in(0,\G_{1**})$ and $\G_{2**}\in(\G_{1**},+\infty)$, zeros of $H^{(3)}$. We have five sub-cases.
    \begin{enumerate}
        \item[1)] Assume that $H^{(2)}(0)=\mathcal{A}-\mathcal{B}-\mathcal{D}-\lambda>0$ and $\min(H^{(2)}(\G_{2*}),H^{(2)}(\G_{2**}))>0$. Then there exists a unique $\G_{7*}\in(\G_{2**},+\infty)$ such that $H^{(2)}(\G_{7*})=0$. One has two sub-cases.
        \begin{enumerate}
            \item[a)] Assume that $H^{(1)}(\G_{7*})\leq0$. Since $\mathcal{A}-\mathcal{B}>0$, then there exists a unique $G_{*1}\in[0,+\infty)$ such that $H(G_{*1})=0$. Hence, there exists at most one savanna equilibrium $\textbf{E}_*=(G_{*1},T_{*1})'$ whenever $T_{*1}>0$.
\comment{
            \begin{enumerate}
                \item[i)] If $a-b<0$, then no plausible savanna equilibria exist.
                \item[ii)] If $a-b>0$, then there exists a unique $G_{*1}\in[0,+\infty)$ such that $H(G_{*1}=0$. Hence, there exists at most one savanna equilibrium $\textbf{E}_*=(G_{*1},T_{*1})$ whenever $T_{*1}>0$.
            \end{enumerate}
}
\comment{
            Since $H(0)=(a-b)\alpha^2>0$ and $\lim\limits_{G\rightarrow+\infty}H(G)=-\infty$, then there exists a unique $G_{*1}\in[0,+\infty)$ such that $H(G_{*1}=0$. Hence, there exists at most one savanna equilibrium $\textbf{E}_*=(G_{*1},T_{*1})$ whenever $G_{*1}>G^*$.
}
            \item[b)] Assume that $H^{(1)}(\G_{7*})>0$. Then, there exist $\G_{11*}\in(0, \G_{7*})$ and $\G_{11**}\in(\G_{7*},+\infty)$ that are zeros of $H^{(1)}$.
            \begin{enumerate}
                \item[i)] If $\min(H(\G_{11*}), H(\G_{11**}))>0$ then there exists at most one savanna equilibrium $\textbf{E}_*=(G_{*1},T_{*1})'$ whenever $G_{*1}\in(\G_{11**},+\infty)$, $H(G_{*1})=0$ and $T_{*1}>0$.
                \item[ii)] If $H(\G_{11*})<0$ and $H(\G_{11**})>0$, then there exist at most three savanna equilibria $\textbf{E}_*^i=(G_{*i},T_{*i})'$ whenever $G_{*1}\in(0,\G_{11*})$, $G_{*2}\in(\G_{11*},\G_{11**})$, $G_{*3}\in(\G_{11**},+\infty)$, $T_{*i}>0$ and $H(G_{*i})=0$, $i=1,2,3$.
                \item[iii)] If $\max(H(\G_{11*}), H(\G_{11**}))<0$ then there exists at most one savanna equilibrium $\textbf{E}_*=(G_{*1},T_{*1})'$ whenever $G_{*1}\in(0,\G_{11*})$, $H(G_{*1})=0$ and $T_{*1}>0$.
            \end{enumerate}
        \end{enumerate}
        \item[2)] Assume that $\mathcal{A}-\mathcal{B}-\mathcal{D}-\lambda>0$, $H^{(2)}(\G_{2*})<0$ and $H^{(2)}(\G_{2**})>0$. Then there exist $\G_{8*}\in(0,\G_{2*})$, $\G_{8**}\in(\G_{2*},\G_{2**})$ and $\G_{8***}\in(\G_{2**},+\infty)$ such that $H^{(2)}(\G_{8*})=H^{(2)}(\G_{8**})=H^{(2)}(\G_{8***})=0$. One has five sub-cases.
        \begin{itemize}
            \item[a)] Assume that $\max(H^{(1)}(\G_{8*}),H^{(1)}(\G_{8***}))\leq0$. Since $H(0)=(\mathcal{A}-\mathcal{B})\alpha^2>0$ and $\lim\limits_{G\rightarrow+\infty}H(G)=-\infty$, then there exists a unique $G_{*1}\in[0,+\infty)$ such that $H(G_{*1}=0$. Hence, there exists at most one savanna equilibrium $\textbf{E}_*=(G_{*1},T_{*1})'$ whenever $T_{*1}>0$.
            \item[b)] Assume that $H^{(1)}(\G_{8**})>0$. Then, there exist $\G_{12*}\in(0, \G_{8*})$ and $\G_{12**}\in(\G_{8***},+\infty)$ that are zeros of $H^{(1)}$.
            \begin{enumerate}
                \item[i)] If $\min(H(\G_{12*}), H(\G_{12**}))>0$ then there exists at most one savanna equilibrium $\textbf{E}_*=(G_{*1},T_{*1})'$ whenever $G_{*1}\in(\G_{12**},+\infty)$, $H(G_{*1})=0$ and $T_{*1}>0$.
                \item[ii)] If $H(\G_{12*})<0$ and $H(\G_{12**})>0$, then there exist at most three savanna equilibria $\textbf{E}_*^i=(G_{*i},T_{*i})'$ whenever $G_{*1}\in(0,\G_{12*})$, $G_{*2}\in(\G_{12*},\G_{12**})$, $G_{*3}\in(\G_{12**},+\infty)$, $T_{*i}>0$ and $H(G_{*i})=0$, $i=1,2,3$.
                \item[iii)] If $\max(H(\G_{12*}), H(\G_{12**}))<0$ then there exists at most one savanna equilibrium $\textbf{E}_*=(G_{*1},T_{*1})'$ whenever $G_{*1}\in(0,\G_{12*})$, $H(G_{*1})=0$ and $T_{*1}>0$.
          \end{enumerate}
          \item[c)] Assume that $H^{(1)}(\G_{8*})<0$ and $H^{(1)}(\G_{8***})>0$. Then, there exist $\G_{13*}\in(\G_{8**}, \G_{8***})$ and $\G_{13**}\in(\G_{8***},+\infty)$ that are zeros of $H^{(1)}$.
            \begin{enumerate}
                \item[i)] If $\min(H(\G_{13*}), H(\G_{13**}))>0$ then there exists at most one savanna equilibrium $\textbf{E}_*=(G_{*1},T_{*1})'$ whenever $G_{*1}\in(\G_{13**},+\infty)$, $H(G_{*1})=0$ and $T_{*1}>0$.
                \item[ii)] If $H(\G_{13*})<0$ and $H(\G_{13**})>0$, then there exist at most three savanna equilibria $\textbf{E}_*^i=(G_{*i},T_{*i})'$ whenever $G_{*1}\in(0,\G_{13*})$, $G_{*2}\in(\G_{13*},\G_{13**})$, $G_{*3}\in(\G_{13**},+\infty)$, $T_{*i}>0$ and $H(G_{*i})=0$, $i=1,2,3$.
                \item[iii)] If $\max(H(\G_{13*}), H(\G_{13**}))<0$ then there exists at most one savanna equilibrium $\textbf{E}_*=(G_{*1},T_{*1})'$ whenever $G_{*1}\in(0,\G_{13*})$, $H(G_{*1})=0$ and $T_{*1}>0$.
          \end{enumerate}
          \item[d)] Assume that $H^{(1)}(\G_{8*})>0$ and $H^{(1)}(\G_{8***})<0$. Then, there exist $\G_{14*}\in(0, \G_{8*})$ and $\G_{14**}\in(\G_{8*},\G_{8**})$ that are zeros of $H^{(1)}$.
            \begin{enumerate}
                \item[i)] If $\min(H(\G_{14*}), H(\G_{14**}))>0$ then there exists at most one savanna equilibrium $\textbf{E}_*=(G_{*1},T_{*1})'$ whenever $G_{*1}\in(\G_{14**},+\infty)$, $H(G_{*1})=0$ and $T_{*1}>0$.
                \item[ii)] If $H(\G_{14*})<0$ and $H(\G_{14**})>0$, then there exist at most three savanna equilibria $\textbf{E}_*^i=(G_{*i},T_{*i})'$ whenever $G_{*1}\in(0,\G_{14*})$, $G_{*2}\in(\G_{14*},\G_{14**})$, $G_{*3}\in(\G_{14**},+\infty)$, $T_{*i}>0$ and $H(G_{*i})=0$, $i=1,2,3$.
                \item[iii)] If $\max(H(\G_{14*}), H(\G_{14**}))<0$ then there exists at most one savanna equilibrium $\textbf{E}_*=(G_{*1},T_{*1})'$ whenever $G_{*1}\in(0,\G_{14*})$, $H(G_{*1})=0$ and $T_{*1}>0$.
          \end{enumerate}
          \item[e)] Assume that $\min(H^{(1)}(\G_{8*}), H^{(1)}(\G_{8***}))>0$ and $H^{(1)}(\G_{8**})<0$. Then, there exist $\G_{15*}\in(0, \G_{8*})$,  $\G_{15**}\in(\G_{8*},\G_{8**})$, $\G_{15***}\in(\G_{8**},\G_{8***})$ and $\G_{15****}\in(\G_{8***},+\infty)$ that are zeros of $H^{(1)}$.
            \begin{enumerate}
                \item[i)] If $\min(H(\G_{15*}), H(\G_{15***}))>0$ then there exists at most one savanna equilibrium $\textbf{E}_*=(G_{*1},T_{*1})'$ whenever $G_{*1}\in(\G_{15****},+\infty)$, $H(G_{*1})=0$ and $T_{*1}>0$.
                \item[ii)] If $H(\G_{15*})>0$, $H(\G_{15***})<0$ and $H(\G_{15****})>0$, then there exist at most three savanna equilibria $\textbf{E}_*^i=(G_{*i},T_{*i})'$ whenever $G_{*1}\in(\G_{15**},\G_{15***})$, $G_{*2}\in(\G_{15***},\G_{15****})$, $G_{*3}\in(\G_{15****},+\infty)$, $T_{*i}>0$ and $H(G_{*i})=0$, $i=1,2,3$.
                \item[iii)] If $H(\G_{15*})>0$ and $H(\G_{15****})<0$ then there exists at most one savanna equilibrium $\textbf{E}_*=(G_{*1},T_{*1})'$ whenever $G_{*1}\in(\G_{15**},\G_{15***})$, $H(G_{*1})=0$ and $T_{*1}>0$.
                \item[iv)] If $H(\G_{15**})<0$, $H(\G_{15****})<0$ then there exists at most one savanna equilibrium $\textbf{E}_*=(G_{*1},T_{*1})'$ whenever $G_{*1}\in(0, \G_{15*})$, $H(G_{*1})=0$ and $G_{*1}>G^*$.
                \item[v)] If $H(\G_{15*})<0$ and $H(\G_{15***})>0$, then there exist at most two savanna equilibria $\textbf{E}_*^i=(G_{*i},T_{*i})'$ whenever $G_{*1}\in(0, \G_{15*})$, $G_{*2}\in(\G_{15*},\G_{15**})$, $T_{*i}>0$ and $H(G_{*i})=0$, $i=1,2$.
                \item[vi)] If $H(\G_{15**})<0$ and $H(\G_{15****})>0$, then there exist at most two savanna equilibria $\textbf{E}_*^i=(G_{*i},T_{*i})'$ whenever $G_{*1}\in(\G_{15***},\G_{15****})$, $G_{*2}\in(\G_{15****},+\infty)$, $T_{*i}>0$ and $H(G_{*i})=0$, $i=1,2$.
                \item[vii)] If $\max(H(\G_{15*}),H(\G_{15***}))<0$ and $\min(H(\G_{15**}),H(\G_{15****}))>0$, then there exist at most five savanna equilibria $\textbf{E}_*^i=(G_{*i},T_{*i})'$ whenever $G_{*1}\in(0,\G_{15*})$, $G_{*2}\in(\G_{15*},\G_{15**})$, $G_{*3}\in(\G_{15**},\G_{15***})$, $G_{*4}\in(\G_{15***},\G_{15****})$ and $G_{*5}\in(\G_{15****},+\infty)$, $T_{*i}>0$ and $H(G_{*i})=0$, $i=1,2,3,4,5$.
                \item[viii)] If $\max(H(\G_{15*}),H(\G_{15****}))<0$ and $H(\G_{15**})>0$, then there exist at most three savanna equilibria $\textbf{E}_*^i=(G_{*i},T_{*i})'$ whenever $G_{*1}\in(0,\G_{15*})$, $G_{*2}\in(\G_{15*},\G_{15**})$ and $G_{*3}\in(\G_{15**},\G_{15***})$, $T_{*i}>0$ and $H(G_{*i})=0$, $i=1,2,3$.
          \end{enumerate}
        \end{itemize}
        \item[3)] Assume that $\mathcal{A}-\mathcal{B}-\mathcal{D}-\lambda>0$, $\max(H^{(2)}(\G_{2*}), H^{(2)}(\G_{2**}))<0$. Then there exists a unique $\G_{9*}\in(0,\G_{2*})$ such that $H^{(2)}(\G_{9*})=0$. One has two sub-cases.
        \begin{enumerate}
            \item[a)] Assume that $H^{(1)}(\G_{9*})\leq0$. Since $H(0)=(\mathcal{A}-\mathcal{B})\alpha^2>0$ and $\lim\limits_{G\rightarrow+\infty}H(G)=-\infty$, then there exists a unique $G_{*1}\in[0,+\infty)$ such that $H(G_{*1}=0$. Hence, there exists at most one savanna equilibrium $\textbf{E}_*=(G_{*1},T_{*1})'$ whenever $T_{*1}>0$.
            \item[b)] Assume that $H^{(1)}(\G_{9*})>0$. Then, there exist $\G_{16*}\in(0, \G_{9*})$ and $\G_{16**}\in(\G_{9*},+\infty)$ that are zeros of $H^{(1)}$.
            \begin{enumerate}
                \item[i)] If $\min(H(\G_{16*}), H(\G_{16**}))>0$ then there exists at most one savanna equilibrium $\textbf{E}_*=(G_{*1},T_{*1})'$ whenever $G_{*1}\in(\G_{16**},+\infty)$, $H(G_{*1})=0$ and $T_{*1}>0$.
                \item[ii)] If $H(\G_{16*})<0$ and $H(\G_{16**})>0$, then there exist at most three savanna equilibria $\textbf{E}_*^i=(G_{*i},T_{*i})'$ whenever $G_{*1}\in(0,\G_{16*})$, $G_{*2}\in(\G_{16*},\G_{16**})$, $G_{*3}\in(\G_{16**},+\infty)$, $T_{*i}>0$ and $H(G_{*i})=0$, $i=1,2,3$.
                \item[iii)] If $\max(H(\G_{16*}), H(\G_{16**}))<0$ then there exists at most one savanna equilibrium $\textbf{E}_*=(G_{*1},T_{*1})'$ whenever $G_{*1}\in(0,\G_{16*})$, $H(G_{*1})=0$ and $T_{*1}>0$.
            \end{enumerate} 
        \end{enumerate}
        \item[4)] Assume that $\mathcal{A}-\mathcal{B}-\mathcal{D}-\lambda\leq0$, $\max(H^{(2)}(\G_{2*}), H^{(2)}(\G_{2**}))<0$. Then, $H^{(1)}(G)\leq0$ on $\R_+$. Note that $\lim\limits_{G\rightarrow+\infty}H(G)=-\infty$.
        \begin{enumerate}
            \item[a)] If $\mathcal{A}-\mathcal{B}<0$, then no plausible savanna equilibria exist.
            \item[b)] If $\mathcal{A}-\mathcal{B}>0$, then there exists a unique $G_{*1}\in[0,+\infty)$ such that $H(G_{*1})=0$. Hence, there exists at most one savanna equilibrium $\textbf{E}_*=(G_{*1},T_{*1})'$ whenever $T_{*1}>0$.
        \end{enumerate}
\comment{
        there exists a unique $G_{*1}\in[0,+\infty)$ such that $H(G_{*1}=0$. Hence, there exists at most one savanna equilibrium $\textbf{E}_*=(G_{*1},T_{*1})$ whenever $G_{*1}>G^*$.
}
        \item[5)] Assume that $\mathcal{A}-\mathcal{B}-\mathcal{D}-\lambda\leq0$, $H^{(2)}(\G_{2*})<0$ and  $H^{(2)}(\G_{2**})>0$. Then there exist $\G_{10*}\in(\G_{2*}, \G_{2**})$ and $\G_{10**}\in(\G_{2**}, +\infty)$ such that $H^{(2)}(\G_{10*})=H^{(2)}(\G_{10**})=0$. One has two sub-cases.
        \begin{enumerate}
            \item[a)] Assume that $H^{(1)}(\G_{10**})\leq0$. Then $H^{(1)}(G)\leq0$ on $\R_+$. Note that $\lim\limits_{G\rightarrow+\infty}H(G)=-\infty$.
            \begin{enumerate}
                \item[i)] If $\mathcal{A}-\mathcal{B}<0$, then no plausible savanna equilibria exist.
                \item[ii)] If $\mathcal{A}-\mathcal{B}>0$, then there exists a unique $G_{*1}\in[0,+\infty)$ such that $H(G_{*1}=0$. Hence, there exists at most one savanna equilibrium $\textbf{E}_*=(G_{*1},T_{*1})'$ whenever $T_{*1}>0$.
            \end{enumerate}
\comment{
            Since $H(0)=(a-b)\alpha^2>0$ and $\lim\limits_{G\rightarrow+\infty}H(G)=-\infty$, then there exists a unique $G_{*1}\in[0,+\infty)$ such that $H(G_{*1}=0$. Hence, there exists at most one savanna equilibrium $\textbf{E}_*=(G_{*1},T_{*1})$ whenever $G_{*1}>G^*$.
}
            \item[b)] Assume that $H^{(1)}(\G_{10**})>0$ and $\mathcal{A}-\mathcal{B}>0$. Then, there exist $\G_{17*}\in(\G_{10*}, \G_{10**})$ and $\G_{17**}\in(\G_{10**},+\infty)$ that are zeros of $H^{(1)}$.
            \begin{enumerate}
                \item[i)] If $\min(H(\G_{17*}), H(\G_{17**}))>0$ then there exists at most one savanna equilibrium $\textbf{E}_*=(G_{*1},T_{*1})'$ whenever $G_{*1}\in(\G_{17**},+\infty)$, $H(G_{*1})=0$ and $T_{*1}>0$.
                \item[ii)] If $H(\G_{17*})<0$ and $H(\G_{17**})>0$, then there exist at most three savanna equilibria $\textbf{E}_*^i=(G_{*i},T_{*i})'$ whenever $G_{*1}\in(0,\G_{17*})$, $G_{*2}\in(\G_{17*},\G_{17**})$, $G_{*3}\in(\G_{17**},+\infty)$, $T_{*i}>0$ and $H(G_{*i})=0$, $i=1,2,3$.
                \item[iii)] If $\max(H(\G_{17*}), H(\G_{17**}))<0$ then there exists at most one savanna equilibrium $\textbf{E}_*=(G_{*1},T_{*1})'$ whenever $G_{*1}\in(0,\G_{17*})$, $H(G_{*1})=0$ and $T_{*1}>0$.
            \end{enumerate}
            \item[c)] Assume that $H^{(1)}(\G_{10**})>0$ and $\mathcal{A}-\mathcal{B}<0$. Then, there exist $\G_{17*}\in(\G_{10*}, \G_{10**})$ and $\G_{17**}\in(\G_{10**},+\infty)$ that are zeros of $H^{(1)}$.
            \begin{enumerate}
                \item[i)] If $H(\G_{17**})<0$ then no plausible savanna equilibria exist.
                \item[ii)] If $H(\G_{17**})>0$, then there exist at most two savanna equilibria $\textbf{E}_*^i=(G_{*i},T_{*i})'$ whenever  $G_{*1}\in(\G_{17*},\G_{17**})$, $G_{*2}\in(\G_{17**},+\infty)$, $H(G_{*i})=0$ and $T_{*i}>0$, $i=1,2$.
            \end{enumerate}
        \end{enumerate}
    \end{enumerate}
\end{enumerate}

This ends the case $\eta_{TG}(\textbf{W})<0$.\\

\noindent \textbf{Case 3:} $\eta_{TG}(\textbf{W})=0$.\\
From system (\ref{app_eq1}), one has
\begin{equation}\label{app_eq100}
\left\{
\begin{array}{l}
 G_* = G^*, \\
T^*-T_*-\displaystyle\frac{K_T(\textbf{W})}{g_T(\textbf{W})}f\vartheta(T_*)\omega(G^*)=0. 
\end{array}
\right.
\end{equation}
From system (\ref{app_eq100}) one deduces that a necessary condition for the existence of plausible savanna equilibria includes
$$\mathcal{R}^1_{\textbf{W}}>1, \quad \mathcal{R}^2_{\textbf{W}}>1, \quad T_*<T^*.$$
Let us set
$$\begin{array}{l}
 u =\displaystyle\frac{K_T(\textbf{W})}{g_T(\textbf{W})}f\omega(G^*)\lambda_{fT}^{min}, \\
v =\displaystyle\frac{K_T(\textbf{W})}{g_T(\textbf{W})}f\omega(G^*)(\lambda_{fT}^{max}-\lambda_{fT}^{min}),\\
J(T)=T^*-T-u-ve^{-pT}.
\end{array}$$
One has $J^{(1)}(T)=-1+pve^{-pT}$ and
$J^{(2)}(T)=-p^2ve^{-pT}<0$. Hence $J^{(1)}$ is decreasing on $\R_+$ and $\lim\limits_{T\rightarrow+\infty}J^{(1)}(T)=-1$. 
\begin{enumerate}
    \item[(I)] Assume that $J^{(1)}(0)=-1+pv>0$. Then there exists a unique $\bar{T}_{1*}\in\R_+$ such that $J^{(1)}(\bar{T}_{1*})=0$, that is $\bar{T}_{1*}=\ln{(pv)}/p$.
    \begin{enumerate}
        \item[1)] Assume that $J(\bar{T}_{1*})<0$. Then no plausible savanna equilibria exist.
        \item[2)] Assume that $J(\bar{T}_{1*})>0$ and $J(0)=T^*-u-v<0$. Then there exist at most two savanna equilibria $\textbf{E}_*^i=(G_{*},T_{*i})'$ whenever  $T_{*1}\in(0,\bar{T}_{1*})$, $T_{*2}\in(\bar{T}_{1*},+\infty)$, $J(T_{*i})=0$ and $T_{*i}<T^*$, $i=1,2$.
        \item[3)] Assume that $J(\bar{T}_{1*})>0$ and $J(0)=T^*-u-v>0$. Then there exist at most one savanna equilibrium $\textbf{E}_*=(G_{*},T_{*1})'$ whenever  $T_{*1}\in(\bar{T}_{1*},+\infty)$, $J(T_{*1})=0$ and $T_{*1}<T^*$.
    \end{enumerate}
    \item[(II)] Assume that $J^{(1)}(0)=-1+pv\leq0$. Then  $J$ is decreasing on $\R_+$. Note that $\lim\limits_{T\rightarrow+\infty}J(T)=-\infty$.
    \begin{enumerate}
        \item[1)] Assume that $J(0)=T^*-u-v<0$. Then no plausible savanna equilibria exist.
        \item[2)] Assume that $J(0)=T^*-u-v>0$. Then there exist at most one savanna equilibrium $\textbf{E}_*=(G_{*},T_{*1})'$ whenever  $T_{*1}\in(0,+\infty)$, $J(T_{*1})=0$ and $T_{*1}<T^*$.
    \end{enumerate}
\end{enumerate}

This ends the case $\eta_{TG}(\textbf{W})=0$ and the proof of the theorem.
\end{proof}

\section{Proof of Theorem \ref{al_thm2-bis} (Stability of non-hyperbolic equilibria)}\label{proff-al_thm2-bis}
In this section we give the proof of point (1) of Theorem \ref{al_thm2-bis}. Points (2) and (3) are done in the same way.
\begin{itemize}
    \item[(a)] Assume that $\mathcal{R}^{1}_{\textbf{W}}<1$ and $\mathcal{R}^{2}_{\textbf{W}}=1$. Hence, system (\ref{swv_eq1}) becomes
\begin{equation}
\left\{
\begin{array}{l}
\displaystyle \frac{dG}{dt}=-\displaystyle\frac{\gamma_{G}\textbf{W}}{b_{G}+\textbf{W}}\displaystyle\frac{G^2}{K_{G}(\textbf{W})}-\eta_{TG}(\textbf{W})TG,\\
\\
\displaystyle\frac{dT}{dt}=\displaystyle\frac{\gamma_{T}\textbf{W}}{b_{T}+\textbf{W}}T\left(1-\displaystyle\frac{T}{K_{T}(\textbf{W})}\right)-\delta_{T}T-f\vartheta(T)\omega(G)T,\\
\end{array}
\right.
\label{swv_eq1-reduit}
\end{equation}

and the Jacobian matrix of system (\ref{swv_eq1-reduit}) computed at $\textbf{E}_{0}=(0,0)'$ is 
    $$J_{\textbf{E}_{0}}=\left(
  \begin{array}{cc}
  0 & 0\\
  0 & \delta_T(\mathcal{R}^{1}_{\textbf{W}}-1
  \end{array}
\right).$$
Obviously, eigenvalues of $J_{\textbf{E}_{0}}$ are $\xi_1=0$ and $\xi_2=\delta_T(\mathcal{R}^{1}_{\textbf{W}}-1)<0$. An eigenvector corresponding to $\xi_1$ (resp. $\xi_2$) is $u_1=(1,0)'$ (resp. $u_2=(0,1)'$). Therefore, the linear stable manifold is $E^s=\{\alpha u_2,\quad \alpha \in \R\}$ and the linear center manifold is $E^c=\{\alpha u_1,\quad \alpha \in \R\}$. Since both $E^s$ and $E^c$ are invariant by system  (\ref{swv_eq1-reduit}), one deduces that the stable manifold is $W^s=E^s$ and the center manifold is $W^c=E^c$. On the center manifold, that is when $T=0$, we have from the first equation of system (\ref{swv_eq1-reduit}) that $\displaystyle \frac{dG}{dt}<0$. Hence, the non-hyperbolic equilibrium $\textbf{E}_{0}$ is locally stable in the positive orthant of $\R^2$; that is $\R^2_{+}$.

\item[(b)] Assume that $\mathcal{R}^{1}_{\textbf{W}}=1$ and $\mathcal{R}^{2}_{\textbf{W}}=1$. Hence, system (\ref{swv_eq1}) becomes
\begin{equation}
\left\{
\begin{array}{l}
\displaystyle \frac{dG}{dt}=-\displaystyle\frac{\gamma_{G}\textbf{W}}{b_{G}+\textbf{W}}\displaystyle\frac{G^2}{K_{G}(\textbf{W})}-\eta_{TG}(\textbf{W})TG,\\
\\
\displaystyle\frac{dT}{dt}=-\displaystyle\frac{\gamma_{T}\textbf{W}}{b_{T}+\textbf{W}}\displaystyle\frac{T^2}{K_{T}(\textbf{W})}-f\vartheta(T)\omega(G)T,\\
\end{array}
\right.
\label{swv_eq1-reduit-2}
\end{equation}

and the Jacobian matrix of system (\ref{swv_eq1-reduit-2}) computed at $\textbf{E}_{0}=(0,0)'$ is 
    $$J_{\textbf{E}_{0}}=\left(
  \begin{array}{cc}
  0 & 0 \\
  0 &0
  \end{array}
\right).$$
Obviously, eigenvalue of $J_{\textbf{E}_{0}}$ is $\xi=0$ which is double. Every non zero vectors of $\R^2$ is an eigenvector corresponding to $\xi$. Therefore, the center manifold is $W^c=\R^2$. On the set $\{T=0\}$, one has $\displaystyle \frac{dG}{dt}<0$ and on the set $\{G=0\}$, one has $\displaystyle \frac{dT}{dt}<0$. Hence, the non-hyperbolic equilibrium $\textbf{E}_{0}$ is locally stable in the positive orthant of $\R^2$; that is $\R^2_{+}$.

\item[(c)] The case where $\mathcal{R}^{1}_{\textbf{W}}=1$ and $\mathcal{R}^{2}_{\textbf{W}}<1$ is done like item $(a)$.
\end{itemize}

\section{Proof of Theorem \ref{al_thm3} (Stability of the savanna equilibrium)}
\label{al_AppendixB}

The Jacobian matrix at the savanna equilibrium
$\textbf{E}_{S}=(G_{*}, T_{*})'$ is given by

\begin{displaymath}
J_{*}=J(G_{*}, T_{*})=\left(
\begin{array}{ccc}
J_{*}^{11} & J_{*}^{12}\\
J_{*}^{21} & J_{*}^{22}\\
\end{array}
\right),
\end{displaymath}

where,

\begin{equation}
\left\{
\begin{array}{lcl}
J_{*}^{11}&=&g_{G}(\textbf{W})-(\delta_{G}+\lambda_{fG}f)-2\dfrac{g_{G}(\textbf{W})}{K_{G}(\textbf{W})}G_{*}-\eta_{TG}(\textbf{W})T_{*},\\
 &=&-\dfrac{g_{G}(\textbf{W})}{K_{G}(\textbf{W})}G_{*}.\\
J_{*}^{21}&=&-f\vartheta(T_{*})\omega^{'}(G_{*})T_{*}.\\
J_{*}^{12}&=&-\eta_{TG}(\textbf{W})G_{*}.\\
J_{*}^{22}&=&g_{T}(\textbf{W})-\delta_{T}-2\dfrac{g_{T}(\textbf{W})}{K_{T}(\textbf{W})}T_{*}-f\omega(G_{*})[\vartheta(T_{*})+T_{*}\vartheta^{'}(T_{*})],\\
 &=&-\dfrac{g_{T}(\textbf{W})}{K_{T}(\textbf{W})}T_{*}-f\omega(G_{*})T_{*}\vartheta^{'}(T_{*}).
\end{array}
\right.
\label{app_eq13}
\end{equation}
Recall that $$\vartheta^{'}(T_{*})<0.$$ The characteristic
equation of $J_{*}$ is

\begin{equation}
\mu^{2}-tr(J_{*})\mu+det(J_{*})=0,
\label{app_eq14}
\end{equation}
where, $tr(J_{*})=J_{*}^{11}+J_{*}^{22}$ and
$det(J_{*})=J_{*}^{11}J_{*}^{22}-J_{*}^{21}J_{*}^{12}$. It follows
that all eigenvalues of  the characteristic equation have negative
real part if and only if $tr(J_{*})<0$ and $det(J_{*})>0$.

\par
We have
\begin{equation}\label{trace-jacobien}
\begin{array}{cl}
tr(J_{*})&=J_{*}^{11}+J_{*}^{22}\\
&=-\left(\dfrac{g_{G}(\textbf{W})}{K_{G}(\textbf{W})}G_{*}+\dfrac{g_{T}(\textbf{W})}{K_{T}(\textbf{W})}T_{*}\right)-f\omega(G_{*})T_{*}\vartheta^{'}(T_{*})\\
&=\left(\dfrac{g_{G}(\textbf{W})}{K_{G}(\textbf{W})}G_{*}+\dfrac{g_{T}(\textbf{W})}{K_{T}(\textbf{W})}T_{*}\right)(\mathcal{R}^{1}_{*}-1),
\end{array}
\end{equation}

where,

$$\mathcal{R}^{1}_{*}=\dfrac{-f\omega(G_{*})T_{*}\vartheta^{'}(T_{*})}
{\left(\dfrac{g_{G}(\textbf{W})}{K_{G}(\textbf{W})}G_{*}+\dfrac{g_{T}(\textbf{W})}{K_{T}(\textbf{W})}T_{*}\right)}.$$

When $\eta_{TG}(\textbf{W})>0$, we have:

\begin{equation}\label{determinant-jacobien}
\begin{array}{cl}
det(J_{*})&=J_{*}^{11}J_{*}^{22}-J_{*}^{21}J_{*}^{12}\\
&=\dfrac{g_{G}(\textbf{W})}{K_{G}(\textbf{W})}G_{*}\left(\dfrac{g_{T}(\textbf{W})}{K_{T}(\textbf{W})}T_{*}+f\omega(G_*)\vartheta'(T_*)T_*\right)-f\eta_{TG}(\textbf{W})T_*G_*\vartheta(T_*)\omega'(G_*),\\
&=T_*G_*\left[\dfrac{g_{G}(\textbf{W})g_{T}(\textbf{W})}{K_{G}(\textbf{W})K_{T}(\textbf{W})}+f\dfrac{g_{G}(\textbf{W})}{K_{G}(\textbf{W})}\omega(G_*)\vartheta'(T_*)-f\eta_{TG}(\textbf{W})\vartheta(T_*)\omega'(G_*)\right],\\
&=\eta_{TG}(\textbf{W})T_*G_*\left[\dfrac{g_{G}(\textbf{W})g_{T}(\textbf{W})}{\eta_{TG}(\textbf{W})K_{G}(\textbf{W})K_{T}(\textbf{W})}-\left(\left.\dfrac{d}{dG}(f\vartheta(T(G))\omega(G))\right|_{G=G_*}\right)\right],\\
&=T_*G_*\left[-f\dfrac{g_{G}(\textbf{W})}{K_{G}(\textbf{W})}\omega(G_*)\vartheta'(T_*)+f\eta_{TG}(\textbf{W})\vartheta(T_*)\omega'(G_*)\right](\mathcal{R}^2_*-1),
\end{array}
\end{equation}
where
$$\mathcal{R}^{2}_{*}=\dfrac{\dfrac{g_{G}(\textbf{W})g_{T}(\textbf{W})}{K_{G}(\textbf{W})K_{T}(\textbf{W})}}
{\left(-f\dfrac{g_{G}(\textbf{W})}{K_{G}(\textbf{W})}\omega(G_*)\vartheta'(T_*)+f\eta_{TG}(\textbf{W})\vartheta(T_*)\omega'(G_*)
\right)}.$$

 Recall that, the expression of $T(G)$ is given by
(\ref{app_eq2}), page \pageref{app_eq2}. Based on the chain rule, we
prove that
$$\dfrac{d}{dG}\vartheta(T(G))=\dfrac{d\vartheta(T)}{dT}\dfrac{dT}{dG}=-\vartheta'(T)\dfrac{g_{G}(\textbf{W})}{\eta_{TG}(\textbf{W})K_{G}(\textbf{W})}.$$

When $\eta_{TG}(\textbf{W})<0$, we have:

\begin{equation}\label{determinant-jacobien-enta-negatif}
\begin{array}{cl}
det(J_{*})&=J_{*}^{11}J_{*}^{22}-J_{*}^{21}J_{*}^{12}\\
&=T_*G_*\left[\dfrac{g_{G}(\textbf{W})g_{T}(\textbf{W})}{K_{G}(\textbf{W})K_{T}(\textbf{W})}+f\dfrac{g_{G}(\textbf{W})}{K_{G}(\textbf{W})}\omega(G_*)\vartheta'(T_*)-f\eta_{TG}(\textbf{W})\vartheta(T_*)\omega'(G_*)\right],\\
&=T_*G_*\left[-f\dfrac{g_{G}(\textbf{W})}{K_{G}(\textbf{W})}\omega(G_*)\vartheta'(T_*)\right](\mathcal{Q}^2_*-1),
\end{array}
\end{equation}
where
$$\mathcal{Q}^{2}_{*}=\dfrac{\dfrac{g_{G}(\textbf{W})g_{T}(\textbf{W})}{K_{G}(\textbf{W})K_{T}(\textbf{W})}-f\eta_{TG}(\textbf{W})\vartheta(T_*)\omega'(G_*)}
{-f\dfrac{g_{G}(\textbf{W})}{K_{G}(\textbf{W})}\omega(G_*)\vartheta'(T_*)}.$$

Thus, in the case $\eta_{TG}(\textbf{W})>0$, the savanna
equilibrium $\textbf{E}_S=(G_{*}, T_{*})'$ is locally asymptotically
stable whenever $\mathcal{R}^{1}_{*}<1$ and $\mathcal{R}^{2}_{*}>1$.
Similarly,  in the case $\eta_{TG}(\textbf{W})<0$, the savanna
equilibrium $\textbf{E}_S=(G_{*}, T_{*})'$ is locally asymptotically
stable whenever $\mathcal{R}^{1}_{*}<1$ and $\mathcal{Q}^{2}_{*}>1$.
This ends the proof of Theorem \ref{al_thm3}.

\section{Proof of Theorem \ref{Hopf-Lyapunov} (Lyapunov Number)}
\label{AppendixD}

Introducing perturbations
$$x=G-G_*  \quad \mbox{and}\quad y=T-T_*$$
in system
(\ref{swv_eq1}) and then expanding in Taylor series, we have
\begin{equation}\label{New-system-Hopf}
    \begin{array}{l}
      \dfrac{dx}{dt}=a_{10}x+a_{01}y+a_{20}x^2+a_{11}xy+a_{02}y^2+a_{30}x^3+a_{21}x^2y+a_{12}xy^2+a_{03}y^3+\cdot\cdot\cdot, \\
      \dfrac{dy}{dt}=b_{10}x+b_{01}y+b_{20}x^2+b_{11}xy+b_{02}y^2+b_{30}x^3+b_{21}x^2y+b_{12}xy^2+b_{03}y^3+\cdot\cdot\cdot,
    \end{array}
\end{equation}
where $a_{10}=J^{11}_*:=a$, $a_{01}=J^{12}_*:=b$,
$b_{10}=J^{21}_*:=c$ and $b_{01}=J^{22}_*:=d$ are the elements of
the Jacobian matrix evaluated at the savanna equilibrium
$\textbf{E}_{S}=(G_{*}, T_{*})'$ with $f=f_h$ (see equation
(\ref{app_eq13}), page \pageref{app_eq13}). Hence, together with
(\ref{condition-trace}), (\ref{condition-determinant}), we have
$$a_{10}+b_{01}=0 \quad \mbox{and}\quad
\Delta=a_{10}b_{01}-a_{01}b_{10}>0.$$ Let $(F_1,F_2)'$ denotes the
right hand side of system (\ref{swv_eq1}). The expressions of the coefficients
$a_{ij}$ and $b_{ij}$ with $i,j\in\{1,2,3\}$ are given below:
\begin{equation}\label{coeficient}
\begin{array}{l}
a_{20}=\left.\dfrac{1}{2}\dfrac{\partial^2F_1}{\partial
G^2}\right|_{(\textbf{E}_{S},f=f_h)}=-\dfrac{g_{G}(\textbf{W})}{K_{G}(\textbf{W})},
a_{02}=\left.\dfrac{1}{2}\dfrac{\partial^2F_1}{\partial
T^2}\right|_{(\textbf{E}_{S},f=f_h)}=0,
a_{11}=\left.\dfrac{\partial^2F_1}{\partial G\partial T}\right|_{(\textbf{E}_{S},f=f_h)}=-\eta_{TG}(\textbf{W}), \\
a_{12}=\left.\dfrac{1}{2}\dfrac{\partial^3F_1}{\partial G\partial
T^2}\right|_{(\textbf{E}_{S},f=f_h)}=0,
a_{21}=\left.\dfrac{1}{2}\dfrac{\partial^3F_1}{\partial G^2\partial
T}\right|_{(\textbf{E}_{S},f=f_h)}=0,\\
a_{30}=\left.\dfrac{1}{6}\dfrac{\partial^3F_1}{\partial
G^3}\right|_{(\textbf{E}_{S},f=f_h)}=0,
a_{03}=\left.\dfrac{1}{6}\dfrac{\partial^3F_1}{\partial T^3}\right|_{(\textbf{E}_{S},f=f_h)}=0, \\
b_{20}=\left.\dfrac{1}{2}\dfrac{\partial^2F_2}{\partial
G^2}\right|_{(\textbf{E}_{S},f=f_h)},
b_{02}=\left.\dfrac{1}{2}\dfrac{\partial^2F_2}{\partial
T^2}\right|_{(\textbf{E}_{S},f=f_h)},
b_{11}=\left.\dfrac{\partial^2F_2}{\partial G\partial T}\right|_{(\textbf{E}_{S},f=f_h)}, \\
b_{12}=\left.\dfrac{1}{2}\dfrac{\partial^3F_2}{\partial G\partial
T^2}\right|_{(\textbf{E}_{S},f=f_h)},
b_{21}=\left.\dfrac{1}{2}\dfrac{\partial^3F_2}{\partial G^2\partial
T}\right|_{(\textbf{E}_{S},f=f_h)},
b_{30}=\left.\dfrac{1}{6}\dfrac{\partial^3F_2}{\partial
G^3}\right|_{(\textbf{E}_{S},f=f_h)},
b_{03}=\left.\dfrac{1}{6}\dfrac{\partial^3F_2}{\partial T^3}\right|_{(\textbf{E}_{S},f=f_h)}. \\
\end{array}
\end{equation}
The value of the first Lyapunov number, which helps to determine the
nature of the stability of limit cycle arising through Hopf
bifurcation is given by (\citet[page 253]{Andronov1971},
\citet[page 353]{Perko2001})
\begin{equation}\label{sigma -Lyapounov}
\begin{array}{ccl}
\sigma &=&
-\dfrac{3\pi}{2b\Delta^{3/2}}\left\{[ac(a^2_{11}+a_{11}b_{02}+a_{02}b_{11})+ab(b^2_{11}+a_{20}b_{11}+a_{11}b_{02})\right.\\
&&+c^2(a_{11}a_{02}+2a_{02}b_{02})-2ac(b^2_{02}-a_{20}a_{02})-2ab(a^2_{20}-b_{20}b_{02})\\
&&-b^2(2a_{20}b_{20}+b_{11}b_{20})+(bc-2a^2)(b_{11}b_{02}-a_{11}a_{20})]\\
&&\left.-(a^2+bc)[3(cb_{03}-ba_{30})+2a(a_{21}+b_{12})+(ca_{12}-bb_{21})]\right\}\\
&=&
-\dfrac{3\pi}{2b\Delta^{3/2}}\left\{[ac(a^2_{11}+a_{11}b_{02})+ab(b^2_{11}+a_{20}b_{11}+a_{11}b_{02})\right.\\
&&-2acb^2_{02}-2ab(a^2_{20}-b_{20}b_{02})-b^2(2a_{20}b_{20}+b_{11}b_{20})+(bc-2a^2)(b_{11}b_{02}-a_{11}a_{20})]\\
&&\left.-(a^2+bc)[3cb_{03}+2ab_{12}-bb_{21}]\right\}.\\
\end{array}
\end{equation}
Hence, conclusions of Theorem \ref{Hopf-Lyapunov} follow from \citet[Theorem 1, page 352]{Perko2001}.
\end{document}